\documentclass[preprint]{imsart}

\RequirePackage[OT1]{fontenc}
\RequirePackage{amsthm,amsmath,natbib}

\usepackage[utf8]{inputenc}
\usepackage{microtype}
\usepackage{graphicx}
\usepackage{amssymb}
\usepackage{amsfonts}
\usepackage{epstopdf}
\usepackage{hyperref}
\usepackage{enumerate}
\usepackage{float}
\usepackage{lmodern}
\usepackage{array, multirow}
\usepackage{thmtools}
\usepackage{dsfont}

\newtheorem{prop}{Proposition}
\newtheorem{thm}{Theorem}

\newtheorem{cor}{Corollary}
\theoremstyle{remark}
\newtheorem{rmk}{Remark}
\newtheorem{example}{Example}
\newtheorem{assumption}{Assumption}

\usepackage[linesnumbered,ruled,vlined]{algorithm2e}
\graphicspath{{./art/}}

\providecommand{\norm}[1]{\lVert{#1}\rVert}

\providecommand{\abs}[1]{\left\lvert{#1}\right\rvert}
\providecommand{\reals}{\mathbb{R}}

\providecommand{\Rd}{\reals^d}
\providecommand{\eps}{\varepsilon}
\providecommand{\diff}{\mathrm{d}}
\providecommand{\oh}{\operatorname{o}}
\providecommand{\Oh}{\operatorname{O}}
\providecommand{\expec}{\operatorname{\mathbb{E}}}

\providecommand{\1}{\mathds{1}}

\providecommand{\tr}{\operatorname{tr}}
\providecommand{\normal}{\mathcal{N}}

\providecommand{\model}{\mathcal{M}}
\providecommand{\Prob}{\mathcal{P}} % collection of probability measures
\providecommand{\cmp}{\mathcal{K}} % collection of compacta
\providecommand{\elliptical}{\mathcal{E}}

\providecommand{\prob}{\mathrm{P}}  % a typical probability measure
\providecommand{\probQ}{\mathrm{Q}} % another one
\providecommand{\probR}{\mathrm{R}} % another one
\providecommand{\probU}{\mathrm{U}} % uniform distribution on some ball
\providecommand{\emprob}{\widehat{\prob}_n} % an empirical probability measure

\providecommand{\cPna}{c_{n}(\alpha, \prob_0)}
\providecommand{\cMna}{c_{\model,n}(\alpha)}%{c_{n}(\alpha, \model)}
\providecommand{\cMnta}{c_{n,\theta}(\alpha)}%{c_{\model,n,\theta}(\alpha)}
\providecommand{\cMntha}{c_{n,\hat{\theta}_n}(\alpha)}%{c_{\model,n,\hat{\theta}_n}(\alpha)}

\usepackage{color}
\definecolor{darkraspberry}{rgb}{0.53, 0.15, 0.34}
\definecolor{britishracinggreen}{rgb}{0.0, 0.26, 0.15}
\definecolor{burntumber}{rgb}{0.54, 0.2, 0.14}

\newcommand{\js}[1]{\textcolor{red}{\sffamily\small [JS: {#1}]}}
\newcommand{\gm}[1]{\textcolor{britishracinggreen}{\sffamily\small [GM: {#1}]}}

\SetKwInput{KwInput}{Input}                % Set the Input
\SetKwInput{KwOutput}{Output}
\endlocaldefs

\begin{document}
	
\begin{frontmatter}	
\title{Multivariate Goodness-of-Fit Tests Based on Wasserstein Distance}
\runtitle{Wasserstein Goodness-of-Fit Tests}
	
\begin{aug}
\author{\fnms{Marc} \snm{Hallin}}

\address{ECARES and D\'epartement de Math\'ematique, Universit\'e libre de Bruxelles\\
Avenue F.D.\ Roosevelt 50, 1050 Brussels, Belgium\\
e-mail: \href{mailto:mhallin@ulb.ac.be}{\tt mhallin@ulb.ac.be}}

\author{\fnms{Gilles} \snm{Mordant}\corref{}}
\and
\author{\fnms{Johan} \snm{Segers}}

\address{LIDAM/ISBA, UCLouvain\\
Voie du Roman Pays 20/L1.04.01, B-1348 Louvain-la-Neuve, Belgium\\
e-mail: \href{mailto:gilles.mordant@uclouvain.be}{\tt gilles.mordant@uclouvain.be}; \href{mailto:johan.segers@uclouvain.be}{\tt johan.segers@uclouvain.be}}

\runauthor{M. Hallin, G. Mordant and J. Segers}

\affiliation{Universit\'e libre de Bruxelles and UCLouvain}
\end{aug}
	
\begin{abstract}
%Theoretical and computational obstacles have hindered the use of the Wasserstein distance for goodness-of-fit tests for simple and composite null hypotheses for general multivariate distributions. 
Goodness-of-fit tests based on the empirical Wasserstein distance are proposed for simple and composite null hypotheses involving general multivariate distributions. 
For group families, the procedure is to be implemented after preliminary reduction of the data via invariance.
%a reduction of the data makes the test statistic distribution-free with respect to the unknown parameter. 
This property allows for calculation of exact critical values and $p$-values at finite sample sizes. 
Applications include testing for location--scale families and testing for families arising from affine transformations, such as elliptical distributions with given standard radial density and unspecified location vector and scatter matrix. 
A novel test for multivariate normality with unspecified mean vector and covariance matrix arises as a special case.
For more general parametric families, we propose a parametric bootstrap procedure to calculate critical values.
The lack of asymptotic distribution theory for the empirical Wasserstein distance means that the validity of the parametric bootstrap under the null hypothesis remains a conjecture. 
Nevertheless, we show that the test is consistent against fixed alternatives.
To this end, we prove a uniform law of large numbers for the empirical distribution in Wasserstein distance, where the uniformity is over any class of underlying distributions satisfying a uniform integrability condition but no additional moment assumptions.
The calculation of test statistics boils down to solving the well-studied semi-discrete optimal transport problem. 
Extensive numerical experiments demonstrate the practical feasibility and the excellent performance of the proposed tests for the Wasserstein distance of order $p = 1$ and $p = 2$ and for dimensions at least up to $d = 5$. The simulations also lend support to the conjecture of the asymptotic validity of the parametric bootstrap.
%
%
%
%
%explore the potential of goodness-of-fit tests based on the empirical Wasserstein distance for simple and composite null hypotheses for general multivariate distributions. Besides testing for a specific distribution, we propose a Wasserstein test for multivariate normality w hen the parameters are unknown. The test arises as a special case of a goodness-of-fit test for a general location--scatter family of elliptical distributions with known radial distribution. The presence of unknown nuisance parameters is handled by using empirically standardized data. For general parametric families, we propose to compute critical values by a parametric bootstrap. The practical calculation of the test statistics boils down to solving the well-studied semi-discrete optimal transport problem. The proof that the Wasserstein goodness-of-fit test detects alternatives with a probability tending to one relies on the convergence to zero of the empirical Wasserstein distance uniformly over certain families of distributions, a novel result of potential independent interest. A simulation study establishes the practical feasibility and the excellent performance of the proposed tests.
\end{abstract}

\tableofcontents
\end{frontmatter}

%-------------------------------------------------------
\section{Introduction}

Wasserstein distances are metrics on spaces of probability measures with certain finite moments. They measure the distance between two such distributions by the minimal cost of moving probability mass in order to transform one distribution into the other. Wasserstein distances have a long history and continue to attract interest from diverse fields in statistics, machine learning, and computer science, in particular image analysis; see for instance the monographs and reviews by \citet{santambrogio2015optimal}, \citet{peyre+c:2019}, and \citet{panaretos2019statistical}.

A natural application of any meaningful distance between distributions is to the goodness-of-fit (GoF) problem---namely, the problem of testing the null hypothesis that a sample comes from a population with fully specified distribution~$\prob_0$ or with unspecified distribution within some postulated parametric model~$\model$. GoF problems certainly are among the most fundamental and classical ones in statistical inference. Typically, GoF tests are based on some distance between the empirical distribution~$\emprob$ and the null distribution~$\prob_0$ or an estimated distribution in the null model~$\model$.
The most popular ones are the Cram\'{e}r--von Mises \citep{cramer1928, vMises1928} and Kolmogorov--Smirnov \citep{kolmogorov1933, smirnov1939} tests, involving distances between the cumulative distribution function of $\prob_0$ and the empirical one.  Originally defined for univariate distributions only, they have been extended to the multivariate case, for instance in \citet{khmaladze2016unitary}, who proposes a test that has nearly all properties one could wish for, including asymptotic distribution-freeness, but whose implementation is computationally heavy and quickly gets intractable. 

Many other distances have been considered in this context, though. Among them, distances between densities (after kernel smoothing) have attracted much interest, starting with \citet{bickel1973some} in the univariate case. \citet{bakvsajev2015multivariate} recently proposed a multivariate extension; the choice of a bandwidth matrix, however, dramatically affects the outcome of the resulting testing procedure. %;  see \citet{gonzalez2013updated} for a survey in the regression context.
\citet{fan1997goodness} considers a distance between characteristic functions, which accommodates arbitrary dimensions; the idea is appealing but the estimation of the integrals involved in the distance seems tricky. \citet{mcassey2013empirical} proposes a heuristic test that relies on a comparison of the empirical Mahalanobis distance with a simulated one under the null. Still in a multivariate setting, \citet{ebner2018multivariate} define a distance based on sums of powers of weighted volumes of $k$th nearest neighbour spheres. 

The use of the Wasserstein distance for GoF testing has been considered mostly for univariate distributions \citep{munk1998nonparametric, delbarrio1999, del2000contributions, delbarrio2005}. For the multivariate case, available methods are restricted to discrete distributions \citep{sommerfeld2018inference} and Gaussian ones \citep{rippl2016limit}. Indeed, serious difficulties,  both computational and theoretical,  hinder the development of Wasserstein GoF tests for general multivariate continuous distributions, particularly in the case of  composite null hypotheses. 
Such hypotheses
%, under which the unknown distribution of the observations is assumed to belong to some parametric family $\model$ of distributions,  
%certainly no less important and
are generally more realistic than simple ones. 
%   under which that distribution is assumed to coincide with some fully specified $\prob_0$.
    Of particular practical importance is the case of location--scale and location--scatter families: tests of multivariate Gaussianity, tests of elliptical symmetry with given standard radial density, etc., belong to that type. Although the asymptotic null distribution of empirical processes with estimated parameters is well known \citep[Theorem~19.23]{vdv98}, the actual exploitation of that theory in GoF testing remains problematic because of the difficulty of computing critical values.

The aim of this paper is to explore the potential of the Wasserstein distance for GoF tests of simple (consisting of one fully specified distribution)  and composite  (consisting of a parametric family of distributions) null hypotheses involving continuous multivariate distributions. 
%For multivariate distributions, however, no such simple formulas exist, not for the distance nor for the optimal transport plan. This explains why research on the computation of Wasserstein distances and the distributional properties of empirical Wasserstein distances is still ongoing. Much efforts are focused on the discrete case. Also for the multivariate Gaussian case, explicit formulas are available.
%The present paper aims to investigate the potential of the Wasserstein distances for goodness-of-fit testing. 
%We treat the GoF problem both for a specific distribution and a parametric model. 
The tests we are  proposing are based on the Wasserstein distance between the empirical distribution of the data or estimated residuals on the one hand and a model-based estimate thereof on the other hand.
%They are computationally feasible, have the correct size, and enjoy good power properties in comparison with other tests available in the literature. 
%\js{No longer true: for group families, we use estimated `residuals'. The claim that the test has the correct size for general parametric families is too strong, it's only a conjecture, we should be careful.}
%
We concentrate on the continuous case, i.e., the distributions under the null hypothesis are absolutely continuous with respect to the $d$-dimensional Lebesgue measure.  
The test statistic involves the Wasserstein distance between a discrete empirical distribution and a continuous distribution specified by the null hypothesis.
Calculating such a distance requires solving the semi-discrete transportation problem, an active area of research in computer science.% We briefly review the computational aspects in Section~\ref{sec: computations}. 

In case of a simple null hypothesis, the null distribution of the test statistic does not depend on unknown parameters. Exact critical values can be calculated with arbitrary precision via a Monte Carlo procedure, by simulating from the null distribution and computing empirical quantiles.

Exact critical values can also be computed for Wasserstein tests for the GoF of a group family, that is, a model that arises by applying a transformation group to some specified distribution \citep[pp.~16--23]{lehmann:1998}.
If the parameter estimate is equivariant, the data can be reduced in such a way that their distribution no longer depends on the unknown parameter.
The Wasserstein distance between this parameter-free distribution and the empirical distribution of the reduced data then provides a test statistic whose null distribution does not depend on the unknown parameter either.
Important special cases include elliptical distributions with known radial distribution and unknown location vector and scatter matrix.
In particular, our approach yields a novel test for multivariate normality with unknown mean vector and covariance matrix. 

For general parametric models, the test statistic measures the Wasserstein distance between the empirical distribution and the model-based one with estimated parameter.
A reduction via invariance is no longer possible and we rely on the parametric bootstrap to calculate critical values.
Still, some parameters, such as location-scale parameters, can be factored out, again by relying on transformation groups.
The question whether the parametric bootstrap has the correct size under the null hypothesis remains open.
A proof of that property would require asymptotic distribution theory for the empirical Wasserstein distance---a hard and long-standing open problem, which we  briefly review in Section~\ref{sec:asy}, the solution of which is beyond the scope of this paper. 
Monte Carlo experiments, however, support our conjecture that the parametric bootstrap has the correct size at least asymptotically.

% Moreover, we can prove that the test is consistent against fixed alternatives. The proof relies on the convergence in probability to zero of the empirical Wasserstein distance uniformly in the underlying distribution, when the latter ranges over families of distribution satisfying a uniform integrability criterion. This result is new and of potential independent interest.

In all cases, even in the general parametric case, we show that our Wasserstein GoF tests are consistent against fixed alternatives, that is, the null hypothesis under such alternatives is rejected with probability tending to one as the sample size tends to infinity.
For the general parametric case, the proof relies on the uniform consistency in probability of the empirical distribution with respect to the Wasserstein distance, uniformly over families of distributions that satisfy a uniform integrability condition.
To the best of our knowledge, this result is new. 

We conduct an extensive simulation study to assess the finite-sample performance of the Wasserstein tests of order $p \in \{1, 2\}$ in comparison to other GoF tests. The set-up involves both simple and composite null hypotheses as well as a wide variety of alternatives.
The experiments lend support to the conjecture that the parametric bootstrap is valid asymptotically.
In comparison to other GoF tests available in the literature, the Wasserstein test demonstrates good power.
This is especially true for the test of multivariate normality, where, out of the many available tests in the literature, we select the ones of \citet{royston1983}, \citet{henze1990class} and \citet{rizzo2016energy} as benchmarks.

In a recent strand of literature, measure transportation serves to link a multivariate probability measure to a standard reference distribution, yielding novel concepts of multivariate ranks, signs, and quantiles \citep{carlier2016vector, chernozhukov2017monge, del2018center}.
Here we do not make this step, as the Wasserstein distances we are considering are between distributions defined on the sample space.
%These notions have been applied by \citet{shi2019distributionfree}, \citet{deb2019multivariate}, and \citet{ghosal2019multivariate} in the construction of distribution-free tests in a multivariate context, and by \citet{hallin2019center} for R-estimation of VARMA models with unspecified innovation densities. 

%\paragraph{Outline of the paper.} 
The outline of the paper is as follows. 
In the remainder of this introduction, we introduce the Wasserstein distance (Section~\ref{sec:Wasserstein}), review the asymptotic theory of empirical Wasserstein distance (Section~\ref{sec:asy}), and provide some information on the computational methods for the semi-discrete transportation problem underlying the implementation of the Wasserstein GoF tests (Section~\ref{sec: computations}). 
In Section~\ref{sec: Test}, we give a formal description of the GoF test procedure for simple null hypotheses.
Section~\ref{sec:group} addresses the composite null hypothesis that the unknown distribution belongs to some group family.
%an elliptical family with unknown location and scatter (covariance) parameters and known radial distribution; the multivariate normal family is an important special case. 
Composite null hypotheses covering general parametric models are treated in Section~\ref{sec:param}. In Section~\ref{sec:param:group} we mention a hybrid approach, where some components of the parameter vector are factored out by relying on a transformation group.
In Section~\ref{sec:simu}, finally, we report on the results of our numerical experiments. In Appendix~\ref{app:empWassUnif}, the convergence of the empirical Wasserstein distance uniformly over certain classes of underlying distributions is stated and proved. 
Appendix~\ref{sec:boot:cons} is related to the consistency of the parametric bootstrap.
The other appendices contain further details on the simulation study. 
%\bigskip

\subsection{Wasserstein distance}
\label{sec:Wasserstein}

Let $\Prob(\Rd)$ be the set of Borel probability measures on $\Rd$ and let $\Prob_p(\Rd)$ be the subset of such measures with a finite moment of order $p \in [1, \infty)$. For~$\prob, \probQ \in \Prob(\Rd)$, let $\Gamma(\prob,\probQ)$ be the set of probability measures $\gamma$ on $\Rd\times \Rd$ with marginals $\prob$ and $\probQ$, i.e., such that $\gamma(B \times \Rd) = \prob(B)$ and $\gamma(\Rd \times B) = \probQ(B)$ for Borel sets $B \subseteq \Rd$. The $p$-Wasserstein distance between $\prob, \probQ \in \Prob_p(\Rd)$ is
\[
	W_p(\prob,\probQ) := 
	\left(
		\inf_{\gamma\in\Gamma(\prob,\probQ)} 
		\int_{\Rd\times \Rd} \norm{x-y}^p \, \diff\gamma(x,y)
	\right)^{1/p},
\]
with $\norm{\,\cdot\,}$ the Euclidean norm. In terms of random variables $X$ and $Y$ with laws $\prob$ and $\probQ$, respectively, the $p$-Wasserstein distance is the smallest value of $\{\expec(\norm{X-Y}^p)\}^{1/p}$ over all possible joint distributions $\gamma \in \Gamma(\prob, \probQ)$ of $(X, Y)$.

The $p$-Wasserstein distance $W_p$ defines a metric on $\Prob_p(\Rd)$, which thereby becomes a complete separable metric space \citep[Theorem~6.18 and the bibliographical notes]{villani2008optimal}.
Convergence in the $W_p$ metric is equivalent to weak convergence plus convergence of moments of order $p$; see for instance \citet[Lemmas~8.1 and~8.3]{bickel+f:1981} and  \citet[Theorem~6.9]{villani2008optimal}.
%It is thus quite natural to consider $W_p$ in the construction of GoF tests for multivariate distributions. 

For univariate distributions $\prob$ and $\probQ$ with distribution functions $F$ and $G$, the $p$-Wasserstein distance boils down to the~$L^p$-distance 
\begin{equation}
\label{eq:Wp:d1}
W_p(\prob, \probQ) = \left(\int_0^1 \big\vert F^{-1}(u) - G^{-1}(u)\big\vert ^p \, \diff u\right)^{1/p}
\end{equation}
 between the respective quantile functions $F^{-1}$ and $G^{-1}$. This representation considerably facilitates both the computation of the distance and the asymptotic theory of its empirical versions. Also, the optimal transport plan mapping $X \sim \prob$ to $Y \sim \probQ$ is immediate: if $F$ has no atoms, then $Y:=G^{-1} \circ F(X)\sim \probQ$, while monotonicity of $G^{-1} \circ F$ implies the optimality of the coupling $(X,Y)$, see for instance \citet[Section~1.2.3]{panaretos2019statistical}.%$(X, G^{-1} \circ F(X))$.

\subsection{Asymptotic theory: results and an open problem}
\label{sec:asy}

To construct critical values for Wasserstein GoF tests of general parametric models, we will propose in Section~\ref{sec:param} the use of the parametric bootstrap. 
In general, proving consistency of the parametric bootstrap requires having, under contiguous alternatives, non-degenerate limit distributions of the statistic of interest \citep{beran1997, capanu:2019}. 
For Wasserstein distances involving empirical distributions, such results are still far beyond the horizon, as the following short survey will show.

Let $X_1, \ldots, X_n$ be an i.i.d.\ (independent and identically distributed) sample from $\prob \in \Prob(\Rd)$. 
%Its distribution as a random vector in $(\Rd)^n$ is the $n$-fold product~$\prob^n$ of $\prob$ with itself. 
%Let $L_n : (\Rd)^n \to \Prob(\reals^d)$ %denote the
% map %that sends a  sample vector
%   $\mathbf{x}_n = (x_1, \ldots, x_n)\in(\Rd)^n$ to the discrete probability measure $L_n(\mathbf{x}_n) := n^{-1} \sum_{i=1}^n \delta_{x_i}$, with $\delta_{x}$ the Dirac measure at $x$. 
The empirical distribution of the sample is $\emprob := %L_n(\mathbf{X}_n) =
 n^{-1} \sum_{i=1}^n \delta_{X_i}$, with $\delta_{x}$ the Dirac measure at $x$.
% We study its distribution as a random element in $\Prob(\Rd)$.
Assuming that $\prob$ has a finite moment of order $p \in [1, \infty)$, we are interested in the empirical Wasserstein distance~$W_p(\emprob, \prob)$.
%The Wasserstein distance between the empirical distribution $L_n(\mathbf{x}_n)$  and a probability measure $\prob \in \Prob_p(\Rd)$ is the value  at $\mathbf{x}_n \in (\Rd)^n$  of the map 
%\[
%W_p(L_n, \prob) : \mathbf{x}_n\in (\Rd)^n \mapsto W_p(L_n(\mathbf{x}_n), \prob)\in [0, \infty).
%\] 
%   Consider the distribution of this map under $\prob^n$, i.e., for an independent random sample of size $n$ from $\prob$. In perhaps more familiar notation, the random variable of interest is the empirical Wasserstein distance~$W_p(\emprob, \prob)$.

%The Wasserstein distance between the empirical distribution of an observed sample~$\mathbf{x}_n \in (\Rd)^n$ and a probability measure $\prob \in \Prob_p(\Rd)$ is $W_p(L_n(\mathbf{x}_n), \prob)$. We see this number as the value of the map $W_p(L_n, \prob) : (\Rd)^n \to [0, \infty)$ at $\mathbf{x}_n$. Consider the distribution of this map under $\prob^n$, i.e., for an independent random sample of size $n$ from $\prob$. In perhaps more familiar notation, the random variable of interest thus is the empirical Wasserstein distance~$W_p(\emprob, \prob)$.

According to \citet[Lemma 8.4]{bickel+f:1981}, the empirical distribution is strongly consistent in the Wasserstein distance: for an i.i.d.\ sequence~$X_1, X_2, \ldots$ with common distribution $\prob$, we have $W_p(\emprob, \prob) \to 0$ almost surely as~$n \to \infty$. 
Bounds and rates for the expectation of the empirical Wasserstein distance have been studied intensively; see \citet[Section~3.3]{panaretos2019statistical} for a review. 
If $\prob$ is non-degenerate, then~$\expec[W_p(\emprob, \prob)]$ is at least of the order $n^{-1/2}$, and if $\prob$ is absolutely continuous, which is the case of interest here, the convergence rate cannot be faster than $n^{-1/d}$. 
Actually, the rate can be arbitrarily slow, even in the one-dimensional case \citep[Theorem~3.3]{bobkov+l:2019}. 
Precise rates under additional moment assumptions are given, for instance, in \citet{fournier2015rate}. 
In Appendix~\ref{app:empWassUnif}, we will show that the convergence in $p$th mean takes place uniformly over families $\model \subset \Prob_p(\Rd)$ of probability measures satisfying a uniform integrability condition.
For distributions on compact metric spaces, \citet{weed2019sharp} provide sharp rates for $\expec[W_p(\emprob, \prob)]$ in terms of what they coin the \emph{Wasserstein dimension} of $\prob$.
For Lebesgue-absolutely continuous measures on $\Rd$, this dimension is just $d$.
Moreover, they exploit McDiarmid's bounded difference inequality to derive a concentration inequality of $W_p^p(\emprob, \prob)$ around its expectation.

Asymptotic results on the distribution of the empirical Wasserstein distance in dimension~$d \ge 2$ are, however, surprisingly scarce. 
The question is whether there exist sequences $a_n > 0$ and $b_n \ge 0$ such that $a_n \{ W_p^p(\emprob, \prob) - b_n \}$ converges in distribution to a non-degenerate limit.
Although this problem has already attracted a lot of attention, a general answer remains elusive. 

The one-dimensional case is well-studied thanks to the link \eqref{eq:Wp:d1} to empirical quantile processes \citep{delbarrio2005, bobkov+l:2019}.
For discrete distributions, large-sample theory for the empirical Wasserstein distance is available too \citep{sommerfeld2018inference, tameling2019}.
For multivariate Gaussian distributions, a central limit theorem for the empirical Wasserstein of order $p = 2$ between the true distribution and the one with estimated parameters is given in \citet{rippl2016limit}.
Although interesting and useful for GoF testing (see Section~\ref{sec:simu:simple:bivGauss}), this result does not cover the empirical distribution $\emprob$.

\citet{ambrosio2018pde} exploit the possibility to linearize the $2$-Wasserstein distance in dimension $d = 2$ in case the optimal transport plan is close to the identity.
The technique requires balancing the errors due to the dual Sobolev norm approximation and a smoothing step. 
%It still seems impossible to let the regularization of the discrete measure go to zero while using the linearization of the optimal transport problem, hindering the possibility to obtain a central limit theorem, even in this simplified case 
\citet{mena2019statistical} derive a limit theorem for the empirical entropic optimal transport cost. 
We refer to the latter for an introduction to optimal transport with entropic regularization.
Recent progress has been booked in \citet{goldfeld+k:2020}, who obtain a central limit theorem for the empirical $1$-Wasserstein distance after smoothing the empirical and the true distributions with a Gaussian kernel. 

Important advances on the limit distribution have been made by \citet{del2017central} who obtained results under fixed alternatives. For general~$\prob, \probQ \in \Prob_{4+\delta}(\Rd)$ for some $\delta > 0$, they establish a central limit theorem for 
\[
	n^{1/2} \bigl[
		W_2^2(\emprob, \probQ) - \expec\{W_2^2(\emprob, \probQ)\}
	\bigr].
\]
The result is proved using the Efron--Stein inequality combined with stability of optimal transport plans.
Unfortunately, if $\probQ = \prob$, the asymptotic variance is zero, meaning that the random fluctuations of $W_2(\emprob, \prob)$ around its mean are of order smaller than~$n^{-1/2}$.
The authors conclude that their proof technique is of little use for the case we are interested in.
%Moreover, as mentioned above, $\expec\{W_p(\emprob, \prob)\}$ may converge to zero at a slower rate than $n^{-1/2}$. 
The crucial problem of the limiting distribution of the empirical Wasserstein distance thus remains an important and difficult open problem. 
%which apparently precludes the implementation of multivariate analogues of the existing one-dimensional procedures.

\subsection{Computational issues}
\label{sec: computations}

In the last decade, important numerical developments have taken place in the area of measure transportation. 
The problem to be faced here is the computation of the Wasserstein distance between a discrete and a continuous distribution, the so-called semi-discrete optimal transportation problem. 
Most algorithms to date rely on the dual formulation of the problem, assuming that the source continuous probability measure $\prob$ admits a density $f$ w.r.t.\ the Lebesgue measure on $\reals^d$; see, e.g., \citet[Section~6.4.2]{santambrogio2015optimal} for a didactic exposition.
This formulation is the basis for the multi-scale algorithm for the squared Euclidean distance ($p = 2$) developed in \citet{merigot2011multiscale}, with further improvements in \citet{levy2015numerical} and \citet{kitagawa2016convergence}.
%based on the method for solving constrained least-squares assignment problems in \citet{aurenhammer1998minkowski}. 
It requires constructing a power diagram or Laguerre--Voronoi diagram, partitioning $\Rd$ into convex polyhedra called power cells.
With the Euclidean distance as cost function ($p = 1$) the edges of the cells involved in the tessellation are no longer linear, making the computation more demanding \citep{hartmann2020semi}. 
\citet{genevay2016stochastic} show that a semi-discrete reformulation of the dual program can be tackled by the stochastic averaged gradient (SAG) method \citep{schmidt:2017}.

In our numerical experiments in Section~\ref{sec:simu}, we assess the finite-sample performance of the test statistic based on the $p$-Wasserstein distance for $p \in \{1, 2\}$ and for $d$-variate distributions for $d \in \{2, 5\}$.
To the best of our knowledge, an implementation of the SAG method is not yet available in \textsf{R} \citep{R}. 
After preliminary tests and running time assessment, we made the following choices of algorithms and implementations:
\begin{itemize}
	\item In case $p = 2$ and $d = 2$, we relied on the \textsf{R} package \textsf{transport} \citep{transport}, which implements the multi-scale algorithm in \citet{merigot2011multiscale}.
	\item In all other cases ($p = 1$ or $d = 5$), we relied on our own \textsf{C} implementation of the SAG method as employed in \citet{genevay2016stochastic}. 
\end{itemize}
A first version of the package making our implementation available is  to be found on https://github.com/gmordant/WassersteinGoF.
\section{Wasserstein GoF tests for  simple null hypotheses}
\label{sec: Test}

Let $\mathbf{X}_n = (X_1, \ldots, X_n)$ be an independent random sample from some unknown distri\-bution $\prob \in \Prob(\Rd)$. For some given fixed  $\prob_0 \in \Prob_p(\Rd)$, consider testing the simple null hypothesis 
\[
	\mathcal{H}_0^n: \prob = \prob_0 \qquad \text{against} \qquad \mathcal{H}_1^n: \prob\neq\prob_0 
\]
based on the observations $\mathbf{X}_n$. Note that $\prob$, under the alternative, is not required to have finite moments of order $p$.

Let $\emprob = n^{-1} \sum_{i=1}^n \delta_{X_i}$ denote the empirical distribution and consider the test statistic 
\begin{equation} 
\label{eq:Tn}
	T_n:= W_p^p(\emprob,\prob_0),
\end{equation}
the $p$th power of the $p$-Wasserstein distance between $\emprob$ and the distribution $\prob_0$ specified by the null hypothesis. Having bounded support, $\emprob$ trivially belongs to $\Prob_p(\Rd)$, so that $T_n$ is well defined. 

Actual computation of $T_n$ amounts to solving the semi-discrete optimal transport problem. In the numerical experiments in Section~\ref{sec:simu}, we show results for $p \in \{1, 2\}$. The theory, however, is developed for general~$p \ge 1$. 

Let $F_n(t) = \prob_0^n[T_n \le t]$ for $t \in [0, \infty)$ denote the distribution function of the test statistic under $\mathcal{H}_0^n$. Here, $\prob_0^n$ stands for the distribution under $\mathcal{H}_0^n$ of the observation $\mathbf{X}_n$, the $n$-fold product measure of $\prob_0$ on $(\Rd)^n$. 
The $p$-value of the test statistic is $1 - F_n(T_n)$, while the critical value for a test of size $\alpha \in (0, 1)$ is
\begin{equation}
\label{eq:crit:simple}
	\cPna
	:= \inf \big\{ t > 0 : F_n(t) \ge 1-\alpha \big\}
\end{equation}

The test we propose is then
\begin{equation}
	\phi^n_{\prob_0}= 
	\begin{cases} 
		1 & \text{if $1-F_n(T_n) \le \alpha$ or, equivalently, $T_n \ge \cPna$,} \\ 
		0 & \text{otherwise.} 
	\end{cases}
	\label{thetest}
\end{equation}
The exact size of the GoF test in \eqref{thetest} is $1-F_n(\cPna) \le \alpha$, with equality if and only if $F_n$ is continuous at $\cPna$. 
The type~I error is thus bounded by the nominal size~$\alpha$, and often equal to it. The null distribution of~$T_n$ depends on $\prob_0$, so that $\cPna$ needs to be calculated for each $\prob_0$ separately. 

Although $p$-values and critical values usually cannot be calculated analytically, they can be approximated with any desired degree of precision via the following simple Monte Carlo algorithm. Draw a large number of independent random samples of size $n$ from $\prob_0$, compute the test statistic for each such sample, and approximate $F_n$ by the empirical distribution function of the simulated test statistics. Critical values and $p$-values then can be calculated from the approximated $F_n$. By the Donsker theorem, any desired accuracy can be achieved by drawing sufficiently many samples. 

Under the alternative hypothesis, the test rejects the null hypothesis with probability tending to one, i.e., is consistent against any fixed alternative $\prob \ne \prob_0$.

\begin{prop}[Consistency]
	\label{prop:cons}
For every $\prob_0 \in \Prob_p(\Rd)$, the test $\phi^n_{\prob_0}$ is consistent against any $\prob \in \Prob(\Rd)$ with $\prob \ne \prob_0$: 
\[
	\lim_{n\to\infty} \prob^n[\phi^n_{\prob_0} = 1] = 1\qquad\text{for any $\alpha > 0$.}
\]
\end{prop}

\begin{proof}
	Fix $\prob_0 \in \Prob_p(\Rd)$. For any $\alpha > 0$, the critical value $c(\alpha, n, \prob_0)$ tends to zero as $n \to \infty$. Indeed, by \citet[Lemma~8.4]{bickel+f:1981}, we have~$T_n \to 0$ in~$\prob_0^n$-probability and thus $\lim_{n \to \infty} \prob_0^n[T_n > \eps] = 0$ for any $\eps > 0$. It follows that, for every $\alpha > 0$ and every $\eps > 0$, we have $\cPna \le \eps$ for all sufficiently large $n$.
	
	Let $\prob \in \Prob(\Rd)$ with $\prob \ne \prob_0$. We consider two cases according as $\prob$ has  finite moments of order $p$ or not. 
	
	First, suppose that $\prob \in \Prob_p(\Rd)$. Still by \citet[Lem\-ma~8.4]{bickel+f:1981}, we have $W_p(\emprob, \prob) \to 0$ in $\prob^n$-probability as $n \to \infty$. The triangle inequality for the metric $W_p$ yields
	\[
		\bigl| W_p(\emprob, \prob_0) - W_p(\prob, \prob_0) \bigr|
		\le
		W_p(\emprob, \prob)
		\to 0, \qquad n \to \infty
	\]
	in $\prob^n$-probability. Hence $T_n = W_p^p(\emprob, \prob_0) \to W_p^p(\prob, \prob_0)$ in $\prob^n$-probability as~$n \to \infty$. But $W_p^p(\prob, \prob_0) > 0$ since $\prob, \prob_0 \in \Prob_p(\Rd)$ and $\prob \ne \prob_0$ by assumption. It follows that $\lim_{n \to \infty} \prob^n[T_n > \cPna] = 1$, as required.
	
	Second, suppose that $\prob \in \Prob(\Rd) \setminus \Prob_p(\Rd)$. Let $\delta_0$ denote the Dirac measure at~$0 \in \Rd$. Since~$W_p$ is a metric, the triangle inequality implies 
	\[
		W_p(\emprob, \prob_0) \ge W_p(\emprob, \delta_0) - W_p(\prob_0, \delta_0).
	\]
	Now, $W_p(\prob_0, \delta_0)$ is a constant and $W_p^p(\emprob, \delta_0) = n^{-1} \sum_{i=1}^n \| X_i \|^p$. As the expecta\-tion of $\| X_1 \|^p$ under $\prob$ is infinite, the law of large numbers implies that~$W_p^p(\emprob, \delta_0) \to \infty$ in $\prob^n$-probability as $n \to \infty$. The same then holds for $T_n$ and thus 
	\[ 
		\lim_{n \to \infty} \prob^n[T_n > \cPna] = 1. 
		\qedhere 
	\]
\end{proof}

% ------------------------------------------------
\section{Wasserstein GoF tests for group families}
\label{sec:group}

Let $\probQ_0 \in \Prob(\Rd)$ and let $G$ be a group of measurable transformations $g : \reals^d \to \reals^d$. That is, $G$ should be closed under composition ($g_1, g_2 \in G$ implies $g_1 \circ g_2 \in G$) and under inversion ($g \in G$ implies $g^{-1} \in G$). If the random variable $Z$ has distribution $\probQ_0$, the random variable $g(Z)$ has distribution $g_\# \probQ_0 := \probQ_0 \circ g^{-1}$, where the subscripted symbol~$\#$ denotes the push-forward of a measure by a measurable function. Let $\model = \{g_\# \probQ_0 : g \in G \}$ be the group family generated by $G$ and $\probQ_0$. We assume further that the transformation $g$ is identifiable, that is, the map $g \mapsto g_\# \probQ_0$ is one-to-one, so that $g_1 \ne g_2$ implies that $g_1(Z)$ and $g_2(Z)$ have different distributions, with again $Z \sim \probQ_0$.
Note that for any element $\prob$ of $\model$ we have $\model = \{ g_\# \prob : g \in G \}$, so that the choice of $\probQ_0$ in $\model$ is in some sense arbitrary.

Group families form one of the two principal classes of models covered in \citet{lehmann:1998}. Here are some prominent examples of transformation groups $G$ on $\Rd$ and some models $\model$ that they generate.

\begin{example}[Location--scale families]
	\label{ex:locscale}
	For $(a, b) \in \Rd \times (0, \infty)^d$, define $g_{a,b} : \Rd \to \Rd$ by $g(x) = (a_j + b_j x)_{j=1}^d$ for $x \in \Rd$. The model $\model$ is the location-scale family generated by $\prob_0$. We can also consider just the location family generated by the subgroup $x \mapsto g_{a,1}(x) = (x_j+a_j)_{j=1}^d$ and the scale family generated by the subgroup $x \mapsto g_{0,b}(x) = (b_jx_j)_{j=1}^d$. In dimension $d = 1$, we can generate in this way the normal and exponential families, for instance. 
%	Clearly, $\norm{g_{a,b}(x)} \le \norm{a} + \norm{x} \max_j|b_j|$, so that $g_\# \prob_0$ has a finite $p$-th moment whenever $\prob_0$ has.
\end{example}

\begin{example}[Affine transformations and elliptical distributions]
	\label{ex:affine}
	For $a \in \Rd$ and $B \in \Rd \times \Rd$ non-singular, define $g_{a,B} : \Rd \to \Rd$ by $g(x) = a + Bx$ for $x \in \Rd$. If $\probQ_0$ is the $d$-variate standard normal distribution $\normal_d(0, I_d)$, then $\model$ is the family of all $d$-variate Gaussian distributions with positive definite covariance matrix. More generally, if $\probQ_0$ is spherically symmetric around the origin, then $\model$ is the family of elliptical distributions with a given characteristic generator and positive definite scatter matrix \citep{cambanis+h+s:1981, fang+k+n:1990}. Besides the Gaussian family, another common example is the multivariate Student t distribution with a fixed number of degrees of freedom. For elliptical families, the matrix $B$ is not identifiable from the model but only the matrix $B B'$ is. Identifiability can be restored by restricting $B$ to the set of lower triangular matrices with positive elements on the diagonal.\footnote{For every symmetric positive definite matrix $S \in \reals^{d \times d}$, there exists a unique lower triangular matrix $L \in \reals^{d \times d}$ with positive diagonal elements, called Cholesky triangle, producing the Cholesky decomposition $S = L L'$ \citep[Theorem~4.2.5]{golub1996matrix}.} Note that the case of elliptical distributions with possibly degenerate scatter matrices is not covered here, as the corresponding affine transformation is not invertible.
\end{example}

In the examples above, the transformation group $G$ is parametrized by a Euclidean parameter $\theta \in \Theta$ with $\Theta \subseteq \reals^k$ for some dimension $k$, so that $G = \{ g_\theta : \theta \in \Theta \}$. We will assume this to be the case in general and write $\prob_\theta = (g_\theta)_\# \probQ_0$. The model then takes the form $\model = \{ \prob_\theta : \theta \in \Theta \}$. The mappings $\theta \mapsto g_\theta$ and $g \mapsto g_\# \probQ_0$ are assumed to be one-to-one. The parametrization $\theta \mapsto \prob_\theta$ then is also one-to-one, i.e., the model parameter $\theta$ is identifiable. Models generated by infinite-dimensional transformation groups exist as well, but the theory here is intended for the finite-dimensional situation, as the conditions to come seem too restrictive otherwise.

Let $\probQ_0 \in \Prob_p(\Rd)$ for some $p \in [1, \infty)$. Assume that, for all $g \in G$, there exists $c_g > 0$ such that
\begin{equation}
\label{eq:g:bound}
	\forall x \in \Rd, \qquad 
	\norm{g(x)} \le c_g (1+\norm{x}).
\end{equation}
Then it is easy to verify that $g_\# \probQ_0$ belongs to $\Prob_p(\Rd)$ for every $g \in G$ too and, therefore, $\model \subset \Prob_p(\Rd)$. 
This condition on $g$ is fulfilled for the transformations in Examples~\ref{ex:locscale} and~\ref{ex:affine}.

Let $\model = \{ \prob_\theta = (g_\theta)_\# \probQ_0 : \theta \in \Theta \} \subset \Prob_p(\Rd)$ be a group family as just described. Given an i.i.d.\ sample $\mathbf{X}_n = (X_1, \ldots, X_n)$ from some unspecified $\prob \in \Prob_p(\Rd)$, we wish to test the hypothesis
\begin{equation}
\label{eq:hyp:group}
	\mathcal{H}_0^n : \prob \in \model
	\quad \text{against} \quad
	\mathcal{H}_1^n : \prob \not\in \model.
\end{equation}
The parameter $\theta$ of the transformation $g_\theta$ is an unknown nuisance. In contrast to Section~\ref{sec: Test}, the null hypothesis is thus a composite one. An important special case is when $\probQ_0$ is the $d$-variate standard normal distribution and $G$ is the affine group in Example~\ref{ex:affine}: the testing problem~\eqref{eq:hyp:group} then concerns the hypothesis of multivariate normality with unspecified positive definite covariance matrix.

Our testing strategy is to choose some estimator $\hat{\theta}_n$ for $\theta$ and compute ``residuals'' of the form
\begin{equation}
\label{eq:Zni}
	\hat{Z}_{n,i} := g_{\hat{\theta}_n}^{-1}(X_i), \qquad i = 1, \ldots, n,
\end{equation}
yielding an empirical distribution $\emprob^{\hat{Z}} := n^{-1} \sum_{i=1}^n \delta_{\hat{Z}_{n,i}}$. The test statistic we propose is
\begin{equation}
\label{eq:Tgroup}
	T_{\model,n} := W_p^p\bigl(\emprob^{\hat{Z}}, \probQ_0\bigr).
\end{equation}
If the null distribution of $(\hat{Z}_{n,1}, \ldots, \hat{Z}_{n,n})$ does not depend on the unkown parameter $\theta$, then we can compute critical values and $p$-values for $T_{\model, n}$ as if the true distribution is $\probQ_0$. As in Section~\ref{sec: Test}, the null distribution of $T_{\model,n}$ can then be computed up to any desired accuracy via Monte Carlo random sampling from $\probQ_0$, and this prior to having observed the sample. % the elliptical distribution in $\elliptical(f_{\textrm{\tiny rad}})$ with mean zero and identity covariance. 

For any $g \in G$, let $\bar{g} : \Theta \to \Theta$ denote the mapping $\theta \mapsto \bar{g}(\theta)$ characterized by $g \circ g_\theta = g_{\bar{g}(\theta)}$, so that $g_\# \prob_\theta = \prob_{\bar{g}(\theta)}$.
The estimator $\hat{\theta}_n = \theta_n(\mathbf{X}_n)$ is said to be \emph{equivariant} \citep[Definition~2.5]{lehmann:1998} if for every $g \in G$ and for every $(x_1,\ldots,x_n) \in (\Rd)^n$, we have
\begin{equation}
\label{eq:equivar}
\theta_n\bigl(g(x_1),\ldots,g(x_n)\bigr) 
= \bar{g}\bigl(\theta_n(x_1,\ldots,x_n)\bigr).
\end{equation}
%For this procedure to work, the estimator $\hat{\theta}_n$ needs to satisfy a special property. For any $\theta \in \Theta$ and any $g \in G$, the transformation $g \circ g_\theta$ belongs to $G$ too and is thus associated to some parameter in $\Theta$ as well: we write this parameter as $\bar{g}(\theta)$, where $\bar{g} : \Theta \to \Theta$, so that $g_{\bar{g}(\theta)} = g \circ g_{\theta}$. 
%%(Note that this defines a left group action by $G$ on $\Theta$.) 
%The estimator $\hat{\theta}_n = \theta_n(\mathbf{X}_n)$ is said to be \emph{equivariant} \citep[Definition~2.5]{lehmann:1998} if the map $\theta_n : (\Rd)^n \to \Theta$ satisfies the following property: for every $g \in G$ and for every $(x_1,\ldots,x_n) \in (\Rd)^n$, we have
%\begin{equation}
%\label{eq:equivar}
%	\theta_n\bigl(g(x_1),\ldots,g(x_n)\bigr) 
%	= \bar{g}\bigl(\theta_n(x_1,\ldots,x_n)\bigr).
%\end{equation}
%In words, first transforming the sample and then computing the estimate produces the same result as first computing the estimate and then transforming the estimated parameter.
Equivariance is a natural symmetry requirement and is satisfied for many common estimators. For a location parameter, it is satisfied by the mean and the median, or in fact any weighted average of the order statistics. For a scale parameter, it is satisfied by the standard deviation and by the mean or median absolute deviation. For the affine group in Example~\ref{ex:affine} with $B$ restricted to be lower triangular and with positive diagonal elements, it is satisfied by the lower Cholesky triangle of the empirical covariance matrix. The proof of the latter is elementary and follows from the uniqueness of the Cholesky decomposition and the fact that the set of lower triangular matrices with positive diagonal elements forms a multiplicative group.
Equivariance is also satisfied by maximum likelihood estimators provided the transformations $g$ are diffeormorphisms: by the change-of-variables formula, $\hat{\theta}_n$ maximizes the likelihood given the sample $x_1,\ldots,x_n$ if and only if $\bar{g}(\hat{\theta}_n)$ maximizes the likelihood given the sample $g(x_1),\ldots,g(x_n)$.

%\begin{lem}
%	\label{lem:group:invariance}
In the group model $\model = \{ \prob_\theta = (g_\theta)_\# \probQ_0 : \theta \in \Theta \}$, if the estimator $\hat{\theta}_n$ is equivariant, then for an i.i.d.\ sample $X_1,\ldots,X_n$ from $\prob_\theta \in \model$, the joint distribution of $(\hat{Z}_{n,i})_{i=1}^n$ in \eqref{eq:Zni} does not depend on $\theta \in \Theta$ and is the same as if $X_1,\ldots,X_n$ were an i.i.d.\ sample from $\probQ_0$.
%\end{lem}
%
\iffalse
\begin{proof}
	Let $Z_1,\ldots,Z_n$ be an i.i.d.\ sample from $\probQ_0$. Since $\prob_\theta = (g_\theta)_\# \probQ_0$, the sample $(X_i)_{i=1}^n$ is equal in distribution to $(g_\theta(Z_i))_{i=1}^n$. Let $\hat{\theta}_{n,Z} = \theta_n(Z_1,\ldots,Z_n)$ be the parameter estimate based on $(Z_i)_{i=1}^n$. By equivariance \eqref{eq:equivar}, we have
	\[
		\hat{\theta}_n 
		= \theta_n\bigl(g(Z_1),\ldots,g(Z_n)\bigr)
		= \bar{g}(\hat{\theta}_{n,Z}).
	\]
	In terms of transformations, this means that
	\[
		g_{\hat{\theta}_n} = g \circ g_{\hat{\theta}_{n,Z}}.
	\]
	But then, jointly in $i = 1, \ldots, n$, we have
	\[
		\hat{Z}_{n,i}
		\stackrel{d}{=} g_{\hat{\theta}_n}^{-1}\bigl( g(Z_i) \bigr)
		=
		\bigl(g_{\hat{\theta}_{n,Z}}^{-1} \circ g^{-1}\bigr) \bigl(g(Z_i)\bigr)
		=
		g_{\hat{\theta}_{n,Z}}^{-1}(Z_i).
	\]
	But the latter are the residuals computed from $Z_1,\ldots,Z_n$.
\end{proof}
\fi
As a consequence, %for an equivariant estimator $\hat{\theta}_n$, 
the distribution of $T_{\model,n}$ in~\eqref{eq:Tgroup} under any $\prob_\theta \in \model$ is the same as under $\probQ_0$. Let $F_{\model,n}$ denote its cumulative distribution function. The $p$-value of the observed test statistic is
\[
	1 - F_{\model,n}(T_{\model,n})
\]
while the critical value at level $\alpha \in (0, 1)$ is
\[
	\cMna = \inf \{ c > 0 : F_{\model,n}(c) \ge 1-\alpha \}.
\]
For the testing problem \eqref{eq:hyp:group}, we propose the test
\begin{equation}
\label{eq:phi:group}
	\phi_{\model}^n :=
	\begin{cases}
		1 & \text{if $1 - F_{\model,n}(T_{\model,n}) \le \alpha$, or equivalently, $T_{\model,n} \ge \cMna$,} \\
		0 & \text{otherwise.}
	\end{cases}
\end{equation}
The actual size of the test is $1 - F_{\model,n}(\cMna) \le \alpha$, with equality if and only if $F_{\model,n}$ is continuous in $\cMna$. 
Formally, the case of a single null hypothesis in Section~\ref{sec: Test} can be seen as a special case by letting $\prob_0 = \probQ_0$ and $G$ the trivial group containing only the identity mapping.

Since the null distribution $F_{\model,n}$ does not depend on any unknown parameter, critical values and $p$-values values can be computed with arbitrary precision by a Monte Carlo algorithm as we did in Section~\ref{sec: Test}. The difference is now that we generate samples from $\probQ_0$. Note that the critical values can be computed prior to having seen the data.

To show that the test is consistent, we need an extra assumption on $G$: for every $\theta \in \Theta$ there exists $m_\theta > 0$ such that for all $\theta' \in \Theta$ in some neighbourhood of $\theta$, we have
\begin{equation}
\label{eq:g:lip}
	\sup_{x \in \Rd} \frac{\norm{g_{\theta'}^{-1} \circ g_\theta(x) - x}}{1 + \norm{x}}
	\le
	m_\theta \norm{\theta' - \theta}.
\end{equation}
The condition is fulfilled for the transformation groups and parametrizations in Examples~\ref{ex:locscale} and~\ref{ex:affine}. For the affine group in Example~\ref{ex:affine}, the property~\eqref{eq:g:lip} follows from continuity of matrix inversion with respect to the matrix norm induced by the Euclidean norm. We will also need weak consistency of the estimator: for every $\theta \in \Theta$, we have $\hat{\theta}_n \to \theta$ as $n \to \infty$ in $\prob_\theta^n$-probability. To prove this for the affine group in Example~\ref{ex:affine}, it is helpful to know that the map that sends a positive definite symmetric matrix to its Cholesky triangle is differentiable \citep{smith1995differentation} and thus continuous.

\begin{prop}[Consistency against fixed alternatives]
	\label{prop:group}
	Let the group family $\model = \{ \prob_\theta = (g_\theta)_\# \probQ_0 : \theta \in \Theta \} \subset \Prob_p(\Rd)$ be such that $\theta$ is identifiable as above and such that~\eqref{eq:g:bound} and~\eqref{eq:g:lip} are satisfied. Let $\hat{\theta}_n$ be an equivariant and weakly consistent estimator sequence of $\theta \in \Theta$.
	\begin{enumerate}[(i)]
		\item We have $T_{\model,n} \to 0$ as $n \to \infty$ in $\prob_\theta^n$-probability for any $\theta \in \Theta$.
		\item Let $\prob \in \Prob_p(\Rd) \setminus \model$. If $\hat{\theta}_n$ converges weakly to some $\theta \in \Theta$ under $\prob^n$, then $\prob^n[ \phi_{\model}^n = 1 ] \to 1$ as $n \to \infty$.
	\end{enumerate}
\end{prop}

The pseudo-parameter $\theta$ in Proposition~\ref{prop:group}(ii) depends on the estimator: for instance, in a location-scale model and if $p \ge 2$, if we estimate the location and scale parameters by the empirical mean and standard deviation, respectively, then $\theta$ denotes the vector of population means and standard deviations.

\begin{proof}
	(i) %As in the proof of Lemma~\ref{lem:group:invariance}, 
	The sample $(X_i)_{i=1}^n$ is equal in distribution to $(g_\theta(Z_i))_{i=1}^n$ for some $\theta \in \Theta$, where $(Z_i)_{i=1}^n$ is an i.i.d.\ sample from $\probQ_0$.
	Since we are interested in convergence in probability, we can then in fact suppose that $X_i = g_\theta(Z_i)$ for all $i = 1, \ldots, n$ and compute probabilities under $\probQ_0^n$.
	
	The empirical distribution of $(Z_i)_{i=1}^n$ is denoted by $\emprob^{Z}$. By the triangle inequality for the Wasserstein distance,
	\begin{equation}
	\label{eq:TMn1p:bound}
		T_{\model,n}^{1/p}
		= W_p\bigl( \emprob^{\hat{Z}}, \probQ_0) \\
		\le W_p\bigl( \emprob^{\hat{Z}}, \emprob^{Z} \bigr)
		+ W_p\bigl( \emprob^{Z}, \probQ_0).
	\end{equation}
	Since $\probQ_0$ has a finite moment of order $p$, the second term on the right-hand side converges to zero in probability by 	 \citet[Lemma~8.4]{bickel+f:1981}. 
	
	To bound the first term on the right-hand side of the previous equation, consider the coupling of $\emprob^{\hat{Z}}$ and $\emprob^{Z}$ via the discrete uniform distribution on the pairs~$(\hat{Z}_{n,i}, Z_i)$ for $i = 1, \ldots, n$. It follows that
	\[
		W_p^p\bigl(\emprob^Z, \probQ_0\bigr)
		\le
		\frac{1}{n} \sum_{i=1}^n \norm{ \hat{Z}_i - Z_i }^p
		= \frac{1}{n} \sum_{i=1}^n 
		\norm{ g_{\hat{\theta}_{n,Z}}^{-1}(Z_i) - Z_i }^p,
	\]
	where $\hat{\theta}_{n,Z} = \theta_n(Z_1,\ldots,Z_n)$ is the estimated parameter from $(Z_i)_{i=1}^n$.
%	, see the end of the proof of Lemma~\ref{lem:group:invariance}. 
	Let $\theta_e \in \Theta$ denote the parameter that corresponds to the identity transformation: $g_{\theta_e}(x) = x$ for all $x \in \Rd$. Then $\prob_{\theta_e} = \probQ_0$ and, by assumption, $\hat{\theta}_{n,Z} \to \theta_e$ as $n \to \infty$ in probability. Let $\eps > 0$ be small enough so that \eqref{eq:g:lip} holds for all $\theta' \in \Theta$ with $\norm{\theta' - \theta_e} \le \eps$. Then, on the event that $\norm{\hat{\theta}_{n,Z} - \theta_e} \le \eps$, we have
	\[
		\frac{1}{n} \sum_{i=1}^n 
		\norm{ g_{\hat{\theta}_{n,Z}}^{-1}(Z_i) - Z_i }^p
		\le
		m_{\theta_e}^p \norm{\hat{\theta}_{n,Z} - \theta_e}^p \cdot
		\frac{1}{n} \sum_{i=1}^n \left( 1 + \norm{Z}_i \right)^p.
	\]
	As $\probQ_0 \in \Prob_p(\Rd)$, the weak consistency of $\hat{\theta}_{n,Z}$ and the law of large numbers imply that, in probability, 
	\[ 
		\frac{1}{n} \sum_{i=1}^n \norm{ g_{\hat{\theta}_{n,Z}}^{-1}(Z_i) - Z_i }^p \to 0, \qquad n \to \infty.
	\]
	
	We conclude that both terms in the bound \eqref{eq:TMn1p:bound} for $T_{\model,n}^{1/p}$ converge to zero in probability. Hence the same is true for $T_{\model,n}$.
	
	(ii) By (i), it follows that $\lim_{n \to \infty} F_{\model,n}(\eps) = 1$ for every $\eps > 0$. It is thus sufficient to show that, under the alternative hypothesis, there exists $\eps > 0$ such that $\lim_{n \to \infty} \prob^n[T_{\model,n} > \eps] = 1$. 
	
	Let $\probQ_1 = (g_\theta^{-1})_{\#} \prob$, that is, $\probQ_1$ is the law of $Y = g_\theta^{-1}(X)$, where $g_\theta^{-1} \in G$ is the inverse transformation of $g_\theta$ and where $X$ has law $\prob$. By assumption, $\probQ_1 \ne \probQ_0$, for otherwise $\prob \in \model$. Also, $\probQ_1 \in \Prob_p(\Rd)$, since $\prob \in \Prob_p(\Rd)$ and since each $g$ in $G$ satisfies \eqref{eq:g:bound}. 
	
	Put $Y_i = g_\theta^{-1}(X_i)$ for $i = 1, \ldots, n$, an i.i.d.\ sample from $\probQ_1$. Let $\emprob^{Y} = n^{-1} \sum_{i=1}^n \delta_{Y_i}$ be its empirical distribution. The estimated residuals are $\hat{Z}_i = g_{\hat{\theta}_n}^{-1}(X_i) = g_{\hat{\theta}_n}^{-1} \circ g_\theta(Y_i)$. By the same argument as in (i), we have
	\[
		W_p\bigl( \emprob^{\hat{Z}}, \probQ_1\bigr)
		\le
		\left[
		\frac{1}{n} \sum_{i=1}^n 
		\norm{ g_{\hat{\theta}_n}^{-1} \circ g_\theta(Y_i) - Y_i }^p 
		\right]^{1/p}
		+
		W_p\bigl( \emprob^{Y}, \probQ_1 \bigr)
		\to 0,
	\]
	as $n \to \infty$ in probability. By the continuous mapping theorem, it follows that $T_{\model,n} \to W_p^p(\probQ_1, \probQ_0) > 0$ as $n \to \infty$ in probability. But then $1-F_{\model,n}(T_{\model,n}) \to 0$ as $n \to \infty$ in probability, and the null hypothesis is rejected with probability tending to one.
\end{proof}

% ------------------------------------------------------------
\section{Wasserstein GoF tests for  general parametric families}
\label{sec:param}

Extending the scope of Section~\ref{sec:group}, consider the problem of testing whether the unknown common distribution $\prob$ of a sample of observations belongs to some parametric family~$\model := \big\{ \prob_{\theta} : \theta \in \Theta\big\}$ of distributions on $\Rd$. 
The parameter space $\Theta$ is some metric space and the map $\theta \mapsto \prob_{\theta}$ is assumed to be one-to-one and continuous in a sense to be specified. 
Given an independent random sample~$\mathbf{X}_n = (X_1, \ldots, X_n)$ from some unknown distribution $\prob \in \Prob(\Rd)$, the goodness-of-fit problem consists of testing
\begin{equation}
\label{eq:hyp:comp}
	\mathcal{H}_0^n: \prob\in\model
  	\quad \text{against} \quad \mathcal{H}_1^n: \prob \notin \model.
\end{equation}

Assume that every $\prob_\theta\in\model$ has a finite moment of order $p \in [1, \infty)$, that is,~$\model \subseteq \Prob_p(\Rd)$. Recall that $\emprob$ denotes the empirical distribution of the sample. The test statistic we propose is
\begin{equation}
\label{eq:TMntheta}
	T_{\model,n} := W_p^p(\emprob, \prob_{\hat\theta_n})
\end{equation}
where $\hat\theta_n = \theta_n(\mathbf{X}_n)$ is some consistent (under $\mathcal{H}_0^n$) estimator sequence of the true parameter $\theta$. The distribution of $\mathbf{X}_n$ under $\mathcal{H}_0^n$ in \eqref{eq:hyp:comp} being $\prob_\theta^n$ for some~$\theta \in \Theta$, let $F_{n,\theta}(t) = \prob_\theta^n[T_{\model,n} \le t]$ for $t \in \reals$ denote the null distribution function of the test statistic. As $p$-value and critical value, we would like to take
\begin{equation}
\label{eq:crit:parametric}
	1 - F_{n,\theta}(T_{\model,n})
	\text{ and }
	\cMnta
	=
	\inf \{ t \ge 0 : F_{n,\theta}(t) \ge 1-\alpha \},
\end{equation}
respectively, for some $\alpha \in (0, 1)$.
This choice is infeasible, however, since the true parameter $\theta$ is unknown. Therefore, we propose to replace $\cMnta$ by the bootstrapped quantity $\cMntha$, yielding the test
\begin{equation}
	\phi_{\model}^n%_{\prob_0}
	:= 
	\begin{cases} 
		1 & \text{if $1 - F_{n,\hat{\theta}_n}(T_{\model,n}) \le \alpha$ or, equivalently, $T_{\model,n} \ge \cMntha$,} \\ 
		0 & \text{otherwise.} 
	\end{cases}
	\label{eq: thetest2}
\end{equation}
We reject $\mathcal{H}_0^n$ as soon as $T_{\model,n}$ exceeds the critical value at the estimated parameter. The substitution of $\theta$ by $\hat{\theta}_n$ qualifies as a parametric bootstrap.

To compute the critical value $\cMntha$ in practice, we rely, as before, on a Monte Carlo approximation: resample from $\prob_{\hat{\theta}_n}$, compute the test statistic, and approximate $F_{n,\hat{\theta}_n}$ by the empirical distribution function of the resampled test statistics. 
By the Dvoretzky--Kiefer--Wolfowitz inequality \citep{massart1990}, the difference between $F_{n,\hat{\theta}_n}(t)$ and its Monte Carlo approximation can be controlled explicitly and uniformly in $t \ge 0$ and in the unknown parameter, and this in terms of the Monte Carlo sample size only. To speed up the calculations in case of a low-dimensional parameter space, we pre-compute an approximation of the critical value function $\theta \mapsto \cMnta$ in this way for $\theta$ in a finite grid $\Theta' \subset \Theta$ and then compute $\cMntha$ by interpolation and/or smoothing.

Under the null hypothesis and if the true parameter is $\theta$, the size of the test is now the random quantity
\[
	1 - F_{n,\theta}\bigl(\cMntha\bigr).
\]
In contrast to Sections~\ref{sec: Test} and~\ref{sec:group}, it is no longer guaranteed that this risk is bounded by $\alpha$. The question remains open whether under the null hypothesis the actual size of the test indeed converges to $\alpha$. To prove this conjecture would require non-degenerate limit distribution theory for $W_p^p(\emprob, \prob_{\theta})$, not only for fixed $\theta \in \Theta$, but even for sequences $\theta_n$ converging to $\theta$ at certain rates which depend on the model $\model$ under study \citep{beran1997, capanu:2019}. 
As discussed in Section~\ref{sec:asy}, such asymptotic distribution theory is still far beyond the horizon. 
Our numerical experiments in Section~\ref{sec:simu}, however, support the conjecture that the parametric bootstrap produces a test with the right asymptotic size. 
\bgroup
For any $\theta \in \Theta$, any $\eps > 0$, and a sufficiently regular parametric model $\model$ and estimator sequence $\hat{\theta}_n$, we conjecture that $\prob_\theta^n[ 1 - F_{n,\theta}(\cMntha) \le \alpha + \eps ]$ converges to one as $n \to \infty$.
In Appendix~\ref{sec:boot:cons}, we provide a theoretical justification of the consistency of the parametric bootstrap in the univariate case, for which the asymptotic distribution theory of the empirical Wasserstein distance is well developed.
\egroup

Nevertheless, against a fixed alternative, the consistency of the test \eqref{eq: thetest2} based on the parametric bootstrap can be established theoretically. 
The key is a law of large numbers for the empirical distribution in Wasserstein distance uniformly over classes of distributions that satisfy a uniform integrability condition, see Appendix~\ref{app:empWassUnif}.
For the parameter estimator $\hat{\theta}_n$, we assume weak consistency locally uniformly in $\theta$: if $\rho$ denotes the metric on $\Theta$ and if $\cmp(\Theta)$ denotes the collection of compact subsets of $\Theta$, we will require that
\begin{equation}
\label{eq:weakConsUnif}
	\forall \eps > 0, \, \forall K \in \cmp(\Theta), \qquad
	\lim_{n \to \infty} \sup_{\theta \in K}
	\prob_\theta^n\big[ \rho(\hat{\theta}_n, \theta) > \eps\big] = 0.
\end{equation}
As illustrated in Remark~\ref{rmk:uniformity} below, this condition is satisfied, for instance, for moment estimators of a Euclidean parameter under a uniform integrability condition.

%Regarding the model, we assume that its parametrization is $W_2$-continuous, that is, $\lim_{n \to \infty} \theta_n = \theta$ in $\Theta$ implies $\lim_{n \to \infty} W_2(\prob_{\theta_n}, \prob_\theta) = 0$.

\begin{prop}[Consistency]
	\label{prop:param:cons}
	Let $\model = \{ \prob_\theta : \theta \in \Theta \} \subseteq \Prob_p(\reals^d)$, for $p \in [1, \infty)$,  be a model indexed by a metric space  $(\Theta, \rho)$. 
 Assume the following conditions: 
 
	\begin{enumerate}[(a)]
		\item the map $\Theta \to \Prob_p(\Rd) : \theta \mapsto \prob_\theta$ is one-to-one and $W_p$-continuous;
		\item $\hat{\theta}_n$ is weakly consistent locally uniformly in $\theta \in \Theta$, i.e., \eqref{eq:weakConsUnif} holds.
	\end{enumerate}
	
	Then, the following properties hold:
	\begin{enumerate}[(i)]
	\item $T_{\model,n} \to 0$ in $\prob_\theta^n$-probability locally uniformly in $\theta \in \Theta$, i.e., 
	\[
		\forall \eps > 0, \, \forall K \in \cmp(\Theta), \qquad 
		\lim_{n \to \infty} \sup_{\theta \in K} 
		\prob_\theta^n\big[T_{\model,n} > \eps\big] = 0;
	\]
	\item the critical values $\cMnta$ tend to zero uniformly in $\theta$, i.e.,
	\begin{equation*}
%	\label{eq:critUnifParam}	
		\forall \alpha > 0, \, \forall K \in \cmp(\Theta),
		\qquad
		\lim_{n \to \infty} \sup_{\theta \in K} 
		\cMnta = 0;
	\end{equation*}
	\item 
	for every $\prob \in \Prob(\Rd) \setminus \model$ such that there exists $K \in \cmp(\Theta)$ with 
	$$\prob^n\big[\hat{\theta}_n \in K\big] \to 1\qquad \text{as~$n \to \infty$,}$$
	 we have $\lim_{n \to \infty} \prob^n\big[\phi_{\model}^n = 1\big] = 1$.
	\end{enumerate}
\end{prop}

\begin{proof}
	\emph{(i)}
	By the triangle inequality, it follows that
	\begin{equation}\label{trineq}
		T_{\model,n}^{1/p}
		= W_p(\emprob, \prob_{\hat{\theta}_n})
		\le W_p(\emprob, \prob_\theta) 
		+ W_p(\prob_\theta, \prob_{\hat{\theta}_n})
	\end{equation}
	for all $\theta \in \Theta$. 
	It is then sufficient to show that, for any compact $K \subseteq \Theta$, each of the~$W_p$-distances on the right-hand side of \eqref{trineq} converges to $0$ in $\prob_\theta^n$-probability uniformly in~$\theta \in K$.
	
	First, since $K$ is compact and   $\theta \mapsto \prob_\theta$ is $W_p$-continuous, the set $$\model_K := \{\prob_\theta : \theta \in K\}$$ is compact in $\Prob_p(\Rd)$ equipped with the $W_p$-distance. By \citet[Lemma~8.3(b)]{bickel+f:1981} or \citet[Definition~6.8(b) and Theorem~6.9]{villani2008optimal} and a subsequence argument, it follows that $x \mapsto \norm{x}^p$ is uniformly integrable with respect to $\model_K$, i.e.,
	\[
		\lim_{r \to \infty} \sup_{\theta \in K}
		\int_{\norm{x} > r} \norm{x}^p \, \diff \prob_\theta(x)
		= 0.
	\]
	Corollary~\ref{cor:empWassUnif} then implies that $W_p(\emprob, \prob_\theta) \to 0$ in $\prob_\theta^n$-probability as $n \to \infty$, uniformly in $\theta \in K$. 
	
	Second, as $K$ is compact and $\theta \to \prob_\theta$ is $W_p$-continuous, there exists, for every scalar~$\eps > 0$, a scalar $\delta = \delta(\eps) > 0$ such that\footnote{This is a slight generalization of the well-known property that a continuous function on a compact set is uniformly continuous. As a proof, fix $\eps > 0$ and consider for each $\theta \in K$ a scalar~$\delta(\theta) > 0$ such that for all~$\theta' \in \Theta$ with $\rho(\theta, \theta') \le \delta(\theta)$ we have $W_p(\prob_{\theta}, \prob_{\theta'}) \le \eps/2$. Cover~$K$ by open balls with centers $\theta \in K$ and radii $\delta(\theta)/2$. By compactness, extract a finite cover with centers $\theta_1, \ldots, \theta_m \in K$. Put $\delta = \min_j \delta(\theta_j)/2$. For every $\theta \in K$ and $\theta' \in \Theta$ with $\rho(\theta, \theta') \le \delta$, there exists $j = 1, \ldots, m$ such that $\rho(\theta, \theta_j) < \delta(\theta_j)/2$ and then also $\rho(\theta', \theta_j) < \delta(\theta_j)$. By the triangle inequality, $W_p(\prob_{\theta}, \prob_{\theta'}) \le W_p(\prob_{\theta_j}, \prob_{\theta}) + W_p(\prob_{\theta_j}, \prob_{\theta'}) \le \eps$.}
	\[
		\forall \theta \in K, \; \forall \theta' \in \Theta, \qquad
		\rho(\theta, \theta') \le \delta \implies W_p(\prob_{\theta}, \prob_{\theta'}) \le \eps.	
	\]
	It follows that
	\[
		\forall \theta \in K, \qquad
		\prob_\theta^n\big[W_p(\prob_{\theta}, \prob_{\hat{\theta}_n}) > \eps\big]
		\le \prob_\theta^n\big[\rho(\theta, \hat{\theta}_n) > \delta\big].
	\]
	By condition~(b), the latter probability converges to $0$ as $n \to \infty$ uniformly in~$\theta \in K$.

	\emph{(ii)} Fix $\alpha > 0$, $\eps > 0$, and $K \in \cmp(\Theta)$. By (i), there exists an integer $n(\eps) \ge 1$ such that
	\[
		\forall n \ge n(\eps), \; \forall \theta \in K, \qquad
		\prob_\theta^n\big[T_{\model,n} > \eps\big] \le \alpha.
	\]
	By definition of the critical values, also $\cMnta \le \eps$ for all $n \ge n(\eps)$ and~$\theta \in K$.

	\emph{(iii)} Let $\prob$ and $K$ be as in the statement. Put $c_n = \sup_{\theta \in K} \cMnta$. We have
	\begin{align*}
		\prob^n\big[\phi_{\model}^n = 1\big] 
		&\ge \prob\big[T_{\model,n} > \cMntha, \, \hat{\theta}_n \in K\big] \\
		&\ge \prob\big[T_{\model,n} > c_n, \, \hat{\theta}_n \in K\big].
	\end{align*}
	In view of {\it (ii)}, we have $c_n \to 0$ as $n \to \infty$, so that it is sufficient to show that there exists~$\eps > 0$, depending on $\prob$ and $\model$, such that $\lim_{n \to \infty} \prob^n\big[T_{\model,n} > \eps\big] = 1$. Consider two cases,  $\prob \in \Prob_p(\Rd) \setminus \model$ and $\prob \in \Prob(\Rd) \setminus \Prob_p(\Rd)$, according as $\prob$ has a finite moment of order $p$ or not.
	
	First, suppose that $\prob \in \Prob_p(\Rd) \setminus \model$. We have $W_p(\prob, \prob_\theta) > 0$ for every~$\theta \in \Theta$ while the map $\theta \mapsto W_p(\prob, \prob_\theta)$ is continuous. As $K$ is compact,~$\eta := \inf \big\{ W_p(\prob, \prob_\theta) : \theta \in K \big\} > 0$. On the event $\{\hat{\theta}_n \in K\}$, the triangle inequality implies
	\begin{align*}
		T_{\model,n}^{1/p}
		= W_p(\emprob, \prob_{\hat{\theta}_n}) 
		&\ge W_p(\prob, \prob_{\hat{\theta}_n}) - W_p(\emprob, \prob) \\
		&\ge \eta - W_p(\emprob, \prob).
	\end{align*}
	We obtain that
	\begin{align*}
		\prob^n\big[\phi_{\model}^n = 1\big] 
		&\ge \prob\big[T_{\model,n}^{1/p} > c_n^{1/p}, \, \hat{\theta}_n \in K\big] \\
		&\ge \prob\big[W_p(\emprob, \prob) < \eta - c_n^{1/p}, \, \hat{\theta}_n \in K\big].
	\end{align*}
	As $\eta > 0$ and $\lim_{n \to \infty} c_n = 0$, the latter probability converges to one by the assumption made on~$K$ and the fact that~$W_p(\emprob, \prob) \to 0$ in $\prob^n$-probability as~$n \to \infty$.

	Second, suppose that $\prob \in \Prob(\Rd) \setminus \Prob_p(\Rd)$. Since $\theta \mapsto W_p(\prob_\theta, \delta_0)$ is continuous, $\sup_{\theta \in K} W_p(\prob_\theta, \delta_0)$ is finite, with $K$ as in~(iii) and $\delta_0$ the Dirac measure at $0 \in \Rd$. By an argument similar to the second part of the proof of Proposition~\ref{prop:cons}, it follows that $\prob^n[T_{\model,n} > c_n, \, \hat{\theta}_n \in K] \to 1$ as $n \to \infty$.
%	By the triangle inequality, on the event $\{\hat{\theta}_n \in K\}$,
%	\begin{align*}
%		T_{\model,n}^{1/p}
%		= W_p(\emprob, \prob_{\hat{\theta}_n}) 
%		&\ge
%		W_p(\emprob, \delta_0) - W_p(\prob_{\hat{\theta}_n}, \delta_0) \\
%		&\ge
%		W_p(\emprob, \delta_0) - s.
%	\end{align*}
%	 Moreover, $W_p^p(\emprob, \delta_0) = n^{-1} \sum_{i=1}^n \norm{X_i}^p$  diverges to $\int \norm{x}^p \, \diff \prob(x) = \infty$ in $\prob^n$-probability by the weak law of large numbers. It follows that 
%	 $$\lim_{n \to \infty} \prob^n\big[T_{\model,n} > c_n, \, \hat{\theta}_n \in K\big]~\!=~\!1. \vspace{-7mm}$$
	 \end{proof}
 
\begin{rmk}[Uniform consistency]
	\label{rmk:uniformity}
	Under a mild moment condition, the uniform consistency condition (b) in Proposition~\ref{prop:param:cons} is satisfied for {\it method of moment estimators}---call them {\it moment estimators}---of a Euclidean parameter~$\theta \in \Theta \subseteq \reals^k$. In the method of moments, an estimator $\hat{\theta}_n$ of $\theta$ is obtained by solving (with respect to $\theta$)  the equations 
	\[
		\frac{1}{n}\sum_{i=1}^n f_j(X_i) = {\rm E}_\theta[ f_j(X) ], \qquad j=1,\dots,k,
	\]
	for some given $k$-tuple  $f:=(f_1,\dots,f_k)$ of functions such that  %the map
	 $m$  : 
	 $\theta \mapsto {\rm E}_\theta[f(X)]$ is a homeomorphism between $\Theta$ and $m(\Theta)$; see, for instance,~\citet[Chapter~4]{vdv98}. The consistency of $\hat{\theta}_n = m^{-1}(n^{-1} \sum_{i=1}^n f(X_i))$ uniformly in~$\theta \in~\!K$ for any compact $K \subseteq \Theta$ then follows from the uniform consistency  over $K$ of~$n^{-1} \sum_{i=1}^n f(X_i)$ as an estimator of ${\rm E}_\theta[ f(X) ]$ for such $\theta$. By \citet[Proposition~A.5.1]{vdvw96}, a sufficient condition for the latter is that the functions $f_j$ are  $\prob_\theta$-uniformly integrable  for $\theta \in K$, i.e.,
	\[
		\lim_{M \to \infty} \sup_{\theta \in K} {\rm E}_\theta\big[|f_j(X)| \, I\{|f_j(X)| > M \}\big] = 0,
		\qquad j = 1, \ldots, k.
	\]
	Since $I\{|f_j(X)| > M\} \le |f_j(X)|^\eta/M^\eta$ for $\eta > 0$, a further sufficient condition is that there exists $\eta > 0$ such that $\sup_{\theta \in K} {\rm E}_\theta[|f_j(X)|^{1+\eta}] < \infty$ for $j = 1, \ldots, k$.
\end{rmk}

\begin{rmk}[Parameter estimate under the alternative]
	In Proposition~\ref{prop:param:cons}(iii), the condition that there exists a compact $K \subseteq \Theta$ such that $\lim_{n \to \infty} \prob^n[\hat{\theta}_n \in K] = 1$ holds, for instance, when $\Theta$ is locally compact and $\hat{\theta}_n$ is consistent for a pseudo-parameter $\theta(\prob) \in \Theta$. This is the case for the moment estimators of Remark~\ref{rmk:uniformity} when $\Theta \subseteq \reals^k$ is open  and $f$ is~$\prob$-integrable with $\int f(x) \, \diff \prob(x) \in m(\Theta)$.
\end{rmk}

\begin{rmk}[Non locally compact parameter spaces]
	\label{rmk:nonlocK}
	Proposition~\ref{prop:param:cons} allows for infinite-dimensional parameter spaces $\Theta$. An example would be the space of all copulas of given dimension equipped with a metric that metrizes weak convergence, a space that is still compact thanks in view of Prohorov's theorem. If $\Theta$ is not locally compact, however, then condition~(iii) is too severe and the compact set $K$ should be replaced by its enlargement $K^\delta = \{\theta \in \Theta : \exists \theta' \in K, \rho(\theta, \theta') < \delta \}$ for some sufficiently small $\delta > 0$ \citep[Definition~1.3.7]{vdvw96}. The conditions on the model $\theta \mapsto \prob_\theta$ and on the estimator $\hat{\theta}_n$ should then be modified accordingly. We are grateful to an anonymous Referee for pointing this out.
\end{rmk}

\subsection{Parametric models with group subfamilies}
\label{sec:param:group}

Consider again the testing problem \eqref{eq:hyp:comp}.
Sometimes the unknown parameter can be decomposed as $\theta = (\psi, \eta) \in \Psi \times H = \Theta$, where, for fixed $\psi$, the subfamily $\model_\psi = \{ \prob_{\psi,\eta} : \eta \in H \}$ is a group family as in Section~\ref{sec:group}, generated by a group $G = \{ g_\eta : \eta \in H \}$ of transformations $g_\eta : \Rd \to \Rd$ independent of $\psi$. 
Think for instance of the case where $\psi$ is a vector of shape parameters and $\eta$ a vector of location--scale parameters, with $G$ the group of Example~\ref{ex:locscale}.

Suppose further that a weakly consistent estimator $\hat{\theta}_n = (\hat{\psi}_n, \hat{\eta}_n)$ exists with the following two properties: 
\begin{enumerate}[(i)]
	\item $\hat{\psi}_n$ is invariant under $G$: writing $\hat{\psi}_n = \psi_n(X_1,\ldots,X_n)$, we have
	\begin{equation} 	
	\label{eq:psin:invariant}
	\psi_n(x_1,\ldots,x_n) = \psi_n\bigl(g_\eta(x_1),\ldots,g_\eta(x_n)\bigr) 
	\end{equation}
	for all $\eta \in H$ and all possible samples $(x_i)_{i=1}^n$. 
	\item $\hat{\eta}_n$ is equivariant as in \eqref{eq:equivar}. 
\end{enumerate}
%First, $\hat{\eta}_n$ is equivariant as in \eqref{eq:equivar}. Second, $\hat{\psi}_n$ is invariant under $G$: writing $\hat{\psi}_n = \psi_n(X_1,\ldots,X_n)$, we have
%\begin{equation} 	
%\label{eq:psin:invariant}
%	\psi_n(x_1,\ldots,x_n) = \psi_n\bigl(g_\eta(x_1),\ldots,g_\eta(x_n)\bigr) 
%\end{equation}
%for all $\eta \in H$ and all possible samples $(x_i)_{i=1}^n$. 

Then we propose a hybrid approach: compute the estimated residuals 
\begin{equation} 
\label{eq:Zin:eta}
	\hat{Z}_{i,n} = g_{\hat{\eta}_n}^{-1}(X_i), \qquad i = 1, \ldots, n, 
\end{equation}
and form the test statistic 
\begin{equation}
\label{eq:TMn:eta}
	T_{\model,n} = W_p^p\bigl(
		n^{-1} \textstyle{\sum_{i=1}^n} \delta_{\hat{Z}_{i,n}},
		 \prob_{\hat{\psi}_n, \eta_e}
	\bigr), 
\end{equation}
with $\eta_e \in H$ the parameter yielding the identity transformation $g_{\eta_e}(x) = x$. 
For $\theta = (\psi, \eta)$, the distribution function $t \mapsto F_{n,\theta}(t) = \prob_\theta^n[T_{\model,n} \le t]$ of the test statistic depends on $\psi$ but not on $\eta$.
It can thus be computed as if $\eta = \eta_e$, that is, $F_{n,\psi,\eta} = F_{n,\psi,\eta_e}$. 
The proof of this invariance property relies on (i)--(ii) above.
% and is an extension of the one of Lemma~\ref{lem:group:invariance}. 

The actual $p$-value of $T_{\model,n}$ under $\mathcal{H}_n^0$ is
\[
	1 - F_{n,\psi,\eta_e}(T_{\model,n})
\] 
while the critical value for a test of size $\alpha \in (0, 1)$ is now
\[
	c_{n,\psi}(\alpha) 
	= \inf\{ t \ge 0 : F_{n,\psi,\eta_e}(t) \ge 1 - \alpha \}.
\]
Both are infeasible, however, since the null distribution of $T_{\model,n}$ depends on the unknown $\psi$. We therefore compute $p$-values and critical values under $(\hat{\psi}_n,\eta_e)$. More precisely, we apply the parametric bootstrap. The test thus takes the form
\begin{equation}
\label{eq:phi:group:param}
	\phi_{\model}^n =
	\begin{cases}
		1 & \text{if $1 - F_{n,\hat{\psi}_n,\eta_e}(T_{\model,n}) \le \alpha$ or, equivalently, $T_{\model,n} \ge c_{n,\hat{\psi}_n}(\alpha)$,} \\
		0 & \text{otherwise.}
	\end{cases}
\end{equation}
In practice, the distribution $F_{n,\hat{\psi}_n,\eta_e}$ and the associated critical values $c_{n,\hat{\psi}_n}(\alpha)$ are computed by Monte Carlo approximation, as described in the paragraph following \eqref{eq: thetest2}.
If the dimension of $\psi$ is sufficiently low, we can pre-compute the critical values $c_{n,\psi}(\alpha)$ for $\psi$ on a finite grid $\Psi' \subset \Psi$ and then reconstruct the critical value function by interpolation or smoothing. The reduction from $\theta$ to $\psi$ thus brings a clear computational benefit.

We conjecture that, asymptotically, the test has the correct size under the null hypothesis. The obstacle for the proof is the same as before: required is the asymptotic distribution of the empirical Wasserstein distance, which is a very hard, long-standing open problem. 
In Section~\ref{sec:simu:param}, we provide numerical support for the conjecture by an application to a five-dimensional distribution with separate location-scale parameters for each margin (ten parameters in total) and a single copula parameter.
By exploiting invariance, the computation of the critical value is facilitated as the copula parameter remains as single argument. 

Since the distribution of $T_{\model,n}$ under $\theta = (\psi,\eta)$ is the same as under $\theta_e = (\psi, \eta_e)$, consistency of the hybrid test can be formulated and shown by combining ideas from Propositions~\ref{prop:group} and~\ref{prop:param:cons}. 

\section{Finite-sample performance of GoF tests}
\label{sec:simu}
 
This section is devoted to a numerical assessment of the finite-sample performance of the Wasserstein-based GoF tests introduced in the previous sections. We compare them, whenever possible, with other tests.
The case of a simple null hypothesis (Section~\ref{sec: Test}) is treated in Section~\ref{sec:simu:simple}. 
The performances of various tests for multivariate normality, which is a special case of the hypothesis of a group model in Section~\ref{sec:group}, are compared in Section~\ref{sec:simu:group:Gauss}, along with an illustration involving a Student \emph{t} distribution with known degrees of freedom in Section~\ref{sec:simu:group:t}. 
Section~\ref{sec:simu:param} considers, in line with Section~\ref{sec:param:group}, the more general composite null hypothesis of a parametric family indexed by marginal location and scale parameters along with a copula parameter.
Numerical results support the conjecture of the (asymptotic) validity of the parametric bootstrap for calculating critical values.
To the best of our knowledge, no GoF test is available in the literature for such cases except for the method described by \citet{khmaladze2016unitary}, the numerical implementation of which, however, remains unsettled.

%Results for a goodness-of-fit test for a parametric family constructed by joining parametric models for margins and copula model are presented in Section~\ref{sec:simu:param}. There, 
Throughout, we consider the Wasserstein distances of order $p \in \{1, 2\}$. The level~$\alpha$ of the tests is set to 5\%, the sample size is $n = 200$, and the number of replicates considered in the estimation of power curves is~$1\,000$. As mentioned in Section~\ref{sec: computations} we relied on the \textsf{R} package \textsf{transport} \citep{transport} in case $p = 2$ and $d = 2$ and on our own \textsf{C} implementation of the algorithm proposed by \citet{genevay2016stochastic} in all other cases.
See also Appendix~\ref{app:algo} for some details on the calculation of the critical values.
 
 % --------------------------------
\subsection{Simple null hypotheses}
\label{sec:simu:simple}

The setting is as in Section~\ref{sec: Test}: given an independent random sample $X_1, \ldots, X_n$ from some unknown $\prob \in \Prob(\Rd)$, we consider testing the simple null hypothesis~$\mathcal{H}_0^n : \prob = \prob_0$, where $\prob_0 \in \Prob_p(\Rd)$ is fully specified.

Two other goodness-of-fit tests will be used as benchmarks: the test by \citet*{rippl2016limit}, which is based on the $2$-Wasserstein distance and is specific for multivariate Gaussian distributions, and the adaptation of the Kolmogorov--Smirnov test by \citet{khmaladze2016unitary}, which is based on empirical process theory. Both tests are described in some detail in Appendix~\ref{sec:simu:simple:other}.

\subsubsection{Bivariate Gaussian distribution}
\label{sec:simu:simple:bivGauss}

In Figure~\ref{fig: GausNull}, we assess the performance of the GoF tests of $\mathcal{H}_0^n : \prob = \prob_0$ where~$\prob_0 = \normal_2(0, I_2)$ is a centered bivariate Gaussian with identity covariance matrix. The alternatives $\prob$ in  panels (a)--(f) are as follows:
\begin{enumerate}[(a)]
	\item $\prob = \normal_2\bigl(\bigl(\begin{smallmatrix} \mu \\ \mu \end{smallmatrix}\bigr), I_2\bigr)$ with location shift  $\mu$ along the main diagonal (rejection frequencies plotted against $\mu \in \mathbb{R}$);
	\item $\prob = \normal_2(0, \sigma^2 I_2)$ (rejection frequencies   plotted against  $\sigma^2 > 0$);
	\item $\prob = \normal_2\bigl(0, \bigl(\begin{smallmatrix} 1 & \rho \\ \rho & 1 \end{smallmatrix}\bigr)\bigr)$ with   correlation   $\rho$  (rejection frequencies   plotted against  $\rho \in (-1, 1)$);
	\item $\prob$ has standard normal margins but Gumbel copula with parameter $\theta$   (rejection frequencies   plotted against    $\theta \in [1, \infty)$);
	\item $\prob$ has standard Gaussian margins but a bivariate Student \emph{t} copula with~$\nu = 4$ degrees of freedom and correlation parameter $\rho$   (rejection frequencies   plotted against  $\rho \in (-1, 1)$);\footnote{Note that $\prob$ is not Gaussian, even for $\rho = 0$.}
	\item $\prob$ is the ``boomerang-shaped" Gaussian mixture %of $\prob_0$ and  mixture
	 described in Appendix~\ref{sec: bana}  (rejection frequencies   plotted against  the mixing weight $p \in (-1, 1)$).\footnote{The mixture is constructed so that the first and second moments of $\prob$ remain close to those of $\prob_0$.}
\end{enumerate}
The Gumbel and Student \emph{t} copula simulations  in (d) and (e) were implemented from the \textsf{R} package \textsf{copula} \citep{copula}.

\begin{figure}
\begin{center}
\begin{tabular}{@{}c@{}c@{}c}
\includegraphics[width=0.33\textwidth, height=65mm]{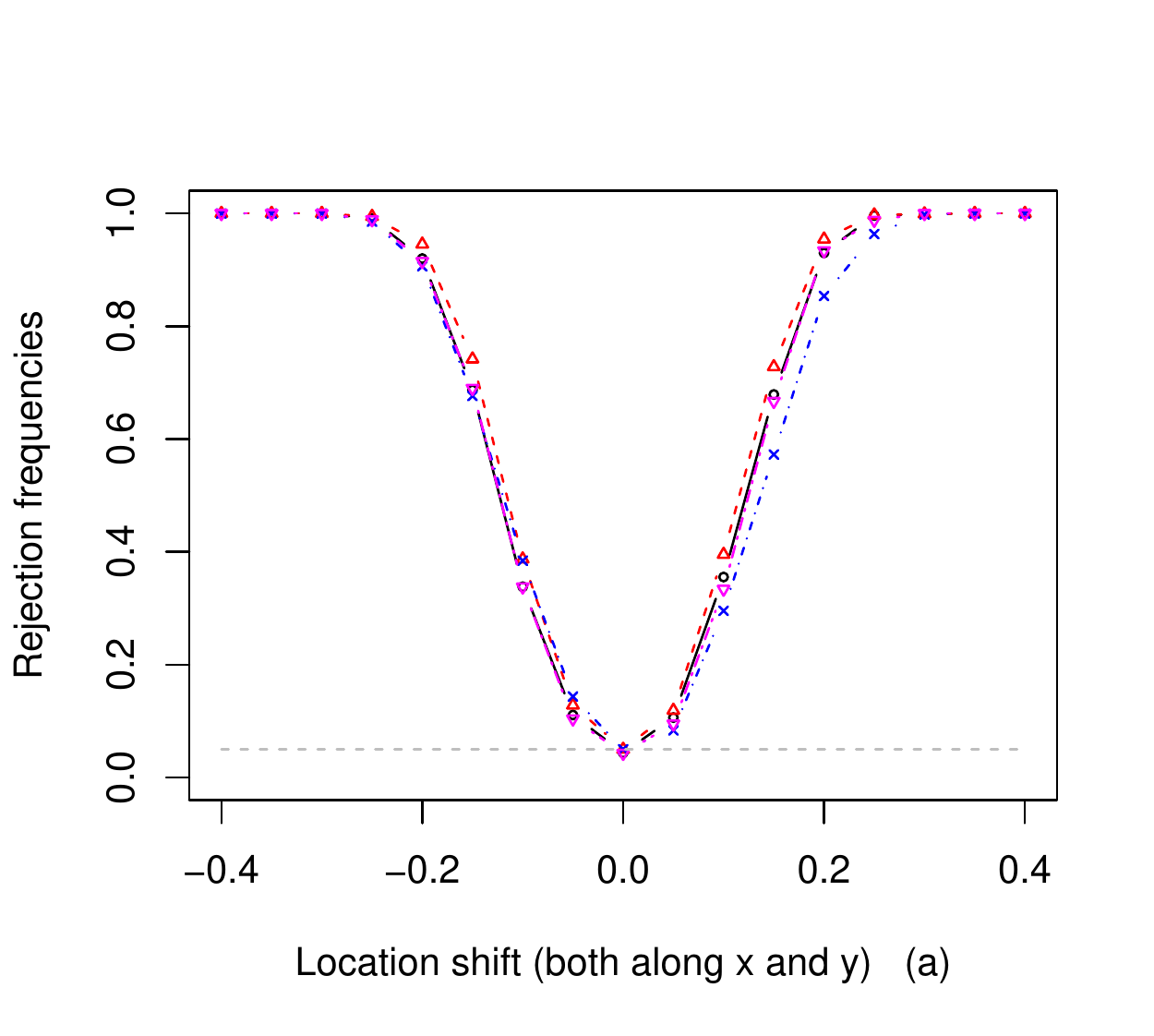}& %Reminder: save pictures as 480x430 
\includegraphics[width=0.33\textwidth, height=65mm]{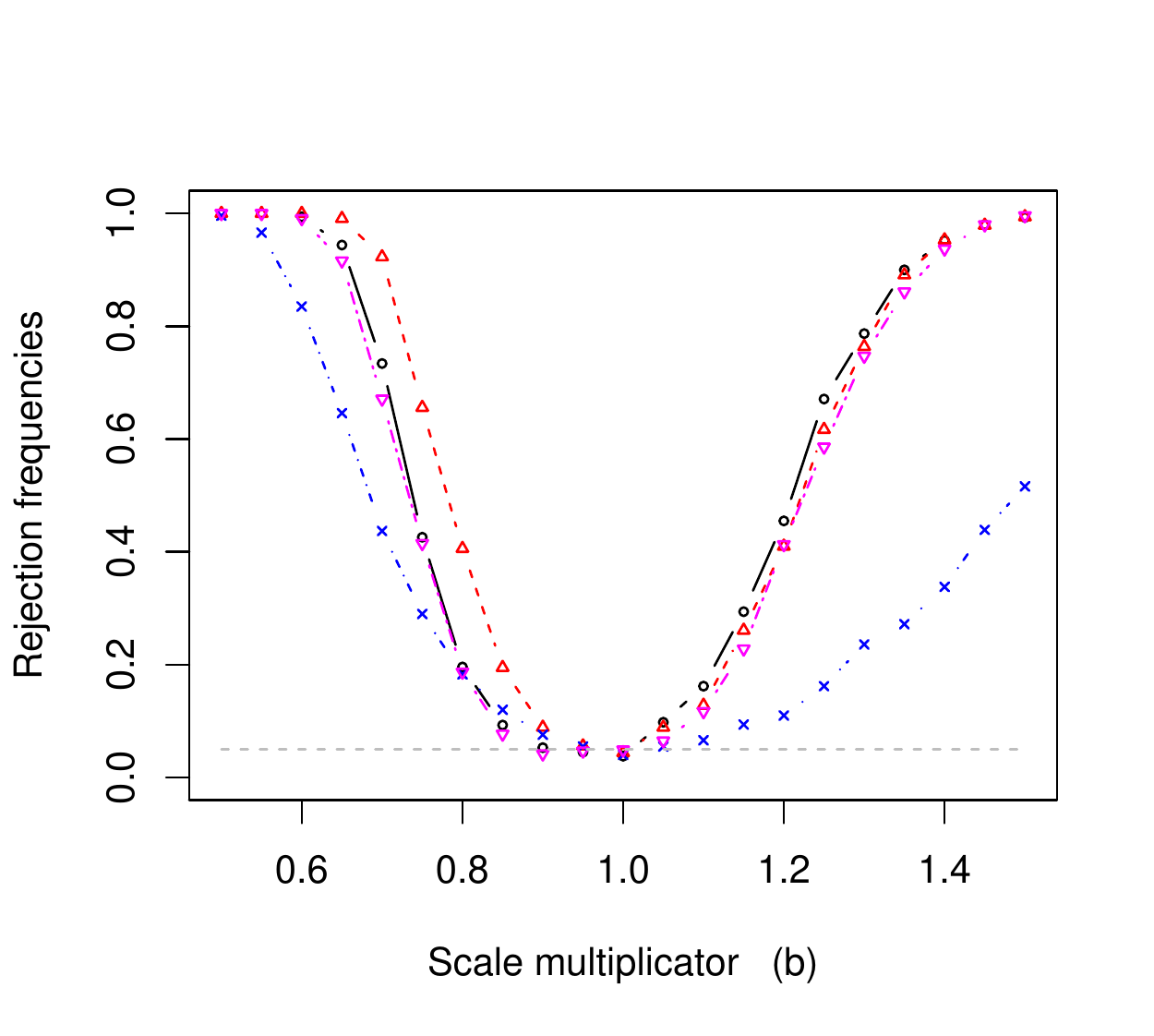}&
\includegraphics[width=0.33\textwidth, height=65mm]{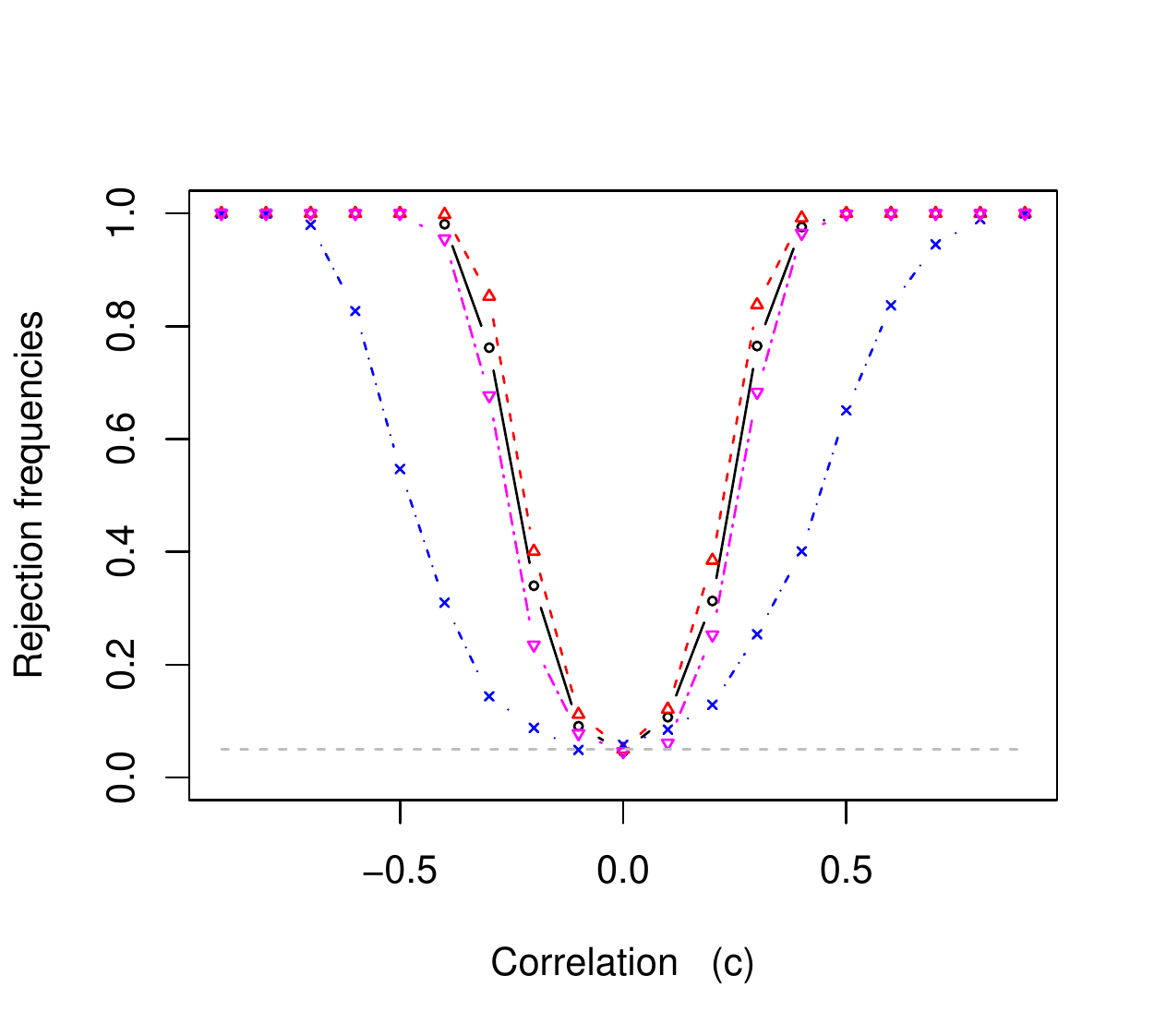}\vspace{-4mm}\\

\includegraphics[width=0.33\textwidth, height=65mm]{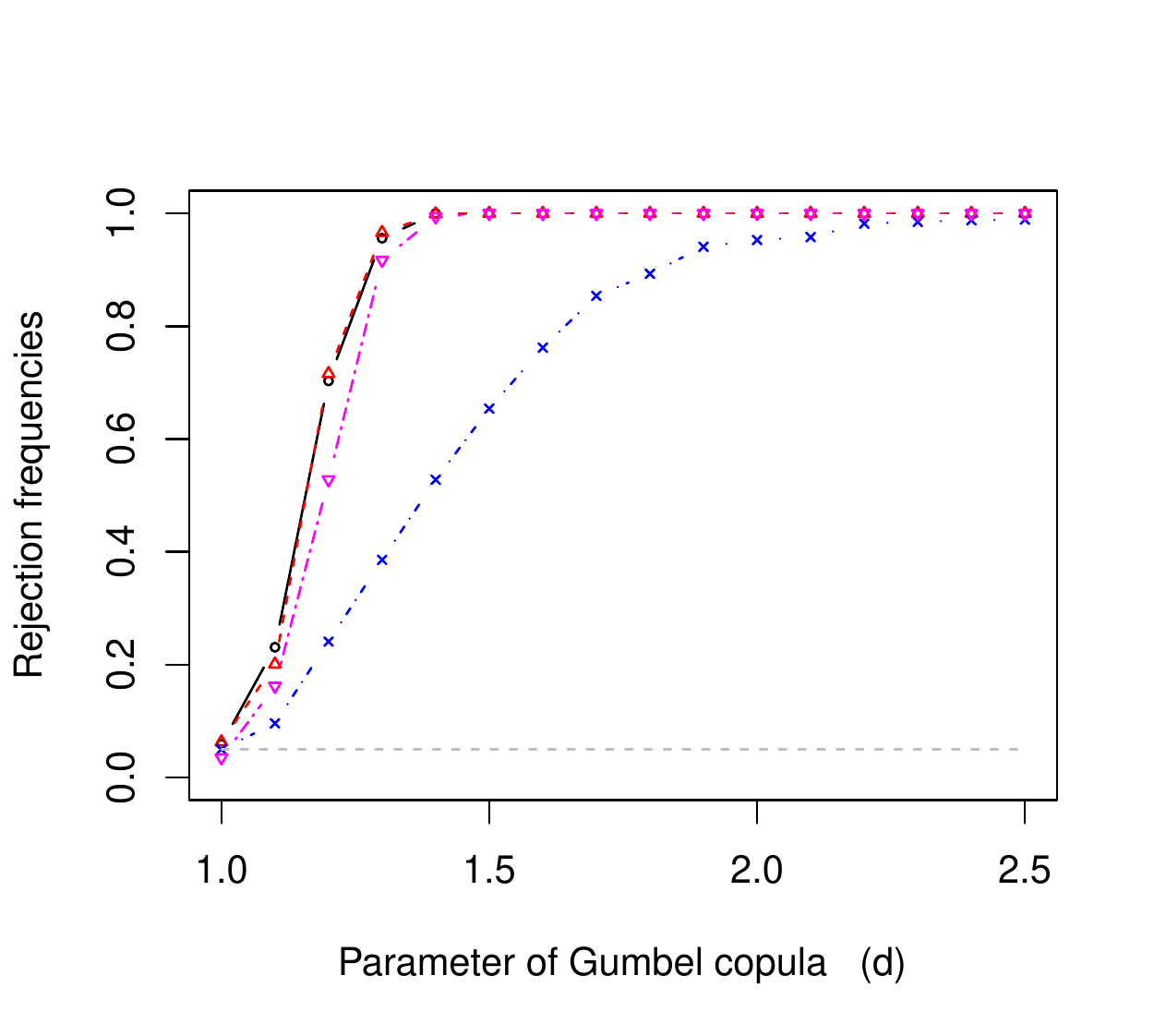}&
\includegraphics[width=0.33\textwidth, height=65mm]{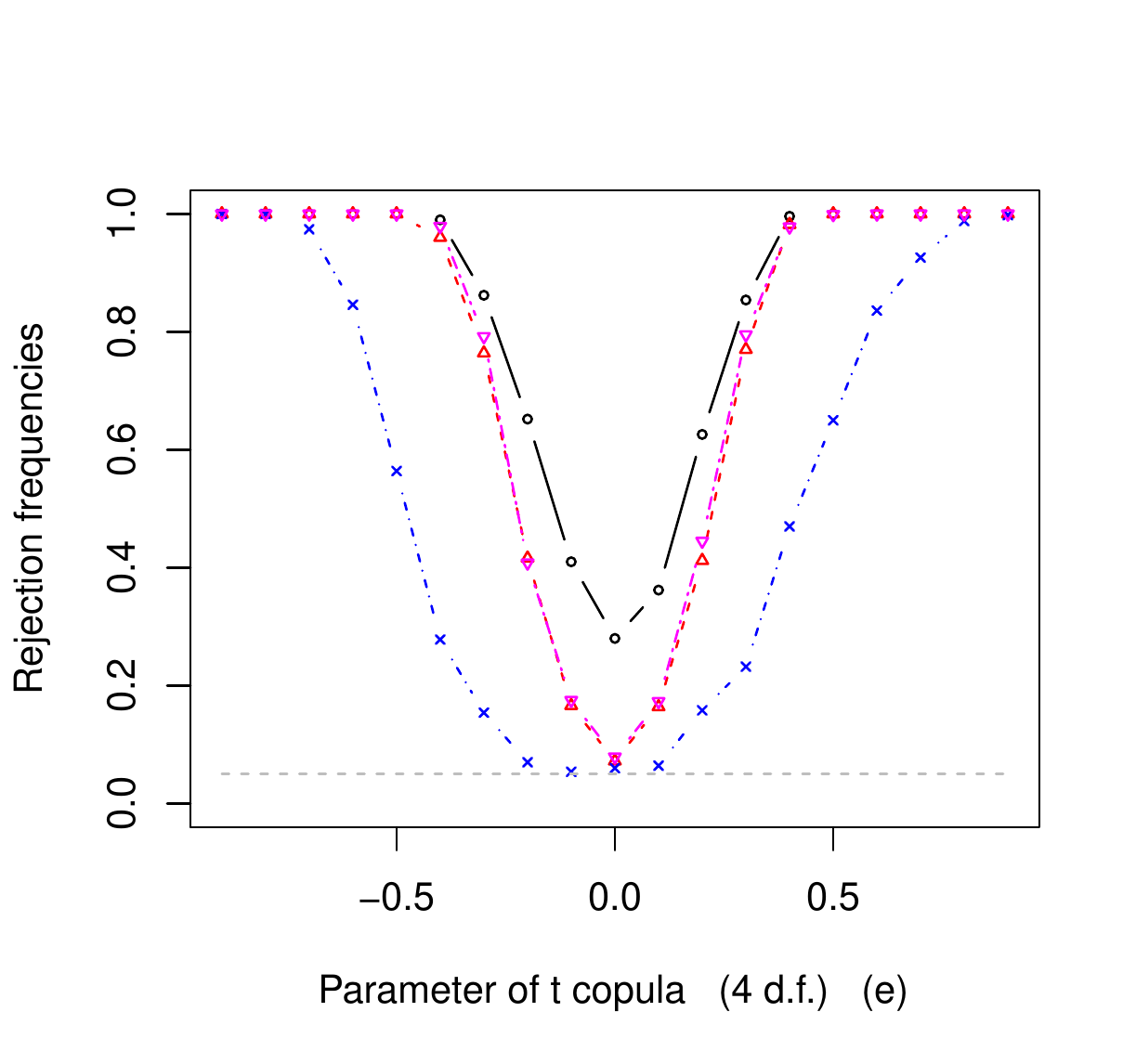}&%
\includegraphics[width=0.33\textwidth, height=65mm]{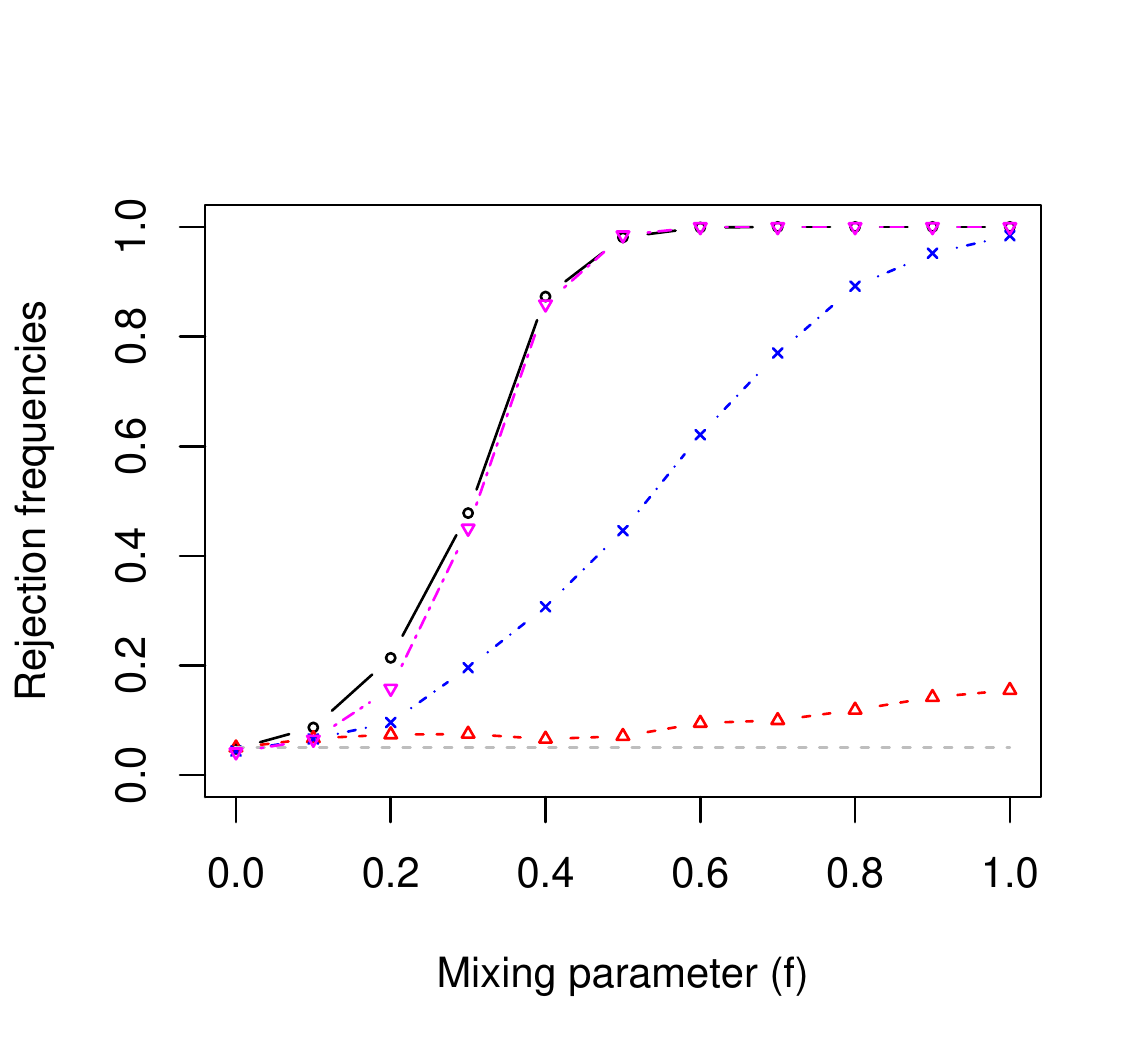}\\ 

\multicolumn{3}{ c }{ 
\includegraphics[width=0.7\textwidth]{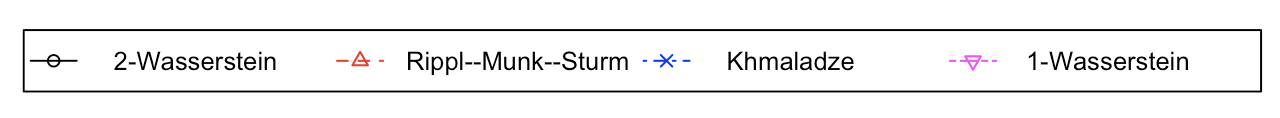}
\vspace{-4mm}}
\end{tabular}
\end{center}
\caption{\label{fig: GausNull} \slshape\small Empirical powers of various GoF tests for the simple Gaussian null hypothesis~$\mathcal{H}_0^n : \prob = \normal_2(0, I_2)$. Four tests are considered: the 2- and 1-Wasserstein distances (Section~\ref{sec: Test}), the Rippl--Munk--Sturm test  \citep{rippl2016limit}, and the Khmaladze Kolmogorov--Smirnov type   test \citep{khmaladze2016unitary}, see Section~\ref{sec:simu:simple:other}. The alternatives~$\prob$ in panels (a)--(f) are described in Section~\ref{sec:simu:simple:bivGauss}. Note that in (e), $\prob$ is not Gaussian even when $\rho = 0$.}
\end{figure}

%One can see from the various settings that the power of the test based on the supremum of the empirical process is often smaller than the power of the tests based on Wasserstein distance. Nor the test by Rippl--Munk--Sturm nor the test based on Wasserstein distance is clearly outperforming the other test. Indeed, the test by Rippl--Munk--Sturm exhibits a better power when the scatter parameter changes but is, as one could expect from the formula of the test, not as sensitive as the Wasserstein distance in the setting with the Student copula.

Inspection of Figure~\ref{fig: GausNull} indicates that the Khmaladze test, as a rule,  is uniformly outperformed by the Rippl--Munk--Sturm and Wasserstein tests. The Rippl--Munk--Sturm test, of course, does relatively well under the Gaussian alternatives of panels (a)--(c) where, however, the Wasserstein test is almost as powerful (while its validity, contrary to that of  the Rippl--Munk--Sturm test, extends largely beyond the Gaussian null hypothesis). Against the non-Gaussian alternatives in panels (d)--(f), the Wasserstein test has  higher power than the Rippl--Munk--Sturm and Khmaladze tests, with the exception of the Gumbel copula alternative in panel (d), where the Rippl--Munk--Sturm and Wasserstein tests perform equally well. For the ``boomerang mixture"  of  panel (f), the Rippl--Munk--Sturm test fails to capture the change in distribution. There is little difference between the Wasserstein tests with $p = 1$ and $p = 2$, except for the $t$-copula, where $p = 2$ yields a more powerful test than $p = 1$.

% \subsubsection{Uniform distribution}
%We now turn to tests where the null hypothesis is a uniform over $[0,1]^2$. We test alternatives where the margins of the distribution remains the same but the copula is changed.
%
%\begin{figure}
%   \includegraphics[scale=0.65]{UnifTcop.eps}
%   \caption{\label{fig: 2D Tcopula}  \emph{\scriptsize{The null hypothesis is the uniform over the unit square. Under the alternatives considered, the copula is a t copula with 4 d.f.\ whose correlation parameter varies and the marginals are uniform. Note that when the parameter equals zero, the copula is not the independent copula.}}}
%\end{figure}
%\begin{figure}
%   \includegraphics[scale=0.65]{UnifFrank.eps}
%   \caption{\label{fig: 2D Tcopula}  \emph{\scriptsize{The null hypothesis is the uniform over the unit square. Under the alternatives considered, the copula is a Frank copula whose parameter varies and the marginals are uniform.}}}
%\end{figure}
% One can see that the test based on Wasserstein distance clearly outperforms the one based on the empirical process. However, the power of the test based on the empirical Wasserstein distance is smaller than the one based on the empirical copula. 

%\subsubsection{Non-Gaussian null hypotheses}
\subsubsection{Mixture of bivariate Gaussian distributions}
\label{sec:simu:simple:mixGauss}

Figures~\ref{fig: MixtNull} to~\ref{fig: D5Tmarg} concern non-Gaussian null distributions $\prob_0$, so that the Rippl--Munk--Sturm test no longer applies. In Figure~\ref{fig: MixtNull}, the null distribution is the Gaussian mixture
 $
	\prob_0 
	= 
	0.5 \, \normal_2(0, I_2) + 
	0.5 \, \normal_2\left(\big(\begin{smallmatrix} 3 \\ 0 \end{smallmatrix}\big), I_2\right). 
$ 
The alternatives in both panels are as follows:
\begin{enumerate}[(a)]
	\item $\prob = 0.5 \, \normal_2(0, I_2) + 0.5 \, \normal_2\left(\big( \begin{smallmatrix}3+\delta \\  0\end{smallmatrix}\big), I\right)$ (rejection frequencies plotted against the   location shift $\delta \in \reals$);
	\item $\prob_0 = \lambda \, \normal_2(0, I_2) + (1-\lambda) \, \normal_2\left(\big(\begin{smallmatrix} 3 \\ 0 \end{smallmatrix}\big), I_2\right)$  (rejection frequencies plotted against the  mixing weight $\lambda \in [0, 1]$). 
\end{enumerate}
Both Wasserstein tests have higher power than the \citet{khmaladze2016unitary} test.

\begin{figure}
	\begin{center}
		\begin{tabular}{@{}c@{}c}
			\includegraphics[width=0.4\textwidth, height=65mm]{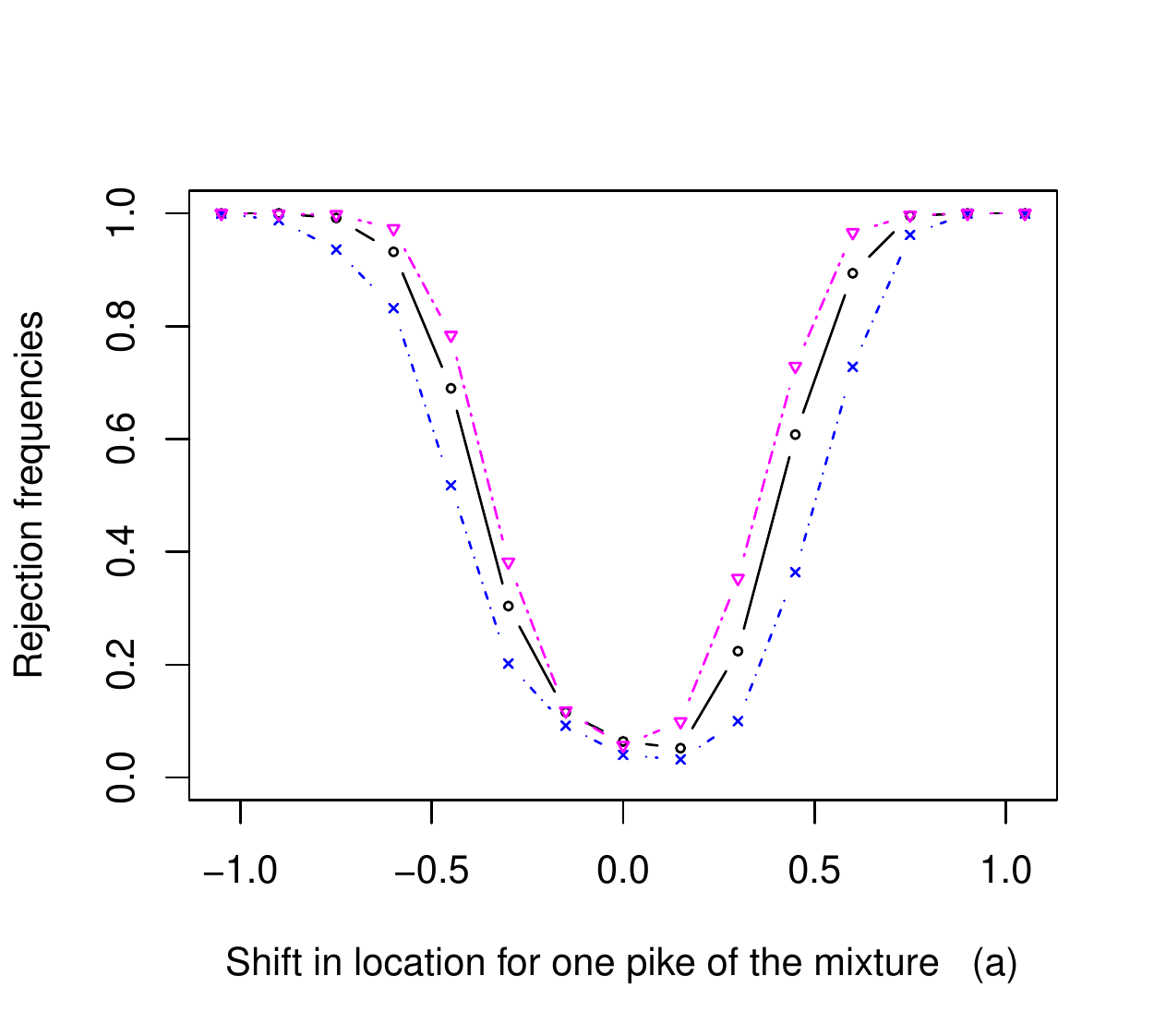}&%
			\includegraphics[width=0.4\textwidth, height=65mm]{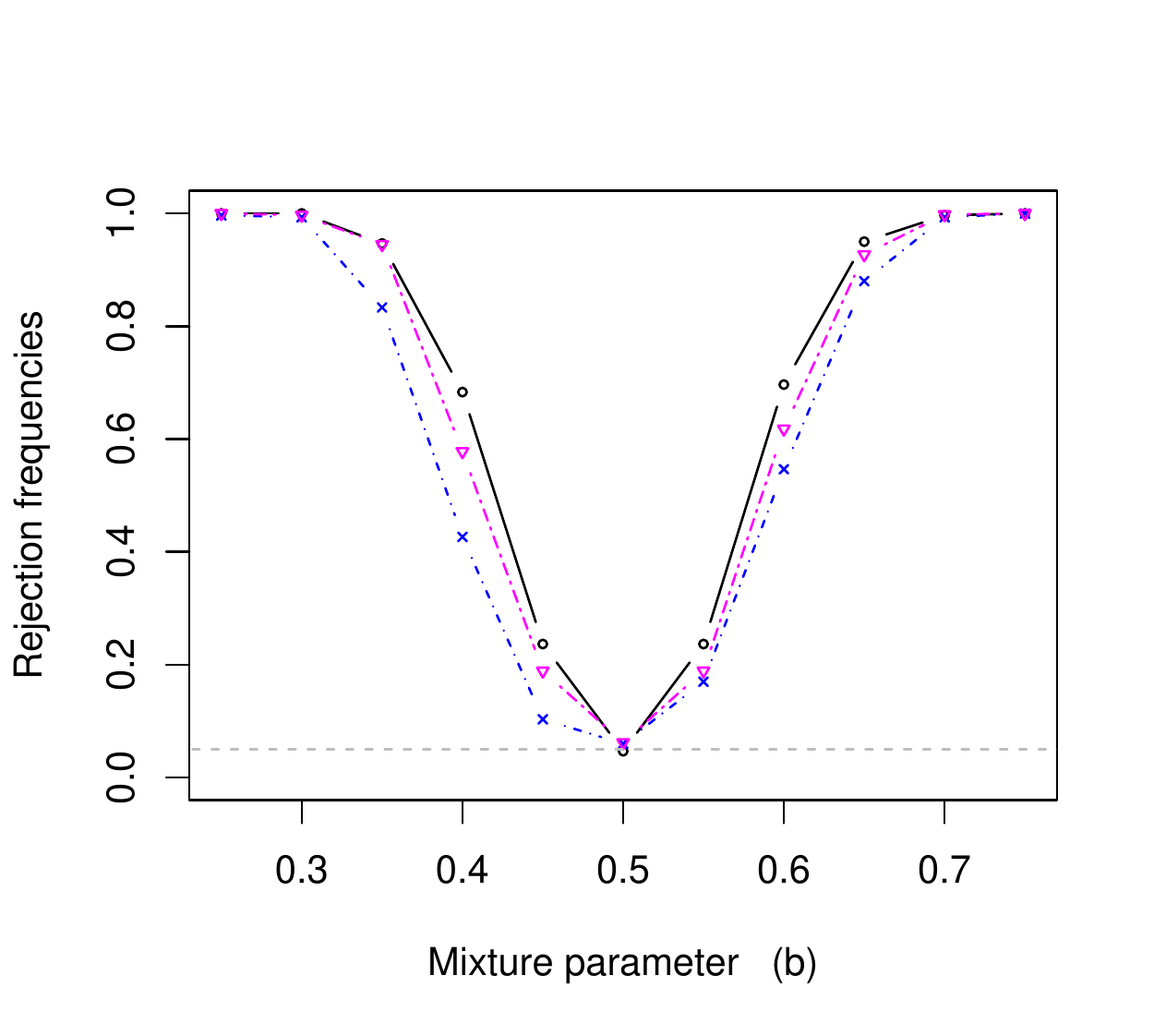}\\
			\multicolumn{2}{ c }{ 
				\includegraphics[width=0.5\textwidth]{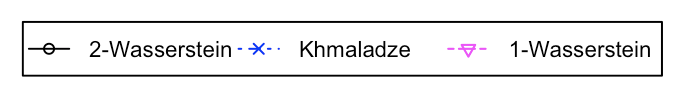}
				\vspace{-4.5mm}}
		\end{tabular}
	\end{center}
	\caption[]{\label{fig: MixtNull} \slshape\small Empirical powers of the Wasserstein and \citet{khmaladze2016unitary}  tests   for the simple null hypothesis $\mathcal{H}_0^n : \prob = \prob_0$ with $\prob_0$ an equal-weights mixture of $\normal_2(0, I_2) $ and~$\normal_2\left(\big(\begin{smallmatrix} 3 \\ 0 \end{smallmatrix}\big), I_2\right)$. The alternatives in (a)--(b) are described in Section~\ref{sec:simu:simple:mixGauss}.}
%	In panel (a), the alternative $\prob\vspace{1mm}$ is an equal-weights mixture of $\normal_2(0, I_2) $ and $\normal_2\left(\big(\begin{smallmatrix}3+\delta \\  0\end{smallmatrix}\big), I_2\right)$; rejection frequencies are plotted against~$\delta\in[-1,1]$. In panel (b), the alternative $\prob$ is a mixture of the same two components, but with  weights $\lambda\in(0,1)$ and $(1-\lambda)$;  rejection frequencies are plotted against~$\lambda \in [0.25, 0.75]$.}
\end{figure}

\subsubsection{Gumbel copula and Gaussian marginals}
\label{sec:simu:simple:gumbelGauss}

In Figure~\ref{fig: GausGumb}, $\prob_0$ has standard Gaussian margins and a Gumbel copula with parameter~$\theta = 1.7$. 
The alternative $\prob$ is of the same form but with another value~$\theta\neq 1.7$ of the copula parameter $\theta \in [1, \infty)$. 
Again, the Wasserstein tests at $p \in \{1, 2\}$ have quite comparable performance and yield higher empirical powers in most cases. 
  
\begin{figure}
\begin{center}
\includegraphics[width=0.7\textwidth, height=70mm]{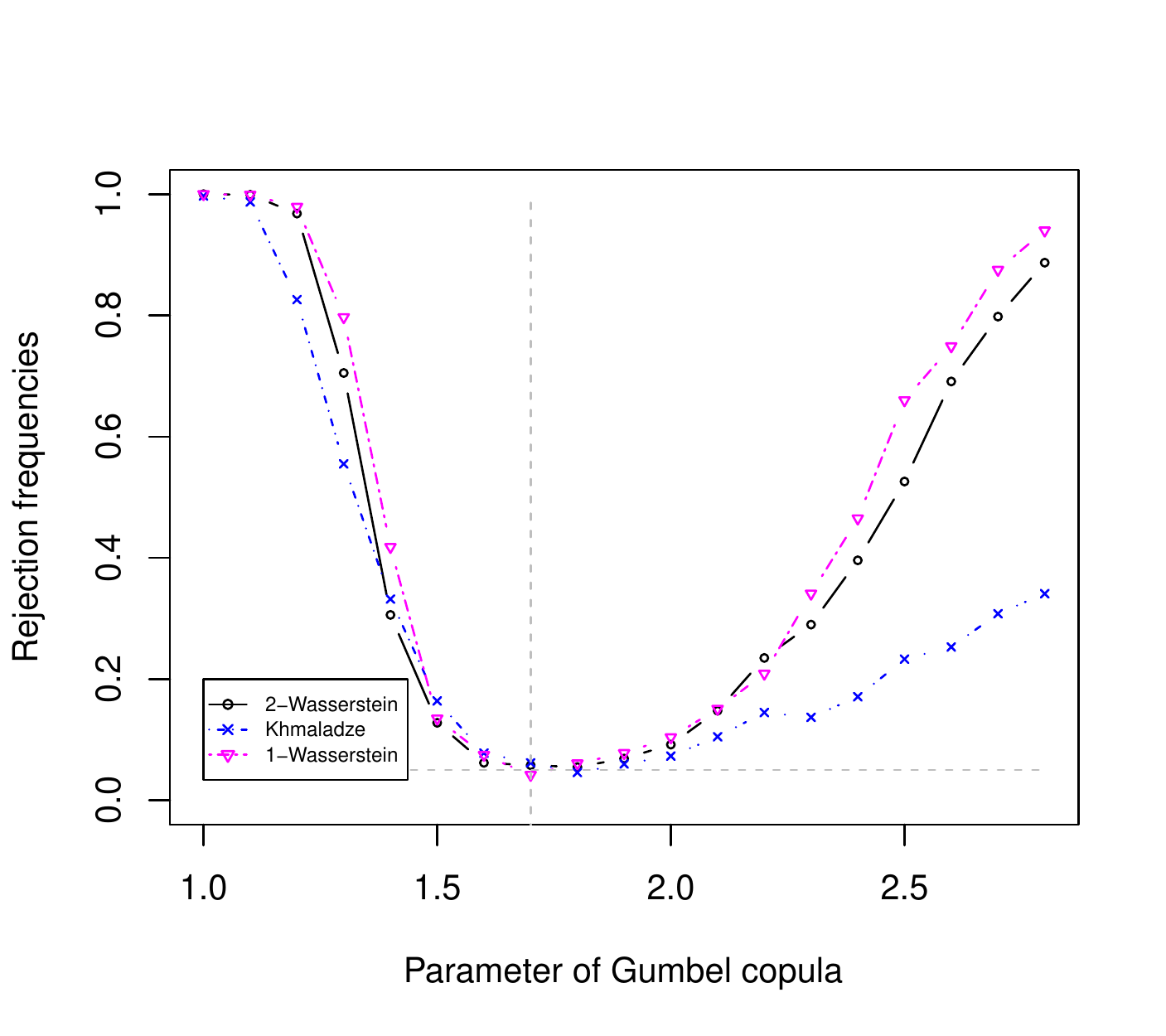}\vspace{-8mm}
\end{center}
\caption{\label{fig: GausGumb} \slshape\small Empirical powers of the Wasserstein and \citet{khmaladze2016unitary} tests for the simple null hypothesis $\mathcal{H}_0^n : \prob = \prob_0$ with $\prob_0$ a bivariate distribution with standard Gaussian margins and Gumbel copula with parameter $\theta = 1.7$. Rejection frequencies are plotted against the alternative copula parameter~$\theta$. 
%\js{Legend: consistent naming of Khmaladze's test; increase font size. Perhaps add horizontal line at $\alpha = 0.05$ and vertical line at $\theta = 1.7$ (in gray, dashed, \ldots).}
}
\end{figure}

\subsubsection{A five-dimensional Student $t$ distribution}
\label{sec:simu:simple:5t}

Let us turn now to a higher-dimensional case. 
In Figure~\ref{fig: D5Tmarg}, we test for the null hypothesis $\mathcal{H}_0^n : \prob = \prob_0$ with $\prob_0 = \otimes_{i=1}^5 t_{25}$ a five-dimensional distribution with independence copula and Student $t_{25}$ margins. The following alternatives are considered:
\begin{enumerate}[(a)]
\item A distribution with independence copula and Student $t_\nu$ margins. The rejection frequencies are plotted against $\nu$.
\item A distribution with independence copula and margins equal to the Student $t_{25}$ distribution shifted by $\mu$. The rejection frequencies are plotted against $\mu$.
\item A distribution $t_{25} \otimes t_{25} \otimes t_{25} \otimes t_{25,\delta}$, where $t_{\nu,\delta}$ is the bivariate Student $t$ distribution with $\nu$ degrees of freedom and dependence parameter $\delta$. Note that $\delta = 0$ does not correspond to the null hypothesis. The rejection frequencies are plotted against $\delta$.
\end{enumerate}
The Khmaladze Kolmogorov--Smirnov test is most sensitive to the change in location (b), although the two Wasserstein tests perform quite well too. For the other two alternatives, the Wasserstein tests have much higher power than the Khmaladze test. For the Wasserstein test, there is little difference between $p = 1$ and $p = 2$, except for case~(c), in which the choice of $p = 2$ yields higher power.

\begin{figure}
\begin{center}
\begin{tabular}{@{}c@{}c@{}c}
\includegraphics[width=0.33\textwidth, height=70mm]{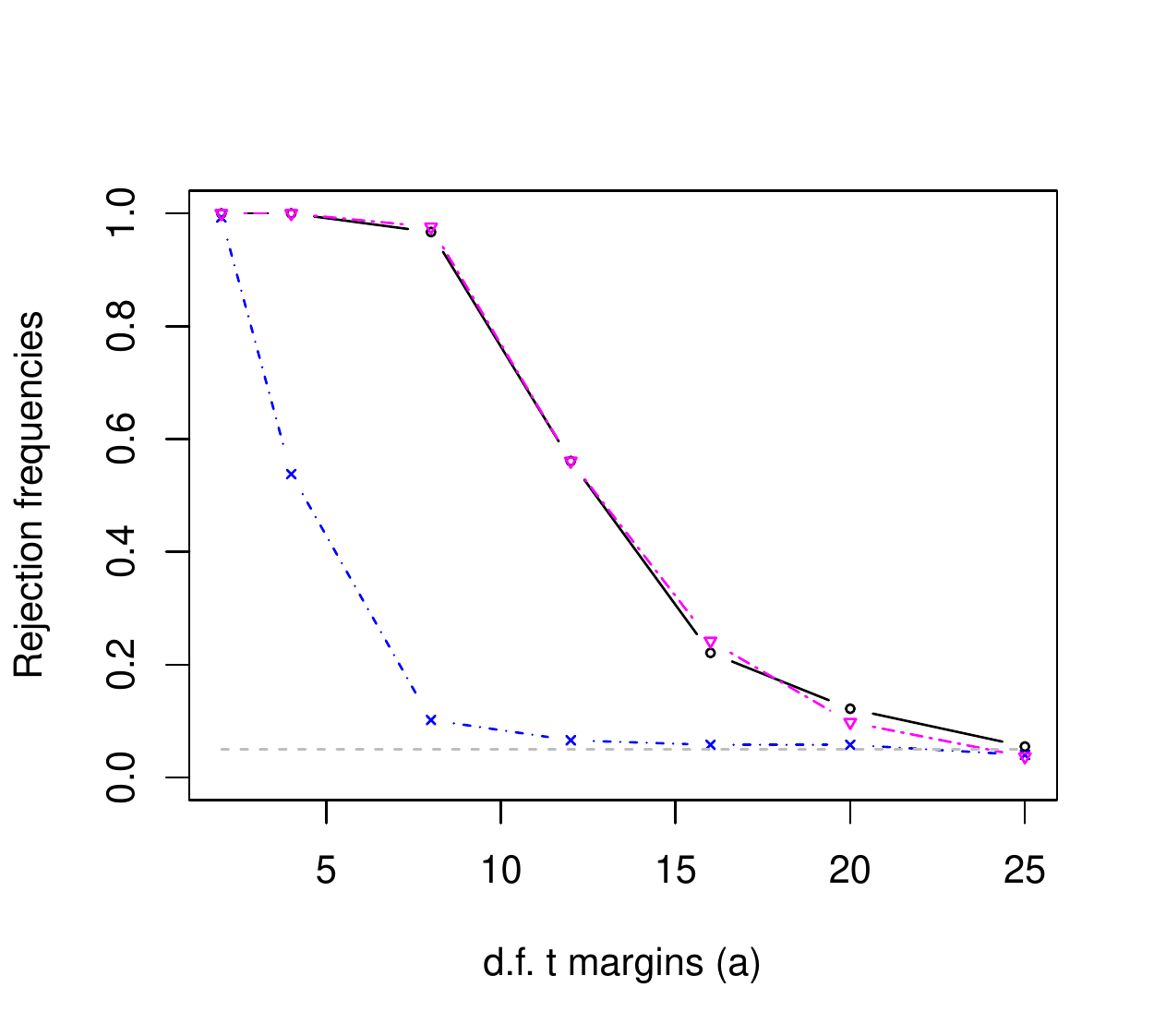}&
\includegraphics[width=0.33\textwidth, height=70mm]{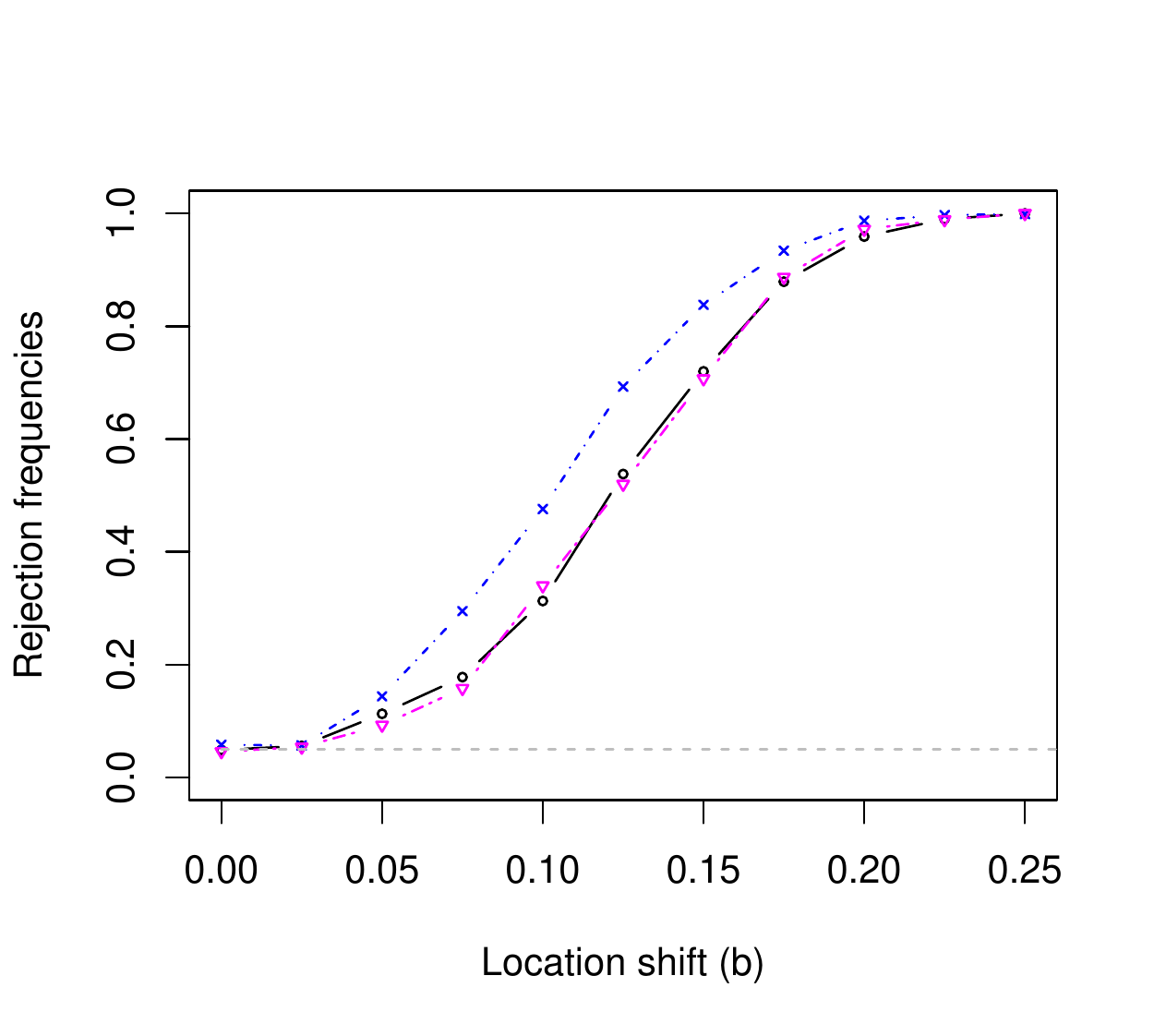} &
\includegraphics[width=0.33\textwidth, height=70mm]{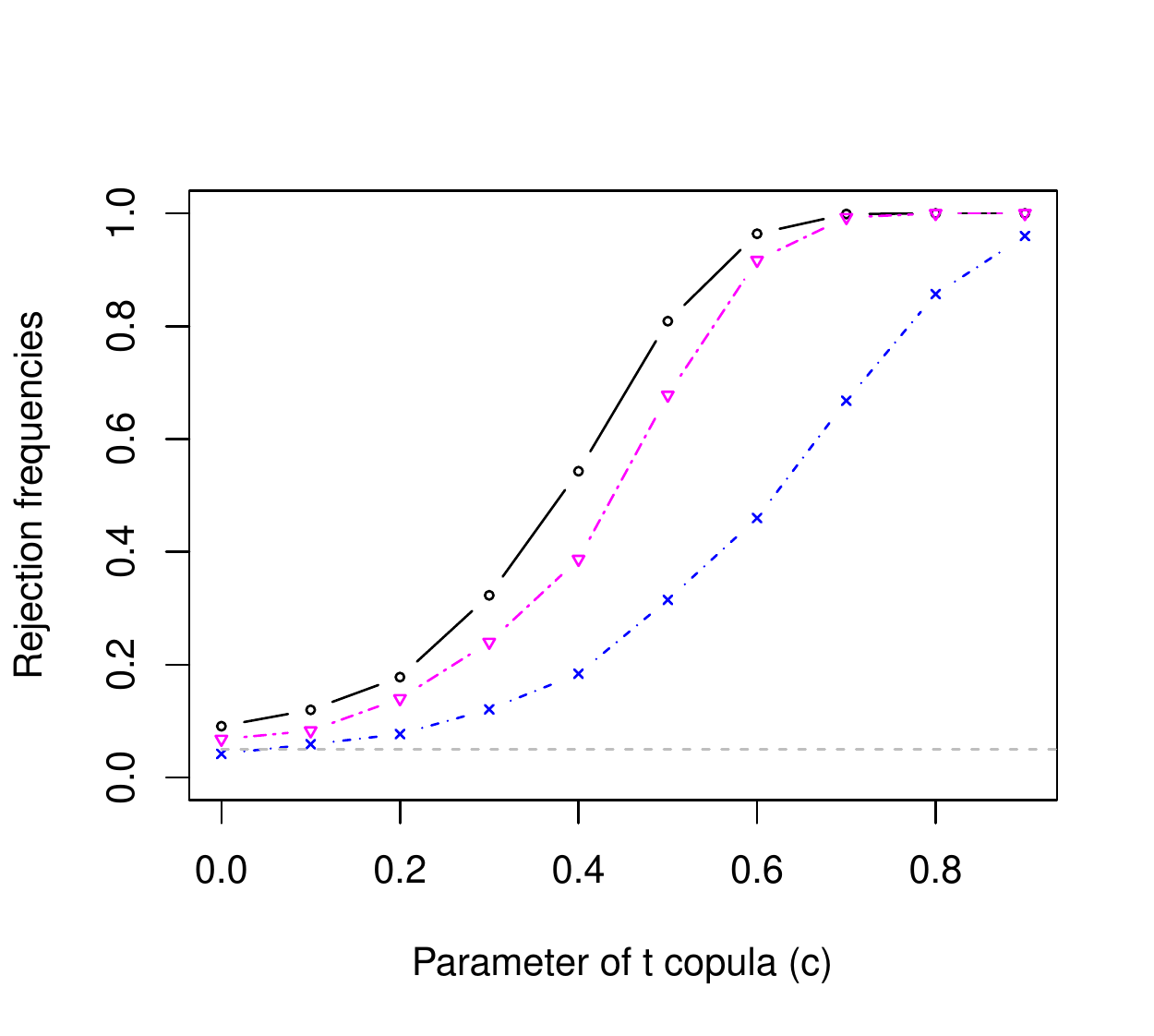}\vspace{-3mm} \\
\multicolumn{3}{ c }{ 
\includegraphics[width=0.5\textwidth]{Legend2B.png}
\vspace{-3mm}}
\end{tabular}
\end{center}
\caption{\label{fig: D5Tmarg} \slshape\small Empirical powers of various GoF tests for the simple hypothesis $\mathcal{H}_0^n$ that $\prob \in \Prob_2(\reals^5)$ is the five-fold product of the Student $t_{25}$ distribution with itself. Shown are rejection frequencies obtained via the Wasserstein test of order $p \in \{1, 2\}$ and the Khmaladze Kolmogorov--Smirnov one.}
	The alternatives (a)--(c) are explained in Section~\ref{sec:simu:simple:5t}. In~(c), no setting corresponds to the null hypothesis.
\end{figure}

%-----------------------------------------------
\subsection{Elliptical families as group models}
\label{sec:simu:group}

Elliptical distributions arise as group models for the group of affine transformations, see Example~\ref{ex:affine} in Section~\ref{sec:group}. 
Two notable examples are the multivariate normal family and the multivariate Student $t$ distribution with a fixed number of degrees of freedom. 
We assess the finite-sample performance of the Wasserstein test in~\eqref{eq:phi:group} for $p \in \{1, 2\}$ with residuals computed by
\[ 
	\hat{Z}_{n,i} = \hat{L}_{n}^{-1} (X_i - \hat{\mu}_n),
	\qquad i = 1, \ldots, n, 
\]
with $\hat{\mu}_n$ the sample mean vector and $\hat{L}_n$ the lower Cholesky triangle of the empirical covariance matrix of $X_1,\ldots,X_n$.

\subsubsection{Testing for multivariate normality} 
\label{sec:simu:group:Gauss}

Testing for multivariate normality is a well-studied problem for which many tests have been put forward. 
As benchmarks, we will consider here the tests proposed in \citet{royston1983}, \citet{henze1990class}, and \citet{rizzo2016energy}. 
Royston's test is a generalisation of the well-known Shapiro--Wilks test. It only tests whether the margins are Gaussian and ignores the dependence structure.
The Henze--Zirkler test statistic is an integrated weighted squared distance between the characteristic function of the multivariate standard normal distribution and the empirical characteristic function of the empirically standardized data.
Interestingly, \citet{ramdas2017wasserstein} showed that the Wasserstein distance and the energy distance of \citet{rizzo2016energy} are connected, as the so-called entropy-penalized Wasserstein distance interpolates between them two. 
We borrowed the implementation of these tests from the \textsf{R} package \textsf{MVN} \citep{MVN}. 
The test by \citet{rippl2016limit} considered in Section~\ref{sec:simu:simple} does not apply, since it only can handle fully specified Gaussian distributions, while here, the mean vector and covariance matrix are unknown.

In Figure~\ref{fig: GausFam}, we consider dimensions $d = 2$ [top row, panels (a) and (b)] and $d = 5$ [bottom row, panels (c) and (d)].
Here are the alternative distributions~$\prob$:
\begin{enumerate}[(a)]
	\item A bivariate distribution with standard normal margins and a bivariate Gumbel copula with parameter~$\psi \in [1, \infty)$. Rejection frequencies are plotted against $\psi \in [1, \infty)$. 
	\item A bivariate distribution with independent margins, one of which is standard normal while the other one is Student \emph{t} with $\nu > 0$ degrees of freedom. Rejection frequencies are plotted against $\nu>0$.
	\item A five-variate distribution given by $\mathcal{N}_3(0,I_3) \otimes \mathcal{D} $, where $\mathcal{D}$ is a bivariate distribution with Gumbel copula indexed by a parameter $\psi \in [1, \infty)$ and with standard normal margins. Rejection frequencies are plotted against $\psi$.
	\item A five-variate distribution with independent margins, all of which are Student~$t_\nu$. Rejection frequencies are plotted against the common parameter~$\nu$.
\end{enumerate}

\begin{figure}
\begin{center}
\begin{tabular}{@{}c@{}c}
\includegraphics[width=0.5\textwidth, height=70mm]{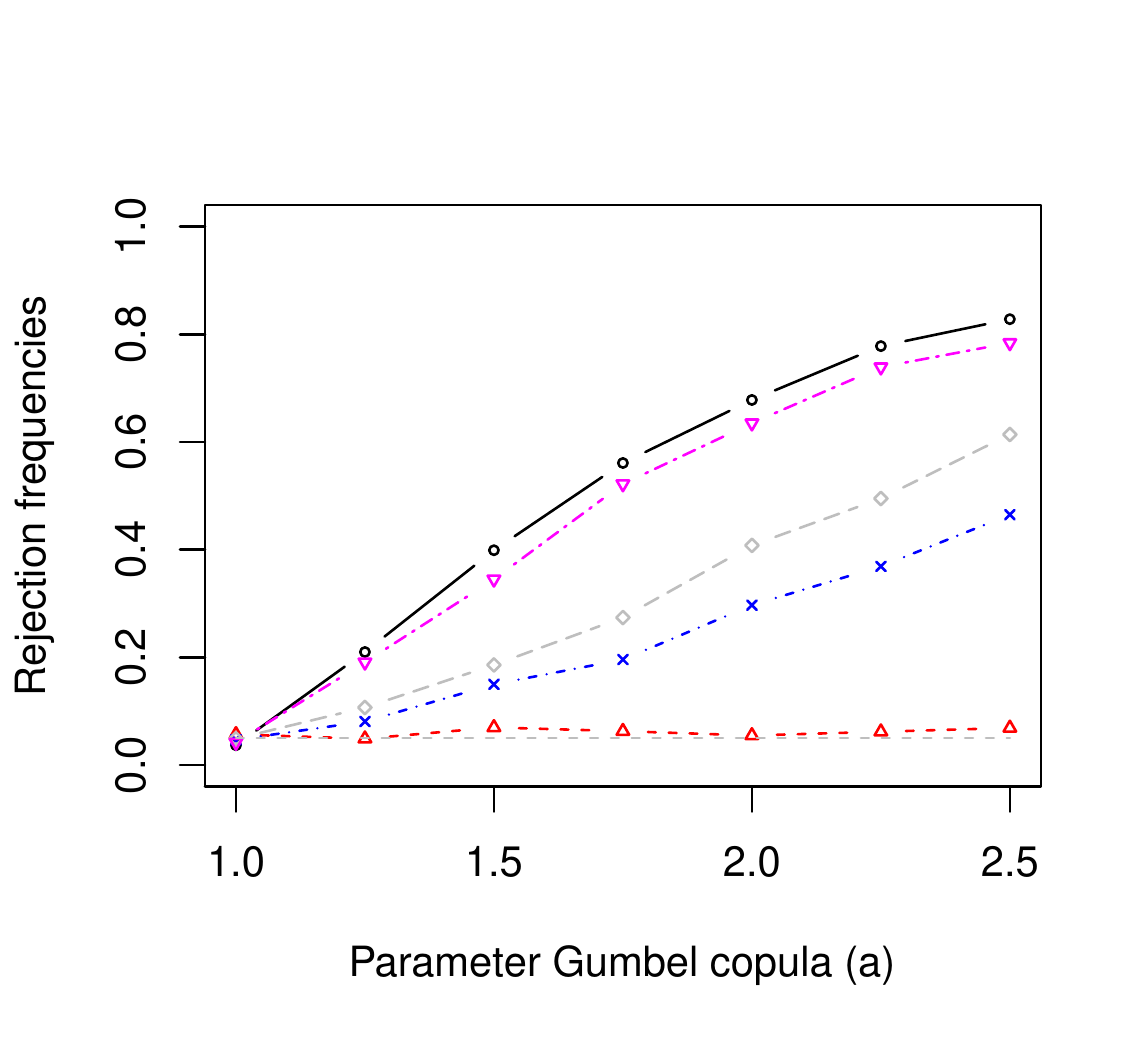}&%
\includegraphics[width=0.5\textwidth, height=70mm]{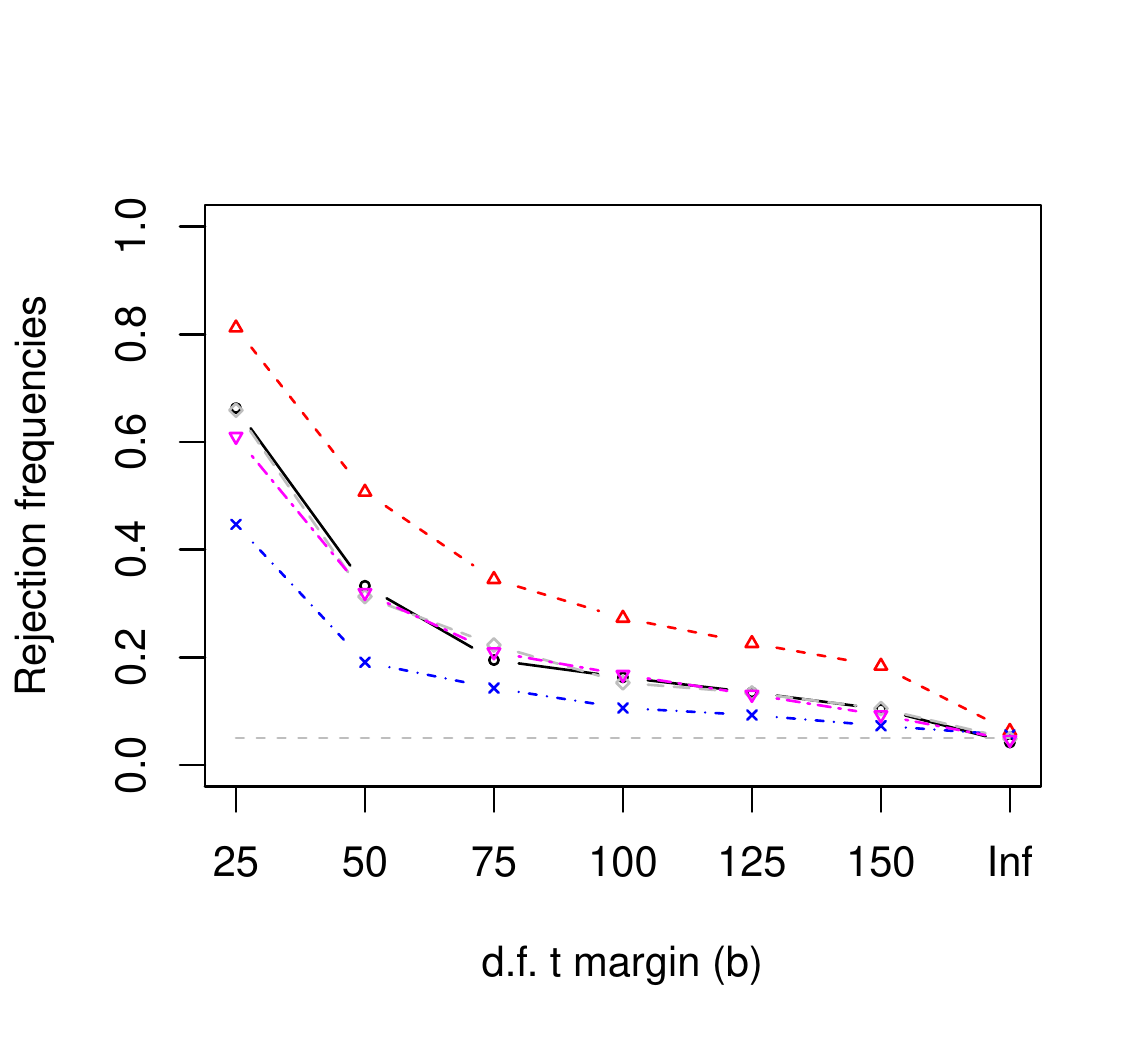}\\[-3mm]
\includegraphics[width=0.5\textwidth, height=70mm]{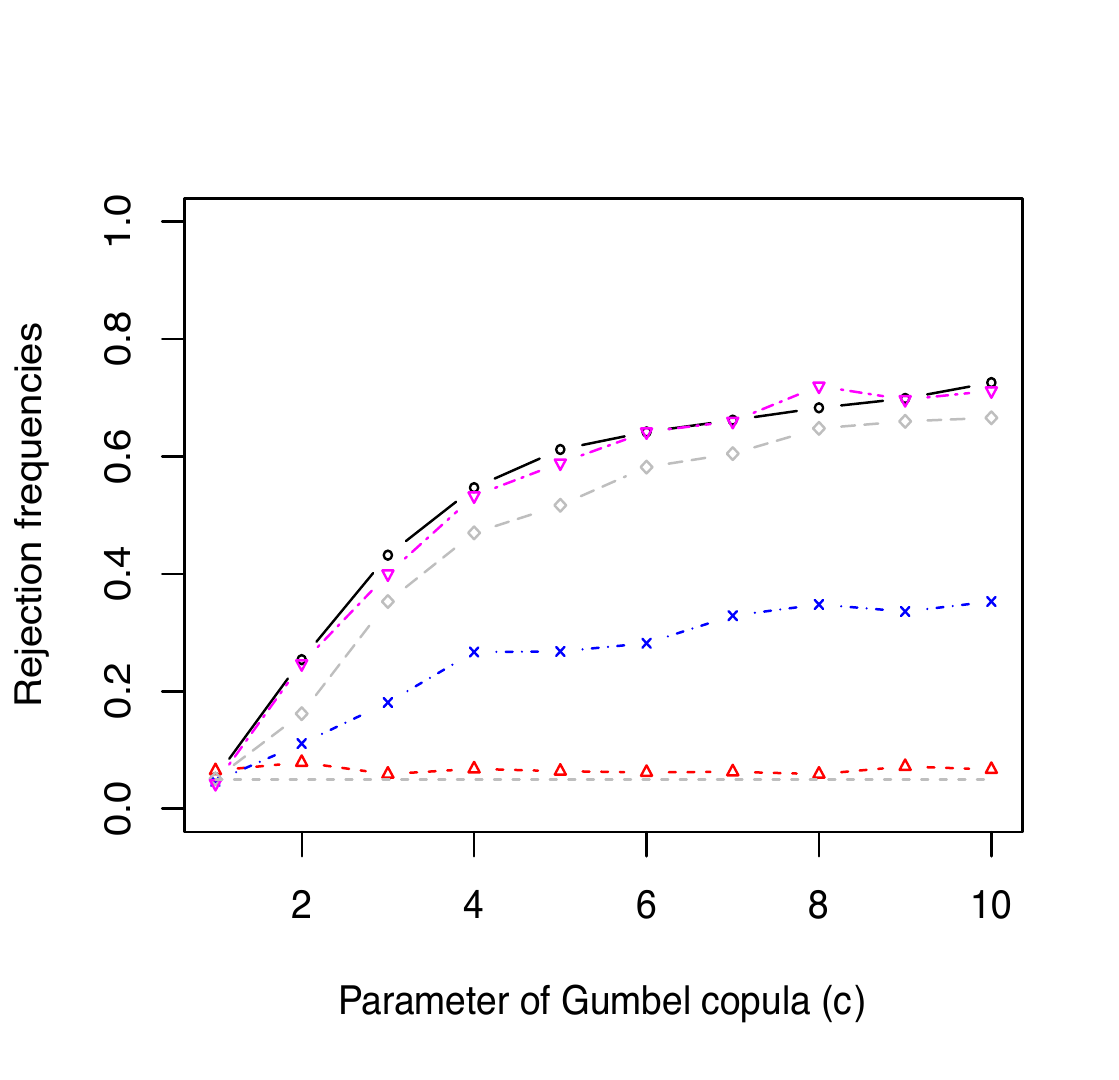}& \includegraphics[width=0.5\textwidth, height=70mm]{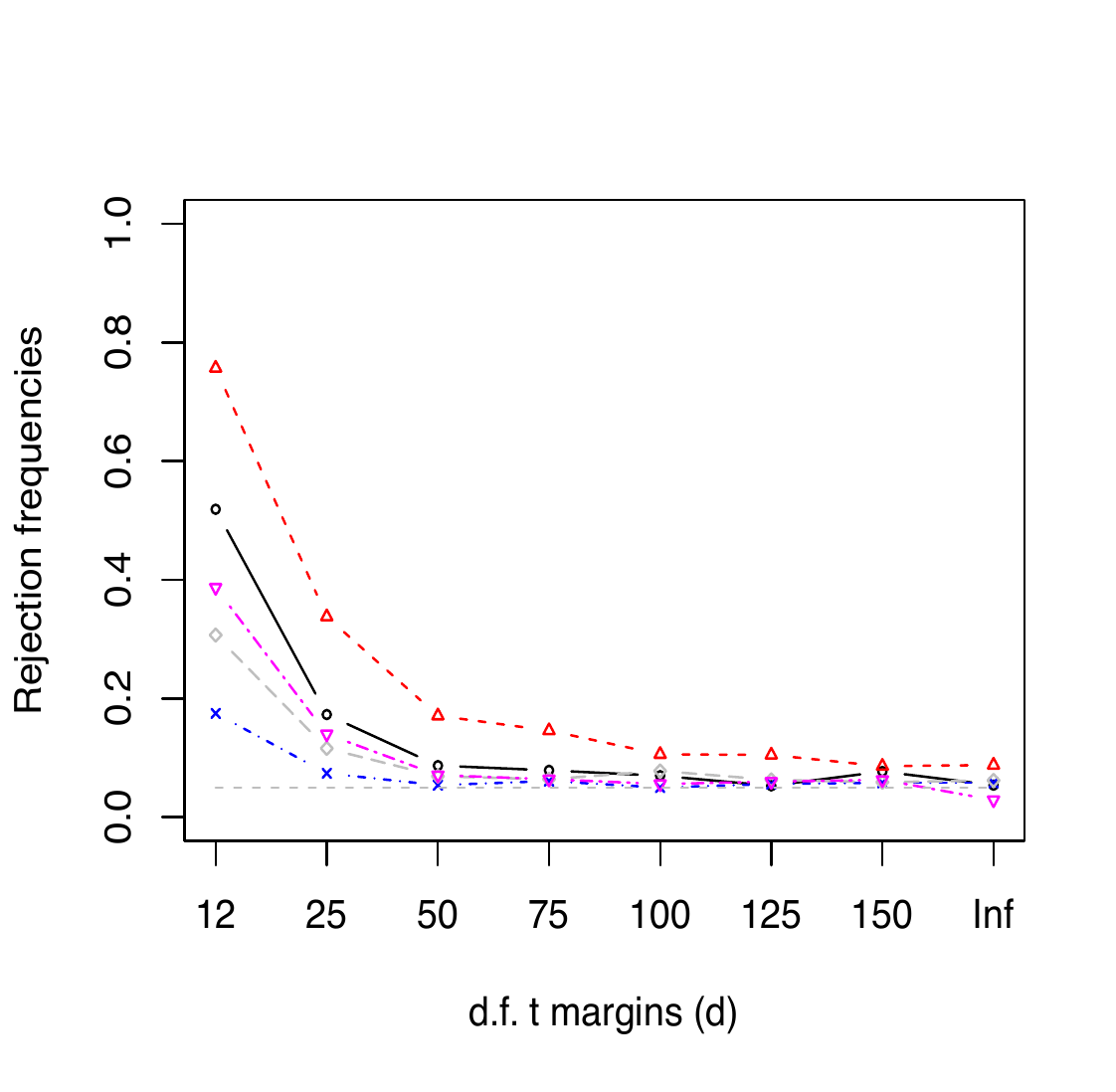}\\[-3mm]
\multicolumn{2}{ c }{ 
\includegraphics[width=0.7\textwidth]{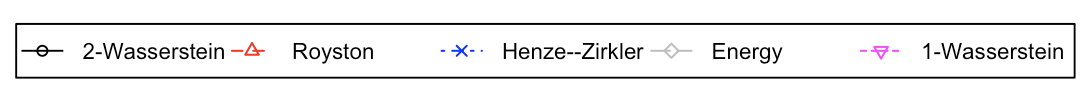}
\vspace{-5mm}}
\end{tabular}
\end{center}
\caption{\label{fig: GausFam} \slshape\small Empirical power curves of various tests that $\prob$ is $d$-variate Gaussian with unknown mean vector and covariance matrix. Top: $d = 2$. Bottom: $d = 5$. The Wasserstein test in Section~\ref{sec:group} is compared to three other multivariate normality tests mentioned in Section~\ref{sec:simu:group:Gauss}, where the alternatives (a)--(d) are described as well.}
\end{figure}

The Wasserstein tests have the highest power against the copula alternatives in~(a) and~(c), while Royston's test has no power at all, as expected.
For the Student \emph{t} alternatives in panels (b) and (d), Royston's test comes out as most powerful, but the Wasserstein and energy tests \citep{rizzo2016energy} perform quite well too.
It is also worth noticing that in panel~(b), Royston's test had a type~I error of 6.3\% and in panel~(d) this rose to 8.8\%. 

\subsubsection{Bivariate elliptical Student $t$ with unknown location and scatter}
\label{sec:simu:group:t}

For fixed scalar $\nu > 0$, the $d$-variate elliptical Student $t$ family with $\nu$ degrees of freedom and unknown location and scatter is generated by the affine transformation group in Example~\ref{ex:affine} applied to a spherical distribution whose radial density is that of the square root of a rescaled Fisher~$F(d,\nu)$ variable.
In dimension $d = 2$, we consider the hypothesis that $\prob$ is of this form with $\nu = 12$. 
%\js{Remarque à ignorer: si on avait considéré $\nu$ comme inconnu aussi, on aurait eu une application intéressante de l'approche hybride de la Section~\ref{sec:param:group}. L'obstacle, je suppose, c'est que $\nu$ est difficile à estimer.}

Figure~\ref{fig: TSkew} provides a plot of rejection frequencies under bivariate skew-\emph{t} alternatives \citep{azzalini:2014} with marginal skewness parameters~$\alpha_1$ and $\alpha_2$. 
%\js{Encore une remarque à ignorer: une autre distribution alternative aurait été de prendre $\nu \ne 12$. Même obstable probablement.}
Simulations were based on  %the function \textsf{rmst} from 
the \textsf{R} package \textsf{sn} \citep{sn}.
In principle, the empirical process approach in \citet{khmaladze2016unitary} leads to test statistics that are asymptotically distribution-free, but their numerical implementation involves a number of multiple integrals, the computation of which remains problematic.

% 
% We let $\prob_0$ be a bivariate Student \emph{t} distribution with $\nu = 12$ degrees of freedom, correlation parameter $\rho = 0$, and standardized to have unit variance. The null hypothesis is then that $\prob$ follows a bivariate Student \emph{t} distribution with $\nu = 12$ degrees of freedom, unknown mean vector and unknown positive definite covariance matrix. 
% 
% We compute the power of the Wasserstein test against alternatives $\prob$ from the bivariate skew-\emph{t} family \citep{azzalini:2014} with marginal skewness parameters $\alpha_1$ and $\alpha_2$. To draw samples from this distribution, we used the function \textsf{rmst} from the \textsf{R} package \textsf{sn} \citep{sn}. Figure~\ref{fig: TSkew} shows the power of the Wasserstein test against alternatives parametrized by $(\alpha_1,\alpha_2)$. In principle, the empirical process approach in \citet{khmaladze2016unitary} leads to test statistics that are asymptotically distribution free, but their calculation involves many multiple integrals, hindering practical implementation.

\begin{figure}
\includegraphics[width=0.70\textwidth, height=0.50\textwidth]{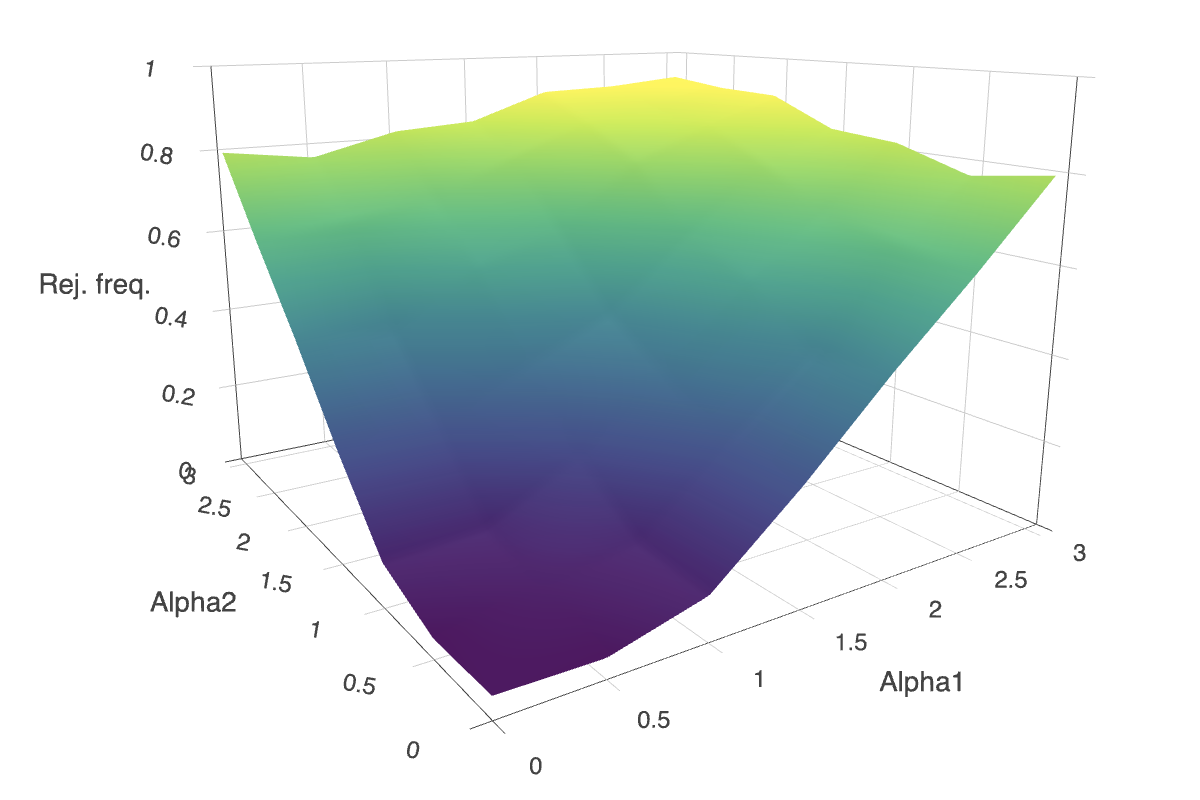}
\caption{\label{fig: TSkew} \slshape\small Empirical power of the Wasserstein test in \eqref{eq:phi:group} for the hypothesis that $\prob$ is bivariate Student~$t$ with $\nu = 12$ degrees and unknown mean vector and covariance matrix. The alternatives $\prob$ are bivariate skew-$t$ with skewness parameters $\alpha_1$ and $\alpha_2$.}
%\vspace{-5mm}
\end{figure}

\subsection{General parametric families}
\label{sec:simu:param}

Turning to more general parametric families $\model = \{ \prob_\theta : \theta \in \Theta\}$, we investigate the finite-sample performance of the Wasserstein-based tests in Section~\ref{sec:param}.
We consider models indexed by location, scale, and shape parameters.
As in Section~\ref{sec:param:group}, the location-scale parameters are treated as stemming from the transformation group in Example~\ref{ex:locscale} of Section~\ref{sec:group}, while for the shape parameters, we apply the parametric bootstrap. The test is thus the one we define in~\eqref{eq:phi:group:param}.
We numerically investigate our conjecture that the test has the correct size asymptotically.
In theory, the \citet{khmaladze2016unitary} approach also applies, but its implementation is intricate and remains unsettled, especially when there are multiple parameters.

\subsubsection{Gaussian margins and AMH copula}
\label{sec:simu:param:GaussAMH}

Let $\model$ consist of the bivariate distributions with Gaussian marginals and an Ali--Mikhail--Haq (AMH) copula, yielding a five-dimensional parameter vector $\theta = (\psi, \mu_1, \sigma_1, \mu_2, \sigma_2)$ with AMH copula parameter $\psi \in \Theta = [-1, 1]$, means $\mu_1, \mu_2 \in \reals$ and standard deviations $\sigma_1, \sigma_2 \in (0, \infty)$. 
%We applied the method involving the location-scale reduction described  in Remark~\ref{rem:ls}. 
The means and standard deviations are estimated by their empirical counterparts, so that the residuals~\eqref{eq:Zin:eta} are $\hat{Z}_{i,n} = (\hat{Z}_{i,1,n}, \hat{Z}_{i,2,n})$ with
\begin{equation}
\label{eq:Znij}
	\hat{Z}_{i,j,n} = (X_{i,j} - \hat{\mu}_{j,n} ) / \hat{\sigma}_{j,n}
\end{equation}
for $i = 1, \ldots, n$ and $j = 1, 2$.
Following \citet{genest+g+r:1995}, the copula parameter $\psi$ is estimated via a rank-based maximum pseudo-likelihood estimator.
Obviously, the component-wise ranks of the data and those of the residuals in \eqref{eq:Znij} coincide, so that $\hat{\psi}_n$, as required, depends on the data only through the residuals. 
The test statistic $T_{\model,n}$ in \eqref{eq:TMn:eta} is the Wasserstein distance between the empirical distribution of the residuals and the bivariate distribution with standard Gaussian margins and AMH copula with the estimated parameter. 

We first checked the validity of the parametric bootstrap procedure of Section~\ref{sec:param}. To do so, we simulated $1\,000$ independent random samples of size $n = 200$ from~$\prob \in \model$ with $\psi = 0.7$. For each sample, we calculated the test statistic $T_{\model,n}$ in~\eqref{eq:TMn:eta} and checked whether or not it exceeds the bootstrapped critical value~$c_{\model}(\alpha, n, \hat{\psi}_n)$ for $\alpha$ equal to multiples of $5\%$.
The critical values were computed as described below~\eqref{eq:phi:group:param}.
%The critical value function~$\theta \mapsto c_{\model}^{\mathrm{ls}}(\alpha, n, \theta)$ in~\eqref{eq:cMls} was pre-computed by Algorithm~\ref{algo: Bootstrap}, or rather a variation thereof taking into account the estimated residuals in Eq.~\eqref{eq:Znij}. 
The points in Figure~\ref{fig: ParamFamGausAMH}(a) show the empirical type~I errors as a function of~$\alpha$. 
The diagonal line fits the points well, supporting the conjecture that the parametric bootstrap is asymptotically valid.

Figure~\ref{fig: ParamFamGausAMH}(b) similarly displays the rejection frequencies of the Wasserstein test for $p = 2$ under an alternative $\prob$ whose copula belongs to the Frank family. 
If the Frank copula parameter is equal to zero, the Frank copula reduces to the independence copula, which is a member of the AMH family too. 

\begin{figure}
\begin{center}
\begin{tabular}{@{}c@{}c}
\includegraphics[width=0.50\textwidth, height=75mm]{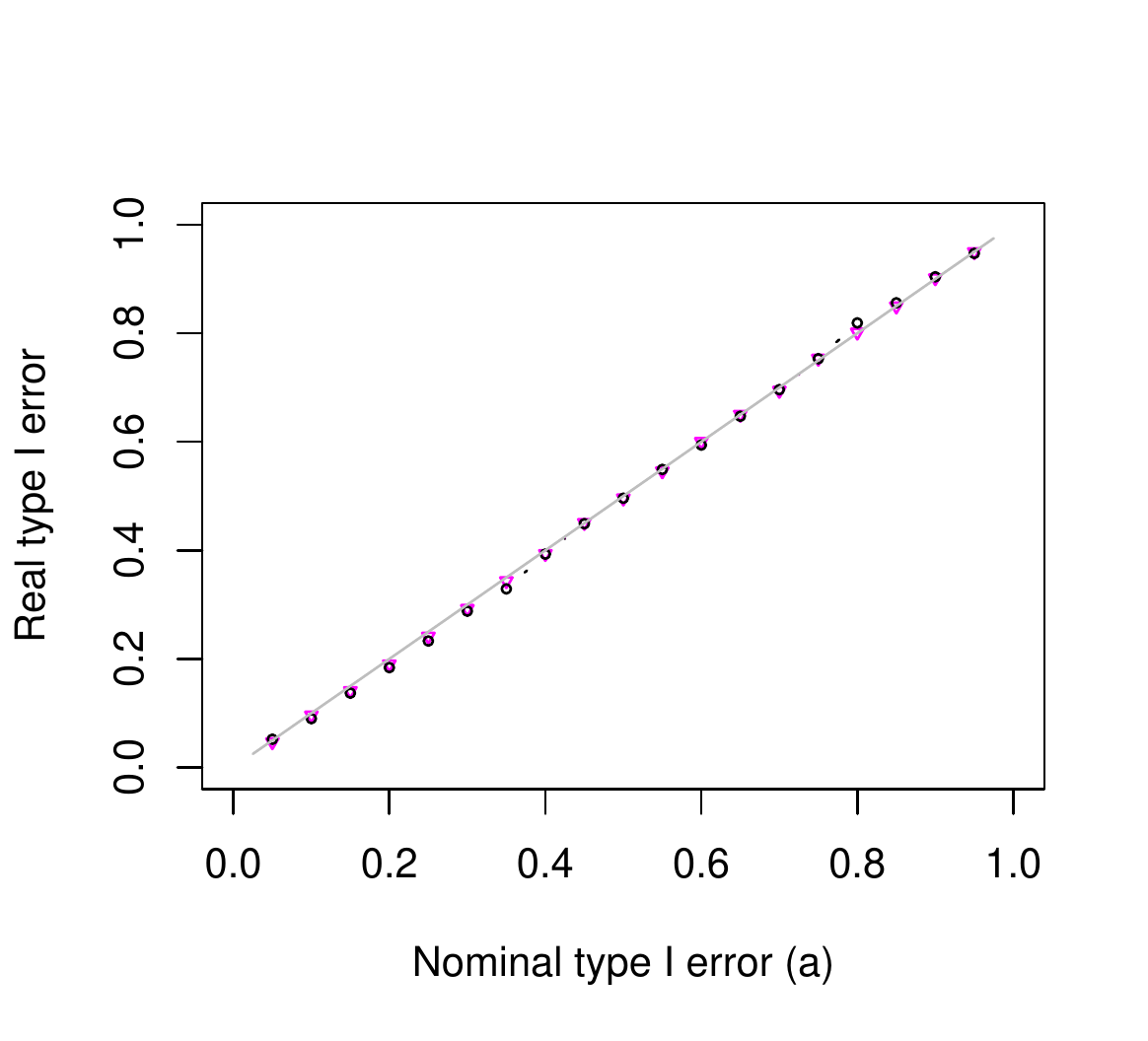}&
\includegraphics[width=0.50\textwidth, height=75mm]{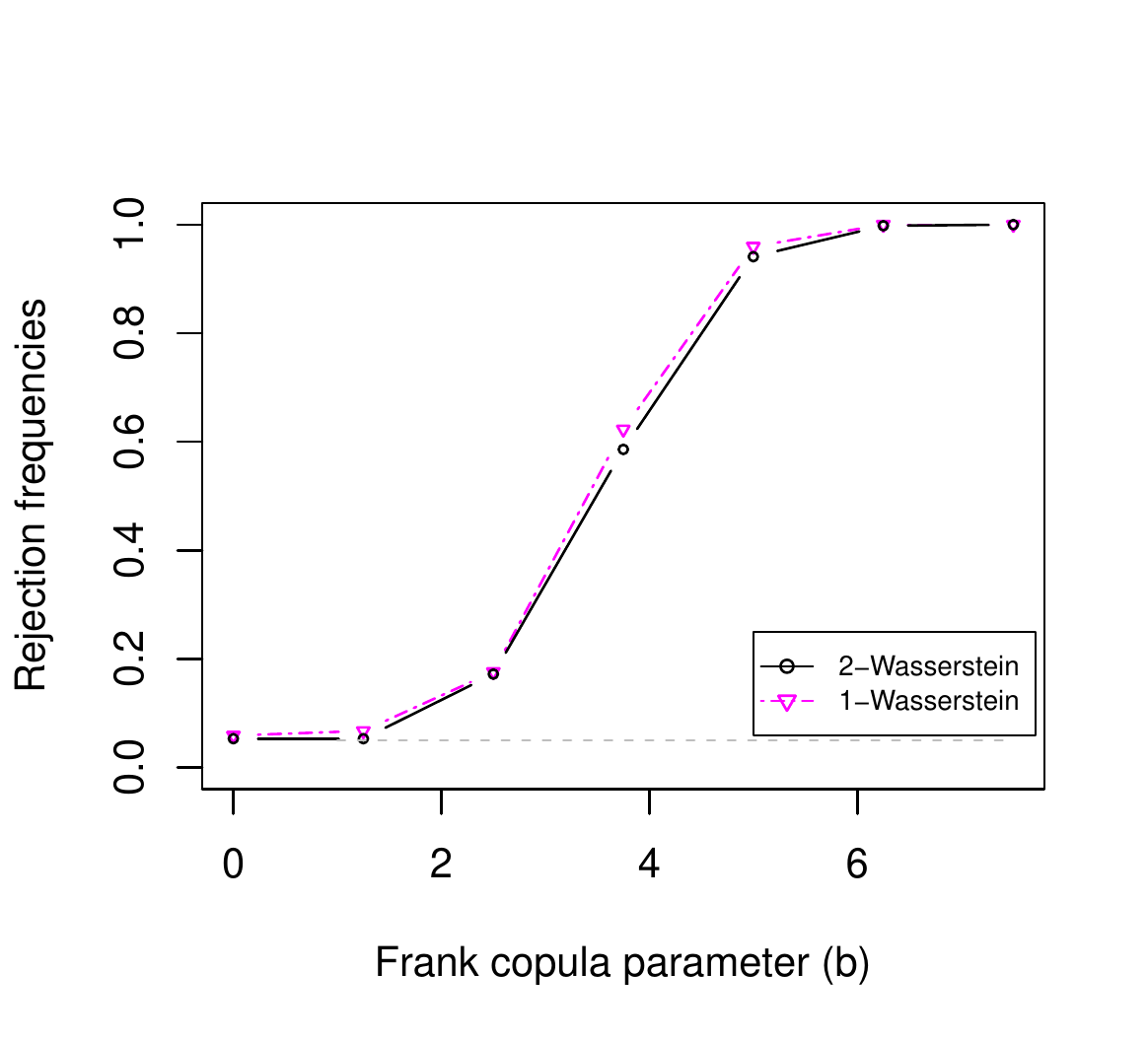}
\vspace{-5mm}\\
\end{tabular}
\end{center}
\caption{\label{fig: ParamFamGausAMH} \slshape\small Wasserstein tests for $\mathcal{H}_0^n : \prob \in \model$ with $\model$ the parametric family of bivariate distributions with Gaussian margins and AMH copula (Section~\ref{sec:simu:param}). 
Test statistics and critical values are computed from estimated residuals via a parametric bootstrap as in Section~\ref{sec:param:group}. Panel~(a): real versus nominal type~I errors $\alpha$ based on $1\,000$ samples of size $n = 200$ drawn from $\prob \in \model$ with $\psi = 0.7$.  Panel~(b): powers against alternatives $\prob$ with Gaussian marginals and Frank copula with varying parameter; if the latter is zero, the Frank copula is the independence one, which belongs to the AMH family too.}
\end{figure}

\subsubsection{A multivariate Gumbel max-stable family}
\label{sec:simu:param:maxstab}

Next, let $\model = \{ \probQ_\theta : \theta \in \Theta\}$ be the family of $d$-variate distributions with Gumbel margins with unknown location and scale parameters $(l_j, s_j) \in \reals \times (0, \infty)$ for $j = 1, \ldots, d$ and a Gumbel copula with unknown shape parameter $\psi \in [1, \infty)$. 
Each $\prob_\theta \in \model$ is thus a $d$-variate max-stable distribution, that is, a possible large-sample limit of the vector of affinely normalized component-wise maxima of an i.i.d.\ sample from a common distribution \citep[Chapter~9]{bgst:2004}. 

There are $2d + 1$ parameters in total. We treat the $2d$ location-scale parameters as indexing the transformation group in Example~\ref{ex:locscale} of Section~\ref{sec:group}. The parameters are estimated in two stages:
\begin{enumerate}
	\item The $2d$ location-scale parameters are estimated separately for each margin $j = 1, \ldots, d$ by maximum likelihood, producing $\hat{l}_j$ and $\hat{s}_j$.
	\item The copula parameter is estimated by maximum likelihood on the basis of the estimated residuals $\hat{Z}_{i,n} = (\hat{Z}_{i,j,n})_{j=1}^d$ with
	\[
		\hat{Z}_{i,j,n} = (X_{i,j} - \hat{l}_j) / \hat{s}_j,
		\qquad i = 1, \ldots, n.
	\]
\end{enumerate}
This two-stage maximum pseudo-likelihood estimation procedure usually enjoys a high relative efficiency with respect to the full maximum likelihood estimator and is computationally much simpler \citep{joe:2005}.
The location-scale estimators are equivariant under location-scale transformations.
The residuals and the copula parameter estimator are thus invariant under such transformations.

We then proceed as in Section~\ref{sec:param:group}.
The goodness-of-fit statistic $T_{\model,n}$ in \eqref{eq:TMn:eta} measures the Wasserstein distance between the empirical distribution of the estimated residuals and the $d$-variate max-stable distribution with standard Gumbel margins and Gumbel copula with the estimated parameter. 
The goodness-of-fit test is carried out as in \eqref{eq:phi:group:param}.

Figure~\ref{fig: ParamFamEVDlog} shows simulation results in dimensions $d = 2$ and $d = 5$ on top and bottom rows, respectively. 
\begin{itemize}
\item 
	On the left, the evaluation of the bootstrap accuracy is carried out as in Figure~\ref{fig: ParamFamGausAMH}(a). 
	Samples are generated from a distribution in the model with Gumbel copula parameter $\psi = 5/3$.
	The results support the conjecture that the Wasserstein-based tests with critical values calculated by the parametric bootstrap have the correct size, at least asymptotically.
\item 
	On the right, the power is calculated against alternative distributions whose margins are Generalized Extreme-Value (GEV) distributions with common shape parameter $\xi$ indicated on the horizontal axis.
	Note that $\xi = 0$ corresponds to the null hypothesis. 
	For $\xi > 0$, the distribution has finite moments up to order $p < 1/\xi$ only.
	This explains perhaps why the power of the Wasserstein test for $p = 2$ is less than for $p = 1$.
\end{itemize}

\begin{figure}
\begin{center}
\begin{tabular}{@{}c@{}c}
\includegraphics[width=0.5\textwidth, height=75mm]{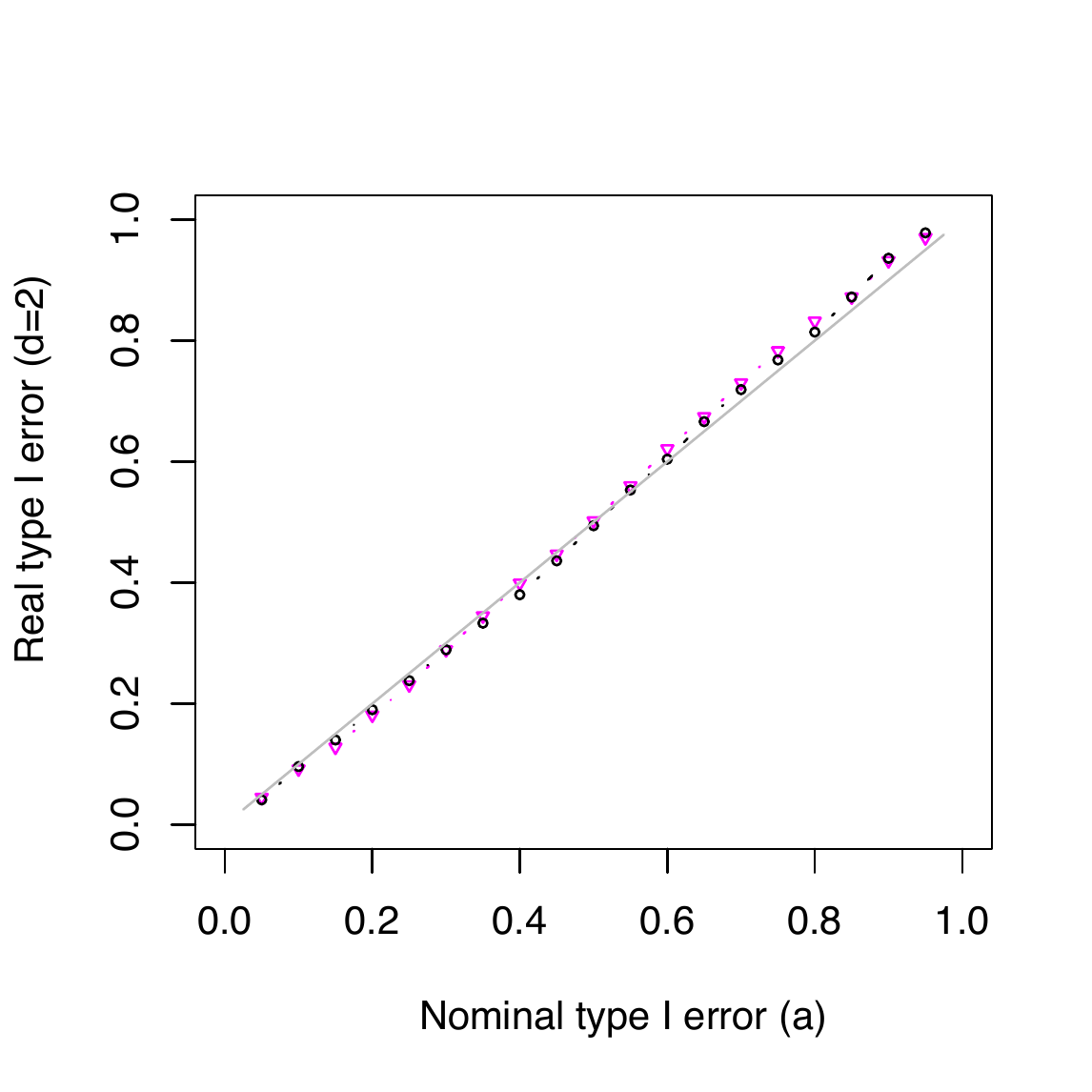}&
\includegraphics[width=0.5\textwidth, height=75mm]{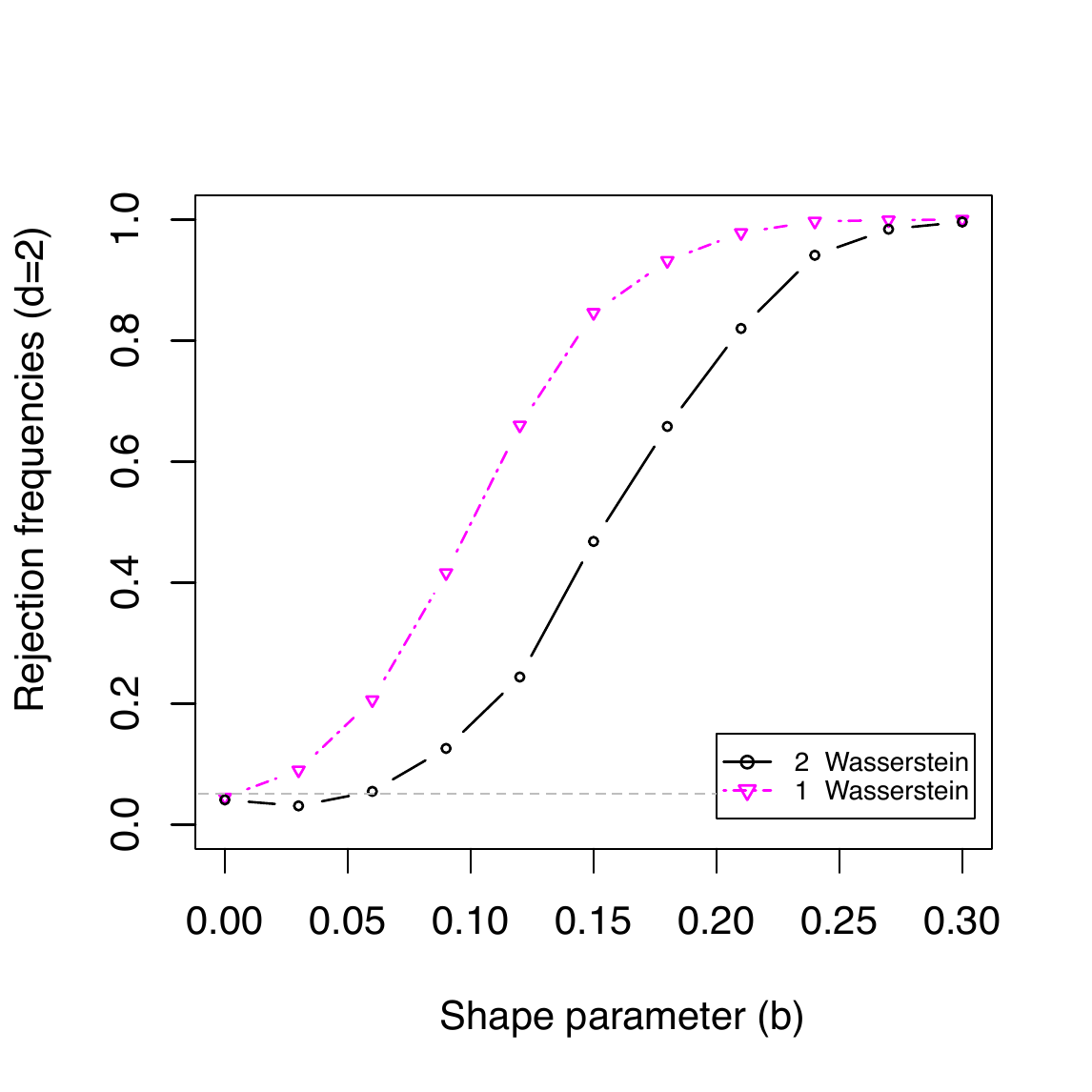}\\
\includegraphics[width=0.5\textwidth, height=75mm]{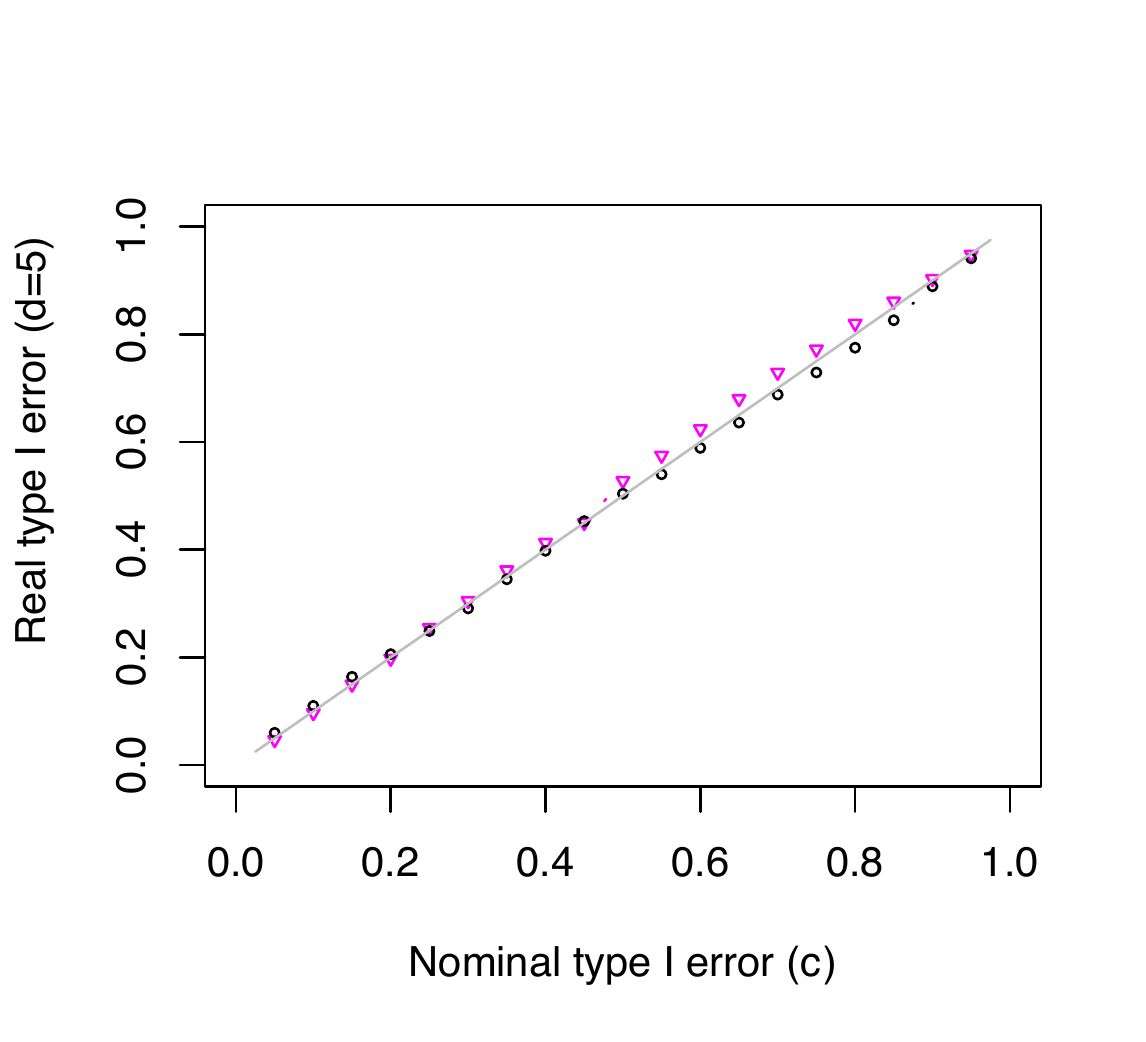}&
\includegraphics[width=0.5\textwidth, height=75mm]{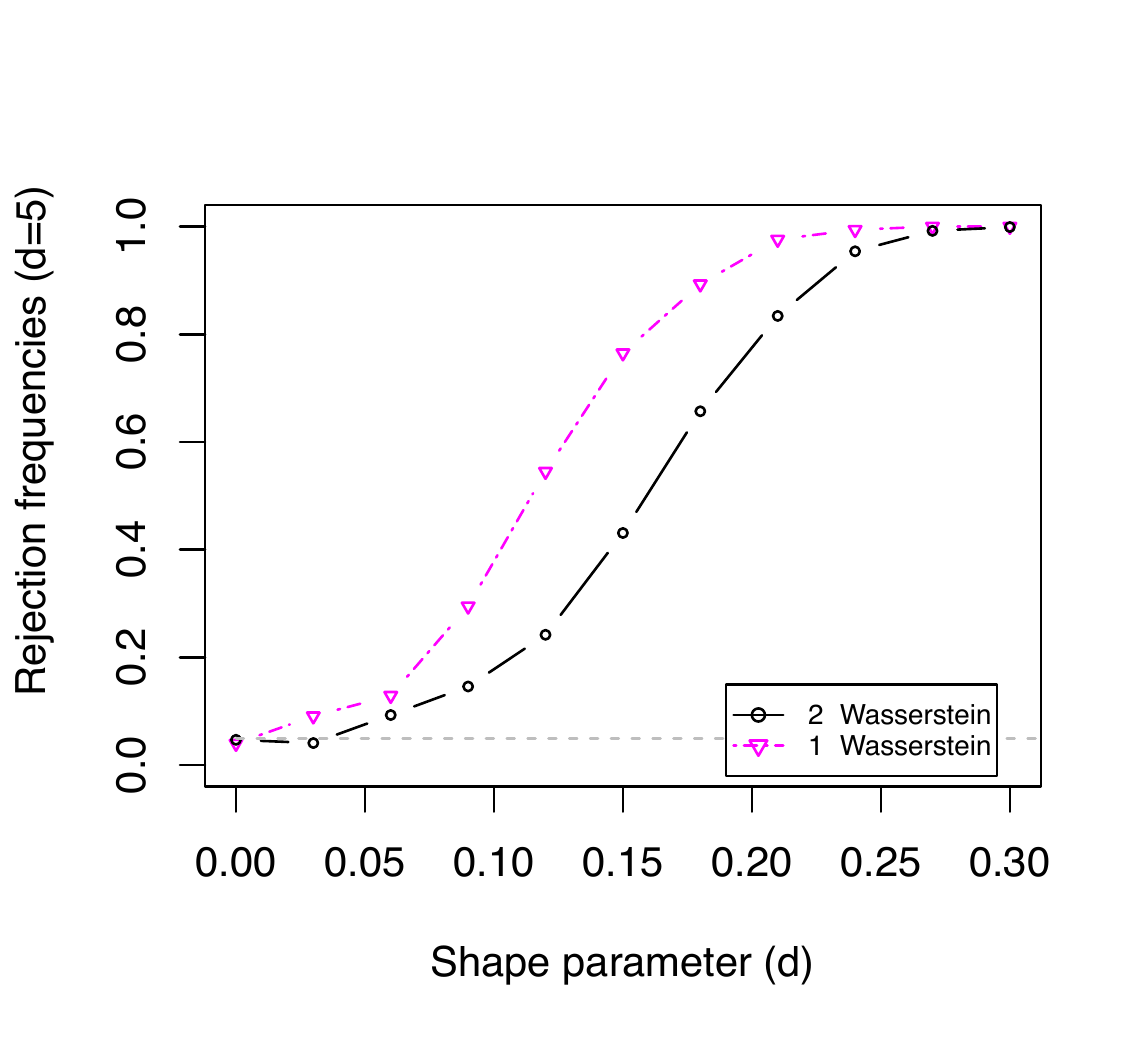}
\end{tabular}
\end{center}
\caption{\label{fig: ParamFamEVDlog} \slshape\small Wasserstein test for $\mathcal{H}_0^n : \prob \in \model$ with $\model$ the family of $d$-variate distributions with Gumbel margins with unknown location--scale parameters and Gumbel copula with unknown shape parameter $\psi \in [1, \infty)$ (Section~\ref{sec:simu:param:maxstab}). 
Test statistic and critical values based on estimated residuals and parametric bootstrap as in Section~\ref{sec:param:group}. 
Top: $d = 2$. Bottom: $d = 5$. 
Left: real versus nominal type~I errors $\alpha$ based on $1\,000$ samples of size $n = 200$ drawn from $\prob \in \model$ with Gumbel copula shape parameter $\psi = 5/3$.
%\js{On avait $\psi = 0.6$ avec la paramétrisation $\psi \in [0, 1]$. Maintenant, $\psi \in [1, \infty)$ comme dans la Section~\ref{sec:simu:group:Gauss} et $1/0.6 = 5/3$. Correct?} \gm{Parfaitement}
Right: power against alternatives $\prob$ with Gumbel copula and GEV marginals (shape parameter on the horizontal axis).}
\end{figure}

\section*{Acknowledgements}

The authors gratefully acknowledge the remarks and comments made by the reviewers that greatly helped improve the paper. J. Segers gratefully acknowledges funding from FNRS-F.R.S.\ grant CDR~J.0146.19.

\bibliographystyle{imsart-nameyear}
\bibliography{OptimalMatching}

\appendix

\section{Uniform convergence of the empirical Wasserstein distance}
\label{app:empWassUnif}

We establish here the convergence to zero in probability, uniformly in the underlying distribution $\prob \in \model$, of  the empirical Wasserstein distance~$W_p(\emprob, \prob)$ when~$\model \subseteq \Prob_p(\Rd)$ has a compact $W_p$-closure. The result is thus a law of large numbers for the empirical distribution in Wasserstein distance uniformly in the underlying distribution akin to Chung's uniform law of large numbers \citep[Proposition~A.5.1]{vdvw96}.

Actually, Theorem~\ref{thm:empWassUnif} establishes the stronger result that the convergence to zero holds uniformly in the $p$-th mean. The Markov inequality then implies (Corollary~\ref{cor:empWassUnif}) the desired uniform convergence in probability. The  notation is that of~Section~\ref{sec:asy}, with~$\expec_{\prob}$ standing for expectation under an independent random sample from $\prob$.

\begin{thm}
	\label{thm:empWassUnif}
	Let $\model \subseteq \Prob_p(\Rd)$ be such that
	\begin{equation}
	\label{eq:ui}
	\lim_{r \to \infty} \sup_{\prob \in \model}
	\int_{\norm{x} > r} \norm{x}^p \, \diff \prob(x) = 0.
	\end{equation}
	Then we have
	\[
	\lim_{n \to \infty} \sup_{\prob \in \model}
	\expec_{\prob}\big\{W_p^p(\emprob, \prob)\big\} = 0.
	\]
\end{thm}

%A more precise way to write this expectation is
%\[
%\int_{(\Rd)^n} W_p^p(L_n(\mathbf{x}_n), \prob) \, \diff \prob^n(\mathbf{x}_n).
%\]
%
The condition on $\model$ is equivalent to assuming that the closure of $\model$ in the metric space $(\Prob_p(\Rd), W_p)$ is compact. This follows from Prohorov's theorem and the characterization of $W_p$-convergence in \citet[Lemma~8.3]{bickel+f:1981} or \citet[Theorem~6.9]{villani2008optimal}. 

The convergence rate of $\expec_{\prob}\{W_p^p(\emprob, \prob)\}$ has been studied intensively. 
In \citet[Theorem~1]{fournier2015rate}, for instance, the expectation is bounded by an explicit expression involving $n,p,d$, and the moment of order $q$ of $\prob$ for some $q > p$. 
Bounds on such moments for all $\prob \in \model$ then imply a uniform rate of convergence in $\prob \in \model$.
In contrast, we do not impose the existence of moments of order $q$ higher than $p$, but only the uniform integrability of the $p$-th order moment.

The challenge in the proof of Theorem~\ref{thm:empWassUnif} is to obtain a sufficiently sharp and explicit bound on $\expec_{\prob}\{W_p^p(\emprob, \prob)\}$. Such a bound is known for absolutely continuous measures in terms of a weighted total variation distance \citep[Theorem~6.15]{villani2008optimal}. To apply that bound, an additional smoothing step is needed, and the whole procedure needs to work uniformly in the underlying probability measure, relying only on the uniform integrability condition \eqref{eq:ui}.

\begin{proof}[Proof of Theorem~\ref{thm:empWassUnif}]
	The following smoothing argument is inspired by the proof of Theorem~1.1 in \citet{horowitz1994mean}. Let $\probU_\sigma$ denote the Lebesgue-uniform distribution on the ball $\{x \in \Rd : \norm{x} \le \sigma\}$ in $\Rd$ with radius $\sigma \in (0, \infty)$ and centered at the origin. Denoting by $\ast$ the convolution of probability measures, we have, for any~$\probQ \in \Prob_p(\Rd)$, 
		\[
	W_p(\probQ \ast \probU_\sigma, \probQ) \le \sigma.
	\]
	Indeed, if $X$ and $Y$ are independent random vectors with distributions $\probQ$ and~$\probU_\sigma$, respectively, then $(X+Y, X)$ is a coupling of $\probQ \ast \probU_\sigma$ and $\probQ$, so that 
	\[
	W_p^p(\probQ \ast \probU_\sigma, \probQ) \le \expec[\norm{Y}^p] \le \sigma^p.
	\] By the triangle inequality, it follows that
	\[
	W_p(\emprob, \prob)
	\le 2\sigma + W_p(\emprob \ast \probU_\sigma, \prob \ast \probU_\sigma).
	\]
	Taking expectations and using the elementary inequality 
	$$(a + b)^p \le 2^{p-1}(a^p + b^p)\quad \text{ for } p \geq 1,\  a\geq 0,\ \text{ and } b\geq 0,
	$$
	 we obtain 
	\[
	\expec_{\prob}\big\{W_p^p(\emprob, \prob)\big\}
	\le 2^{p-1} \big[2^p \sigma^p + \expec\big\{W_p^p(\emprob \ast \probU_\sigma, \prob \ast \probU_\sigma)\big\}\big].
	\]
	If we can show that
	\begin{equation}
	\label{eq:empWassUnif:toshow}
	\forall \sigma > 0, \qquad
	\lim_{n \to \infty} \sup_{\prob \in \model}
	\expec\big\{W_p^p(\emprob \ast \probU_\sigma, \prob \ast \probU_\sigma)\big\} = 0,
	\end{equation}
	then it will follow that
	\[
	\forall \sigma > 0, \qquad
	\limsup_{n \to \infty} \sup_{\prob \in \model}
	\expec_{\prob}\big\{W_p^p(\emprob, \prob)\big\}
	\le 2^{2p-1} \sigma^p.
	\]
	But then the $\limsup$ is actually a limit and is equal to zero, as required.
	
	Let us proceed to show \eqref{eq:empWassUnif:toshow}. Fix $\sigma > 0$ for the remainder of the proof. Let~$f_\sigma$ denote the density function of $\probU_\sigma$. The measures $\emprob \ast \probU_\sigma$ and $\prob \ast \probU_\sigma$ are absolutely continuous too and have density functions $x \mapsto n^{-1} \sum_{i=1}^n f_\sigma(x - X_i)$ and $x \mapsto \int_{\Rd} f_\sigma(x - y) \diff \prob(y)$, respectively.
	The Wasserstein distance can be controlled by weighted total variation \citep[Theorem~6.15]{villani2008optimal}:
	\begin{align*}
	\lefteqn{
	W_p^p(\emprob \ast \probU_\sigma, \prob \ast \probU_\sigma)
	} \\
	&\le 2^{p-1} \int_{\Rd} \norm{x}^p \, \diff \lvert \emprob \ast \probU_\sigma - \prob \ast \probU_\sigma \rvert (x) \\
	&= 2^{p-1} \int_{\Rd} \norm{x}^p \, 
	\left|
	\frac{1}{n} \sum_{i=1}^n f_\sigma(x - X_i)
	- \int_{\Rd} f_\sigma(x - y) \, \diff \prob(y)
	\right| \,
	\diff x.
	\end{align*}
	Take expectations and apply Fubini's theorem to see that
	\begin{equation}
	\label{eq:empWassUnif:integral}
	\expec_\prob\big\{
	W_p^p(\emprob \ast \probU_\sigma, 
	\prob \ast \probU_\sigma)
	\big\}
	\le 2^{p-1} 
	\int_{\Rd} \norm{x}^p g_{n}(x; \prob) \, \diff x
	\end{equation}
	where
	\[
	g_{n}(x; \prob) =
	\expec_{\prob} \left[
	\left|
	\frac{1}{n} \sum_{i=1}^n f_\sigma(x - X_i)
	- \int_{\Rd} f_\sigma(x - y) \, \diff \prob(y)
	\right|			
	\right].
	\]
	
	Let $r > \sigma$ and split the integral in \eqref{eq:empWassUnif:integral} according to whether $\norm{x} > r$ or~$\norm{x} \le r$. Note that $f_\sigma(u) = f_{\sigma}(0)$ if $\norm{y} \le \sigma$ and $f_\sigma(u) = 0$ otherwise. For any $\prob \in \Prob(\Rd)$ and any $x \in \Rd$, we have, by the Cauchy--Schwarz inequality,
	\[
	g_{n}(x; \prob) \le n^{-1/2} f_{\sigma}(0).
	\]
	It follows that
	\[
	\lim_{n \to \infty}
	\sup_{\prob \in \Prob(\Rd)}
	\int_{\norm{x} \le r} \norm{x}^p g_{n}(x;\prob) \, \diff x
	= 0.
	\]
	But then, in view of \eqref{eq:empWassUnif:integral}, we have
	\[
	\limsup_{n \to \infty} \sup_{\prob \in \model}
	\expec_{\prob}\big\{W_p^p(\emprob \ast \probU_\sigma, \prob \ast \probU_\sigma)\big\} \\
	\le \limsup_{n \to \infty} \sup_{\prob \in \model}
	2^{p-1} 
	\int_{\norm{x} > r} \norm{x}^p g_{n}(x; \prob) \, \diff x.
	\]
	By the triangle inequality, we also have, for all $n$,
	\[
	g_{n}(x; \prob)
	\le 2 \int_{\Rd} f_\sigma(x - y) \diff \prob(y).
	\]
	Applying Fubini's theorem once more, we obtain
	\begin{align*}
	\int_{\norm{x} > r} \norm{x}^p g_{n}(x; \prob) \, \diff x
	&\le 2 \int_{\norm{x} > r} \norm{x}^p \int_{y \in \Rd} f_\sigma(x - y) \, \diff \prob(y) \, \diff x \\
	&= 2 \int_{y \in \Rd} \int_{\norm{x} > r} \norm{x}^p f_\sigma(x-y) \, \diff x \, \diff \prob(y) \\
	&= 2 \int_{y \in \Rd} \int_{\norm{u+y} > r} \norm{u+y}^p f_\sigma(u) \, \diff u \, \diff \prob(y).
	\end{align*}
	Since $f_\sigma(u) = 0$ whenever $\norm{u} > \sigma$ and since $r > \sigma$, we have
	\[
	\int_{\norm{u+y} > r} \norm{u+y}^p f_{\sigma}(u) \, \diff u
	\le
	\begin{cases}
	2^{p-1}(\sigma^p + \norm{y}^p) & \text{if $\norm{y} > r-\sigma$}, \\
	0 & \text{otherwise}.
	\end{cases}
	\]
	Choosing $r > 2 \sigma$, we get that $\norm{y} > \sigma$ for all $y$ in the non-zero branch above, and thus, for all $n$,
	\[
	\int_{\norm{x} > r} \norm{x}^p g_{n}(x; \prob) \, \diff x
	\le 2^{p+1} \int_{\norm{y} > r-\sigma} \norm{y}^p \, \diff \prob(y).
	\]
	It follows that, for every $r > \sigma$, %we have
	\[
	\limsup_{n \to \infty} \sup_{\prob \in \model}
	\expec_{\prob}\big\{
	W_p^p(\emprob \ast \probU_\sigma, \prob \ast \probU_\sigma)
	\big\}
	\le 2^{2p} \sup_{\prob \in \model} 
	\int_{\norm{y} > r - \sigma} \norm{y}^p \, \diff \prob(y).
	\]
	The left-hand side does not depend on $r$. The condition on $\model$ implies that the right-hand side converges to zero as $r \to \infty$. It follows that the left-hand side must be equal to zero. But this is exactly \eqref{eq:empWassUnif:toshow}, as required. The proof is complete.
\end{proof}

\begin{cor}
	\label{cor:empWassUnif}
	For $\model$ as in Theorem~\ref{thm:empWassUnif}, we have
	\[
	\forall \eps > 0, \qquad
	\lim_{n \to \infty} \sup_{\prob \in \model}
	\prob^n\big[W_p^p(\emprob, \prob) > \eps\big] = 0,
	\]
	i.e.,   $W_p^p(\emprob, \prob) \to 0$ in probability as $n \to \infty$, uniformly in $\prob \in \model$.
\end{cor}

\begin{proof}
	By Markov's inequality, for every $\eps > 0$ and every $\prob \in \Prob_p(\Rd)$, we have 
	\[ 
	\prob^n\big[W_p(\emprob, \prob) > \eps\big] 
	\le \eps^{-p} 
	\expec_{\prob}\big\{W_p^p(\emprob, \prob)\big\}.
	\]
	In view of Theorem~\ref{thm:empWassUnif}, the expectation converges to zero uniformly in $\prob \in \model$.
\end{proof}

For a single $\prob \in \Prob_p(\Rd)$, Lemma~8.4 in \citet{bickel+f:1981} says that $W_p(\emprob, \prob) \to 0$ almost surely as $n \to \infty$. 
Whether Corollary~\ref{cor:empWassUnif} can be strengthened to almost sure convergence uniformly in $\prob$, i.e., whether
\[
	\forall \eps > 0, \qquad
	\lim_{n \to \infty} \sup_{\prob \in \model}
	\prob^n\left[ \sup_{m \ge n} W_p^p(\widehat{\prob}_m, \prob) > \eps \right]
	= 0,
\]
remains an open problem.

\bgroup

% =============================================
\section{Consistency of the parametric bootstrap in the univariate case}
\label{sec:boot:cons}

In Section~\ref{sec:param}, we left open the conjecture of the consistency of the parametric bootstrap procedure for the Wasserstein GoF test in general parametric families.
Proving it requires asymptotic distribution theory for the empirical Wasserstein distance, which, in general, is a difficult and long-standing open problem (Section~\ref{sec:asy}).
In the univariate case, however, the large-sample distribution of the empirical Wasserstein distance is known, enabling a theoretical analysis.

Consider the same notation as in Section~\ref{sec:param} and assume $d = 1$. In the univariate case, the Wasserstein distance can be expressed as the $L_p$ distance between quantile functions, see \citet[Section~1.2.3]{panaretos2019statistical} and the references therein.
Let $\hat{F}_n^{-1}$ denote the empirical quantile function of the sample $X_1,\ldots,X_n$ and let $F_\theta^{-1}$ denote the quantile function of $\prob_\theta$, defined as the (generalized) inverse of the cumulative distribution function $F_\theta$ of $\prob_\theta$.
The normalized Wasserstein GoF test statistic in Section~\ref{sec:param} takes the form
\begin{equation}
\label{eq:Rn}
	R_n 
	:= n^{p/2} T_{\model,n} 
	= n^{p/2} \, W_p^p \bigl( \emprob, \prob_{\hat\theta_n} \bigr)
	= \int_0^1 \abs{\zeta_n(u)}^p \, \diff u,
\end{equation}
where $\zeta_n$ is the empirical quantile process at the estimated parameter:
\begin{equation}
\label{eq:zetan}
	\zeta_n(u) = \sqrt{n} \bigl\{ 
		\hat{F}_n^{-1}(u) - F_{\hat{\theta}_n}^{-1}(u)
	\bigr\},
	\qquad u \in (0, 1).
\end{equation}

We follow the notation and logic of \citet{beran1997}.
Let $H_n(\theta)$ denote the sampling distribution of $R_n$ under $\prob_\theta^n$.
We would like to use $H_n(\theta)$ to draw inference based on the observed value of $R_n$, for instance by comparing the latter to a critical value computed under $H_n(\theta)$.
Since we do not know $\theta$, we do not know $H_n(\theta)$ either. 
The parametric bootstrap consists of estimating the unknown sampling distribution of $R_n$ by the random probability measure $H_n(\hat{\theta}_n)$, the sampling distribution of the statistic $R_n$ under the estimated parameter.
In practice, we calculate relevant quantities related to $H_n(\hat{\theta}_n)$ such as critical values or p-values by Monte Carlo simulation, drawing many bootstrap samples $X_1^*,\ldots,X_n^*$ from $\prob_{\hat{\theta}_n}$ and calculating the test statistic $R_n^*$ from those.

The question is whether inference drawn from $R_n$ based on $H_n(\hat{\theta}_n)$ rather than on $H_n(\theta)$ is still valid, at least asymptotically.
Let $H(\theta)$ denote the limit distribution of $R_n$ under $\prob_\theta^n$ as $n \to \infty$, assuming it exists.
If $H_n(\theta_n)$ converges weakly to the same limit $H(\theta)$ for any sequence $\theta_n \in \Theta$ such that 
\begin{equation} 
\label{eq:thetan}
	\sqrt{n}(\theta_n - \theta) = \Oh(1), \qquad n \to \infty, 
\end{equation}
then it follows that in $\prob_\theta^n$-probability, the estimated sampling distribution $H_n(\hat{\theta}_n)$ converges weakly to $H(\theta)$ for all estimator sequences $\hat{\theta}_n$ such that $\sqrt{n} (\hat{\theta}_n - \theta) = \Oh_{\prob_\theta^n}(1)$ as $n \to \infty$.
[The estimator $\hat{\theta}_n$ in the definition of $R_n$ need not even be the same as the one under which we calculate the sampling distribution $H_n(\hat{\theta}_n)$, but for simplicity, we assume it is.]

The normalized test statistic $R_n$ in Eq.~\eqref{eq:Rn} is a functional of the empirical quantile process $\zeta_n$ in Eq.~\eqref{eq:zetan}.
In Proposition~\ref{prop:boot:cons} below, we will show that the weak limits as $n \to \infty$ of the finite-dimensional distributions of $\zeta_n$ under $\prob_{\theta_n}^n$ for $\theta_n \to \theta \in \Theta$ do not depend on the particular sequence $\theta_n$ as long as Eq.~\eqref{eq:thetan} and certain assumptions on the model and the estimator sequence are fulfilled.

\begin{assumption}
	\label{ass:boot:cons}
	$\model = \{ \prob_\theta : \theta \in \Theta \} \subset \Prob(\reals)$ is a parametric model and $\hat{\theta}_n$ is an estimator sequence satisfying the following properties:
\begin{enumerate}[({A}1)]
	\item $\Theta$ is an open subset of $\reals^k$.
	\item $P_\theta$ has density $f_\theta$ with respect to one-dimensional Lebesgue measure for every $\theta \in \Theta$.
	\item $\model$ is differentiable in quadratic mean at any $\theta \in \Theta$ with score function $\dot{\ell}_\theta = \nabla_\theta \log f_\theta(x) : \reals \to \reals^k$ and non-singular $k \times k$ Fisher information matrix $\mathcal{I}_\theta = \expec_\theta[\dot{\ell}_\theta(X) \dot{\ell}_\theta(X)^\top]$.
	\item The estimator sequence $\hat{\theta}_n$ is regular and asymptotically linear at $\theta \in \Theta$ with influence function $\psi_\theta$.
\end{enumerate}
\end{assumption}

Asymptotic linearity in Assumption~(A4) means that
\begin{equation}
\label{eq:psitheta}
	n^{1/2} (\hat{\theta}_n - \theta) = \frac{1}{\sqrt{n}} \sum_{i=1}^n \psi_\theta(X_i) + \oh_{\prob_\theta^n}(1), \qquad n \to \infty,
\end{equation}
with $\prob_\theta$-square integrable influence function $\psi_\theta : \reals \to \reals^k$ satisfying $\expec_\theta[ \psi_\theta(X) ] = 0$.
Regularity in~(A4) means that the influence function $\psi_\theta$ satisfies
\[
	\psi_\theta - \tilde{\ell}_\theta \perp_\theta \dot{\ell}_\theta
\]
where $\tilde{\ell}_\theta = \mathcal{I}_\theta^{-1} \dot{\ell}_\theta$ is the efficient influence function for estimating $\theta$ and $\perp_\theta$ means orthogonality in $L_2(P_\theta)$; an equivalent criterion is that
\begin{equation}
\label{eq:AEL}
	\expec_\theta[ \psi_\theta(X) \, \dot{\ell}_\theta(X)^\top] = I_k,
\end{equation}
the $k \times k$ identity matrix.
We refer to \citet[Chapters~7--8]{vdv98} and \citet[Chapter~2]{bkrw93} for more background on these assumptions, which are standard. 

\begin{prop}
	\label{prop:boot:cons}
	Let $\model = \{ \prob_\theta : \theta \in \Theta \} \subset \Prob(\reals)$ be a parametric model and $\hat{\theta}_n$ an estimator sequence such that Assumption~\ref{ass:boot:cons} is satisfied.
	Suppose that $f_\theta$ is continuous and strictly positive on the interior of the support of $\prob_\theta$ and that, for every fixed $u \in (0, 1)$, the quantile function $F_\theta^{-1}(u)$ is continuously differentiable as a function of $\theta$.
	Then, for all $\theta, \theta_n \in \Theta$ satisfying Eq.~\eqref{eq:thetan} and for every vector $(u_1, \ldots, u_m) \in (0, 1)^m$, the quantile process $\zeta_n$ in Eq.~\eqref{eq:zetan} satisfies
	\begin{equation}
	\label{eq:zetanlim}
		\bigl( \zeta_n(u_j) \bigr)_{j=1}^m \stackrel{\theta_n}{\rightsquigarrow}
		\normal_m(0, \Gamma_\theta), 
		\qquad n \to \infty,
	\end{equation}
	where $\Gamma_\theta$ is a certain $m \times m$ covariance matrix given in Eq.~\eqref{eq:Gamma} below and where the arrow means weak convergence of the law under $\prob_{\theta_n}^n$ of the random vector on the left-hand side to the law of the random vector on the right-hand side.
\end{prop}

The proof of Proposition~\ref{prop:boot:cons} is given below.
Passing from the asymptotics of the finite-dimensional distributions of $\zeta_n$ to those of $R_n$ in Eq.~\eqref{eq:Rn} requires two things: asymptotic tightness of $\zeta_n$ as well as regularity conditions on $\prob_\theta$ controlling the tails of the quantile functions to be integrated. The former is a classical topic in empirical process theory, see for instance Chapters~19 and~21 in \citet{vdv98}. Regarding the latter, see \citet{delbarrio2005} for the case $p = 2$ and \citet{bobkov+l:2019} for general $p \ge 1$.

The important thing in Proposition~\ref{prop:boot:cons} is that the limit~\eqref{eq:zetanlim} does not depend on the sequence $(\theta_n)_n$.
If integration and passage to the limit in Eq.~\eqref{eq:Rn} is permitted, the asymptotic equivariance in law in Eq.~\eqref{eq:zetanlim} for $\zeta_n$ continues to hold for the normalized test statistic $R_n$.
The consistency of the parametric bootstrap for the Wasserstein GoF test then follows as explained in the lines below Eq.~\eqref{eq:thetan}.

\begin{proof}[Proof of Proposition~\ref{prop:boot:cons}]
	By a subsequence argument, we can and will assume that $h_n = \sqrt{n} (\theta_n - \theta) \to h \in \reals^k$ as $n \to \infty$.
	The proof proceeds by Le Cam's third lemma following the strategy in \citet[Section~7.5]{vdv98}.
	
	By Theorem~7.2 in the same reference, Assumption~\ref{ass:boot:cons} implies that the log-likelihood ratio of $\prob_{\theta_n}^n$ with respect to $\prob_\theta^n$ admits the expansion
	\begin{equation}
	\label{eq:LAN}
		\log \prod_{i=1}^n \frac{f_{\theta_n}}{f_\theta}(X_i)
		= \frac{1}{\sqrt{n}} \sum_{i=1}^n h^\top \dot{\ell}_\theta(X_i)
		- \frac{1}{2} h^\top \mathcal{I}_\theta h + \oh_{\prob_\theta^n}(1),
		\qquad n \to \infty,
	\end{equation}
	with $n^{-1/2} \sum_{i=1}^n \dot{\ell}_\theta(X_i)$ asymptotically $\normal_k(0, \mathcal{I}_\theta)$ under $\prob_\theta^n$ as $n \to \infty$.
	This means that the sequence of statistical experiments $\{ \prob_\theta^n : \theta \in \Theta \}$ is locally asymptotically normal.
	
	To show Eq.~\eqref{eq:zetanlim}, we need to find the joint limit distribution under $\theta$ of the finite-dimensional distributions of $\zeta_n$ together with the log-likelihood ratio in Eq.~\eqref{eq:LAN}.
	If the joint limit is Gaussian and if the asymptotic cross-covariance with the term corresponding to the log-likelihood ratio is zero, then by Le Cam's third lemma \citep[Example~6.7]{vdv98}, the asymptotic distribution of $\zeta_n$ under $\theta_n$ is the same as under $\theta$, as required.
	 
	Fix $0 < u < 1$. We derive an asymptotically linear expansion of $\zeta_n(u)$.
	By the delta method \citep[Theorem~3.1]{vdv98} and the asymptotic linearity of $\hat{\theta}_n$ in Eq.~\eqref{eq:psitheta}, we have
	\begin{align*}
		\sqrt{n} \bigl\{ F_{\hat{\theta}_n}^{-1}(u) - F_\theta^{-1}(u) \bigr\}
		&=
		\sqrt{n} ( \hat{\theta}_n - \theta )^\top \, 
		\nabla_\theta F_\theta^{-1}(u)
		+ \oh_{\prob_\theta^n}(1) \\
		&=
		\frac{1}{\sqrt{n}} \sum_{i=1}^n \psi_\theta(X_i)^\top \, \nabla_\theta F_\theta^{-1}(u) + \oh_{\prob_\theta^n}(1), \qquad n \to \infty.
	\end{align*}
	Further, by the functional delta method, the empirical quantile function satisfies
	\[
		\sqrt{n} \bigl\{ \hat{F}_n^{-1}(u) - F_\theta^{-1}(u) \}
		=
		- \frac{1}{\sqrt{n}} \sum_{i=1}^n 
		\frac{\1\{X_i \le F_\theta^{-1}(u)\} - u}{f_\theta(F_\theta^{-1}(u))} + \oh_{\prob_\theta^n}(1), \quad n \to \infty,
	\]
	see \citet[Corollary~21.5]{vdv98}. Subtract both expansions to get
	\begin{equation}
	\label{eq:zetanop1}
		\zeta_n(u) = - \frac{1}{\sqrt{n}} \sum_{i=1}^n g_\theta(X_i, u) + \oh_{\prob_\theta^n}(1), \qquad n \to \infty,
	\end{equation}
	where
	\[
		g_\theta(x, u) =
		\frac{\1\{x \le F_\theta^{-1}(u)\} - u}{f_\theta\bigl(F_\theta^{-1}(u)\bigr)}
		+ \psi_\theta(x)^\top \nabla_\theta F_\theta^{-1}(u).
	\]

	Differentiating the identity $u = \int_{-\infty}^{F_\theta^{-1}(u)} f_\theta(x) \, \diff x$ with respect to $\theta$ using Leibniz' integral rule yields, after some calculations, the identity
	%	\[
	%		0 = f_\theta\bigl( F_\theta^{-1}(u) \bigr) \nabla_\theta F_\theta^{-1}(u)
	%		+ \int_{-\infty}^{F_\theta^{-1}(u)} \nabla_\theta f_\theta(x) \, \diff x
	%	\]
	%	It follows that
	\begin{equation}
	\label{eq:nablaidentity}
	\nabla_\theta F_\theta^{-1}(u)
	= - \frac{1}{f_\theta\bigl( F_\theta^{-1}(u) \bigr)} 
	\expec_\theta[ \dot{\ell}_\theta(X) \, \1\{X \le F_\theta^{-1}(u)\}].
	\end{equation}
	The regularity property~\eqref{eq:AEL} of the influence function $\psi_\theta$, the centering property $\expec_\theta[\dot{\ell}_\theta(X)] = 0$ of the score function in Assumption~\ref{ass:boot:cons} and the identity~\eqref{eq:nablaidentity} for the gradient $\nabla_\theta F_\theta^{-1}(u)$ combine to imply that
	\begin{equation}
	\label{eq:crosscov0}
		\expec_\theta \left[
			\dot{\ell}_\theta(X) \, g_\theta(X, u)
		\right]
		= 0.
	\end{equation}
	
	Let $u_j \in (0, 1)$ for $j = 1, \ldots, m$.
	By the multivariate central limit theorem, the expansions~\eqref{eq:LAN} and~\eqref{eq:zetanop1} combine with Slutsky's lemma to yield the convergence in distribution of the sequence of $(m+1)$-dimensional random vectors
	\[
		\left( 
			\zeta_n(u_1), \ldots, \zeta_n(u_m),
			\log \prod_{i=1}^n \frac{f_{\theta_n}}{f_\theta}(X_i)
		\right)
	\]
	to a certain $(m+1)$-variate normal variable.
	Each of the first $m$ components of the limiting normal random vector is centred and, by Eq.~\eqref{eq:crosscov0}, uncorrelated with the $(m+1)$th one.
	By Le Cam's third lemma \citep[Example~6.7]{vdv98}, the limit distribution of $(\zeta_n(u_j))_{j=1}^m$ under $\theta_n$ is then the same as under $\theta$: an $m$-variate centred normal distribution with covariance matrix $\Gamma_\theta$ having elements
	\begin{equation}
	\label{eq:Gamma}
		\Gamma_\theta(j_1, j_2) =
		\expec_\theta \left[ g_\theta(X, u_{j_1}) \, g_\theta(X, u_{j_2}) \right],
		\qquad j_1, j_2 \in \{1, \ldots, m\}. \qedhere
	\end{equation}
\end{proof}

The above argument for the consistency of the parametric bootstrap in case $d=1$ was made possible by the representation in Eqs~\eqref{eq:Rn}--\eqref{eq:zetan} of the normalized test statistic $R_n$ in terms of the empirical quantile process $\zeta_n$.
This representation made it possible to find the invariant limit distribution of $R_n$ under contiguous alternatives $\theta_n = \theta + \Oh(1/\sqrt{n})$.
In case $d \ge 2$, however, no sufficiently explicit representations of the empirical Wasserstein distance are hitherto known to enable a similar analysis.
\egroup

% =============================================
\section{Some other GoF tests}
\label{sec:simu:simple:other}

We provide details about the tests appearing in the comparisons in Section~\ref{sec:simu:simple}.

\citet*{rippl2016limit} consider the fully specified Gaussian null hypothesis $\mathcal{H}_0^n : \prob = \normal_d(\mu_0, \Sigma_0)$   with given mean and covariance. Recall that the squared $2$-Wasserstein distance between two $d$-variate Gaussian distributions~is  
\[
W_2^2\bigl(\normal_d(\mu_1,\Sigma_1),\normal_d(\mu_2, \Sigma_2)\bigr)
= \norm{\mu_1 -\mu_2}^2 + \tr\bigl\{\Sigma_1 + \Sigma_2 - 2 (\Sigma_1^{1/2}\Sigma_2\Sigma_1^{1/2})^{1/2}\bigr\}.
\]
The Rippl--Munk--Sturm test statistic is $W_2^2\bigl(\normal_d(\overline{X}_n, S_{n,X}), \, \normal_d(\mu_0, \Sigma_0)\bigr)$, with~$\overline{X}_n$ and~$S_{n,X}$ the sample mean and sample covariance matrix, respectively. This test is sensitive to changes in the parameters of the Gaussian distribution but not to other types of alternatives. Calculation of the test statistic is straightforward. To compute critical values, we relied on a Monte Carlo approximation, drawing many samples from the Gaussian null distribution and taking the empirical quantiles of the resulting test statistics.

\citet{khmaladze2016unitary} constructs empirical processes in such a way that they are asymptotically distribution-free, which facilitates their use for hypothesis testing. A special case of the construction is as follows. Let the $d$-variate cumulative distribution function (cdf) $F$ be absolutely continuous with joint density $f$, marginal densities $f_1, \ldots, f_d$, and copula density $c$. Define 
\[
	l(x) = \left\{c\bigl(F_1(x_1), \ldots, F_d(x_d)\bigr)\right\}^{1/2},
	\qquad x \in \Rd,
\]
with $F_1, \ldots, F_d$ the marginal cdfs of $F$. The $d$-variate cdf $G(x) = \prod_{j=1}^d F_j(x_j)$ has the same margins as~$F$, but coupled via the independence copula. Letting 
\[
	\kappa(x) = \int_{(-\infty, x]} l(y) \, f(y) \, \diff y
	\quad\text{ and }\quad \kappa = \int l(y) \, f(y) \, \diff y,
\]
it follows from Corollary~4 in \citet{khmaladze2016unitary}   that the empirical process
\[
\tilde{v}_{F,n}(x) = 
\frac{1}{\sqrt{n}} \sum_{i=1}^n \bigl\{l(X_i) \, \1(X_i \le x) - \kappa(x)\bigr\} 
-
\frac{G(x) - \kappa(x)}{1 - \kappa} \frac{1}{\sqrt{n}} \sum_{i=1}^n \bigl\{l(X_i) - \kappa\bigr\}
\]
of an independent random sample $X_1, \ldots, X_n$ from $F$ converges weakly to\linebreak  a~$G$-Brownian bridge, i.e., the same weak limit of the ordinary empirical process
\[
	v_{G,n}(x) 
	= \frac{1}{\sqrt{n}} \sum_{i=1}^n \bigl\{ \1(Y_i \le x) - G(x) \bigr\}
\]
of an independent random sample $Y_1, \ldots, Y_n$ from $G$. The asymptotic distribution of a test statistic based on $\tilde{v}_{F,n}$ which is invariant with respect to coordinate-wise continuous monotone increasing transformations is thus the same as if $F$ (or $G$) were the uniform distribution on $[0, 1]^d$. This includes the Kolmogorov--Smirnov type statistic~$\sup_{x \in \Rd} |\tilde{v}_{F,n}(x)|$, which (with $F$ the cdf of~$\prob_0$) we consider in Section~\ref{sec:simu:simple} for comparison with the Wasserstein-based test. In case $F$ has independent margins,   $F$ and~$G$ coincide and the procedure reduces to a classical Kolmogorov--Smirnov test. To ensure that the test has the right finite-sample size, we calculate critical values by  Monte Carlo approximation rather than by relying on asymptotic theory. 
%but , simulating many samples from $\prob_0$ and using the empirical quantile of the resulting sample of test statistics. 

% =============================================
\section{Algorithms for the computation of critical values}
\label{app:algo}

%We propose the algorithms used to produce the Figures printed above. 

%
% \begin{algorithm}
% 
% \subsection{Mérigot's algorithm}
%\DontPrintSemicolon
%  
%  \KwInput{\begin{itemize}
%  \item A mesh that supports the source density $f$. 
% \item  The points $x_i$ and their weights (typically $1/n$). 
% \item A tolerance $\eta$
%  \item A number of scales  $L$
%\end{itemize}}
%  \KwOutput{The unique Laguerre diagram consisting of the cells $V_\psi(j)$ such that 
%  \[
%  \int_{V_\psi(j)} f(x)dx  = 1/n, \quad \forall j.
%  \]
%}
%\hrule
% $\psi \leftarrow [0\cdots0]$ \tcp*{Initialization }
%  %  \tcc{Now this is a While loop}
%   \For{$l  = L-1$ to $0$}
%   {
%   set  $\psi_{l,0}(p) := \psi_{l+1} (\pi_l(p))$ for every $p\in S_l$\\
%   k:=0\\
%   \Repeat {$\max_j\lvert\prob_0(V_\psi(j))-1/n\rvert > \eta$}
%   {
%   		Compute $\psi_{l,k+1}$ from $\psi_{l,k}$ using L-BFGS on $F_l$  \tcp*{L-BFGS is a hessian-free optimizer}
%set $\nu_{k+1} := \nabla F_l(\psi_{l,k+1})$, $k:=k+1$\\}
%  $\psi_{l}:= \psi_{l,k}$
%   }
%\caption{\label{alg: SemDis}Multiscale algorithm for semi-discrete optimal transport}
%\end{algorithm}
%
% The latter is fast and implemented in the R `transport' package.  
% 
% 
% \subsection{Computations of critical thresholds}

Our test statistics involve the Wasserstein distance between an empirical measure and a continuous one. Calculating such a distance requires solving a semi-discrete optimal transport problem (Section~\ref{sec: computations}).

In dimension $d = 2$ and for the Wasserstein distance of order $p = 2$, we relied on the function \textsf{semidiscrete} in the \textsf{R} package \textsf{transport} \citep{transport}, which implements the method of \citet{merigot2011multiscale}. The method starts from a discretization of the source density. The quality of approximation can be set by choosing a sufficiently fine mesh and selecting the tolerance parameter to a low value. The meshes considered here are providing approximately $10^5$ cells.

For the simulations involving the Wasserstein distance of order $p = 1$ or in dimension $d$ larger than two, we resort to our own implementation of the stochastic average gradient algorithm as employed in \citet{genevay2016stochastic}. 
The number of random points chosen for the reference measure was $2\times 10^5$ which corresponds to a thousand times the sample sizes considered in the various simulation settings. The $C$ parameter appearing in their algorithm was set to~$1$. 

The test statistics in \eqref{eq:Tn} and \eqref{eq:Tgroup} involve a fixed continuous measure $\prob_0$ or $\probQ_0$, respectively, but the ones in~\eqref{eq:TMntheta} and~\eqref{eq:TMn:eta} concern a continuous measure with estimated parameter $\hat{\theta}_n$ or $\hat{\psi}_n$. 
Moreover, to calculate critical values with the parametric bootstrap, even a single execution of the test requires a large number of evaluations of the test statistic at random parameter values. 
To speed up the calculations, we perform the following two steps prior to observing the data:
\begin{enumerate}
	\item We compute the discretizations of the target density mentioned above for each value of the unknown parameter in a large but finite subset of the parameter space. We then force the Monte Carlo replications of the parameter estimates to take values in that subset. In this way, we do not need to recompute the discretization of the target density each time.
	\item We compute the critical values at a finite subset of the parameter space, by drawing random samples of the test statistic for each value of the parameter in that finite subset and applying the reduction of Step~1. Then we learn the critical value as a function of the (continuous) parameter by smoothing. See Figure~\ref{fig:c_theta} for an illustration.
\end{enumerate}

\begin{figure}
\begin{center}
\includegraphics[width=100mm, height=60mm]{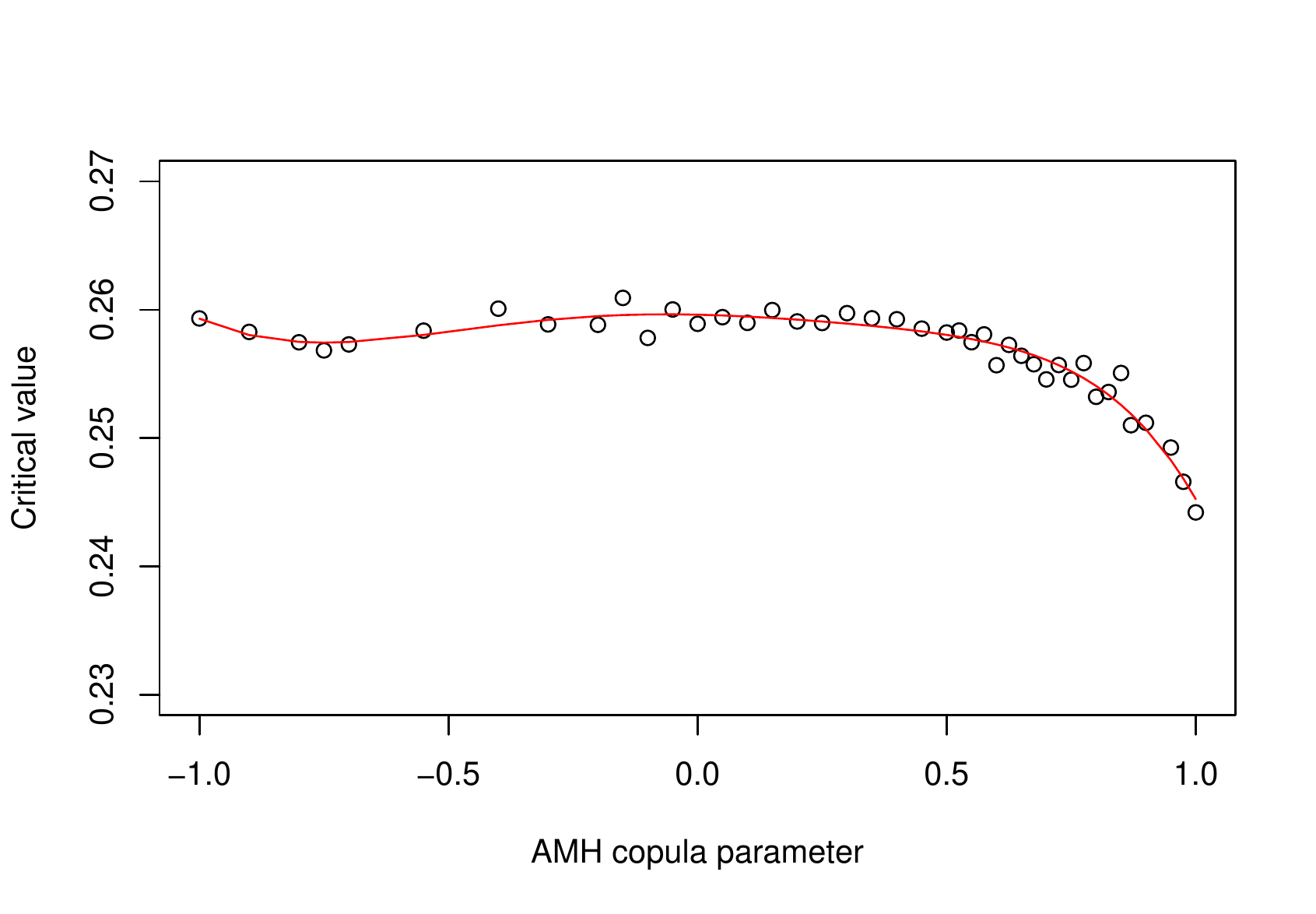}
\end{center}
\caption{\label{fig:c_theta} \slshape\small Illustration of Step~2 in Appendix~\ref{app:algo} for learning the critical value function of the $2$-Wasserstein GoF test for the bivariate five-parameter Gaussian--AMH model in Section~\ref{sec:simu:param:GaussAMH} using the location--scale reduction in Section~\ref{sec:param:group}. The function $\psi \mapsto c_{\model}(\alpha, n, \psi)$ (in red) is constructed by smoothing Monte Carlo estimates (circles) of $c_{\model}(\alpha, n, \psi)$ for $\psi \in \Psi' \subseteq \Psi = [-1, 1]$, with $\alpha = 0.05$, $n = 200$ and $B = 1\,000$ samples per point. The smoother is a 6th-degree polynomial fitted by ordinary least squares.}
\end{figure}

\section{A boomerang-shaped distribution}
\label{sec: bana}

The ``boomerang-shaped'' distribution in Section~\ref{sec:simu:simple} and Figure~\ref{fig: GausNull}(f) is a mixture 
\begin{align}
\nonumber
(1-2p)%0.3
 \, \normal_2\left(
\begin{pmatrix} \phantom{-}0\phantom{.7} \\ -0.7 \end{pmatrix}, 
 \begin{pmatrix} 
 0.35^2&0 \\
0& 0.35^2 
 \end{pmatrix} 
  \right)&+ 
p % 0.35
 \, \normal_2\left(
\begin{pmatrix} -0.9 \\ \phantom{-}0.3 \end{pmatrix}, 
 \begin{pmatrix} 
 \phantom{-}0.358&-0.55 \\
-0.55\phantom{0}& \phantom{-}1.02 
 \end{pmatrix} 
  \right)\\
\label{eq:banana}
  &+ 
p%  0.35
   \, \normal_2\left(
\begin{pmatrix} \phantom{-}0.9 \\ \phantom{-}0.3 \end{pmatrix}, 
 \begin{pmatrix} 
 \phantom{-}0.358& \phantom{-}0.55 \\
\phantom{-}0.55\phantom{0}& \phantom{-}1.02 
 \end{pmatrix}
  \right).
\end{align}
of three Gaussian components. 
Figure~\ref{fig:banana} shows a scatterplot for $p=0.35$ of a random sample of size $n=500$ from this distribution.

\begin{figure}
\begin{center}
\includegraphics[width=100mm, height=60mm]{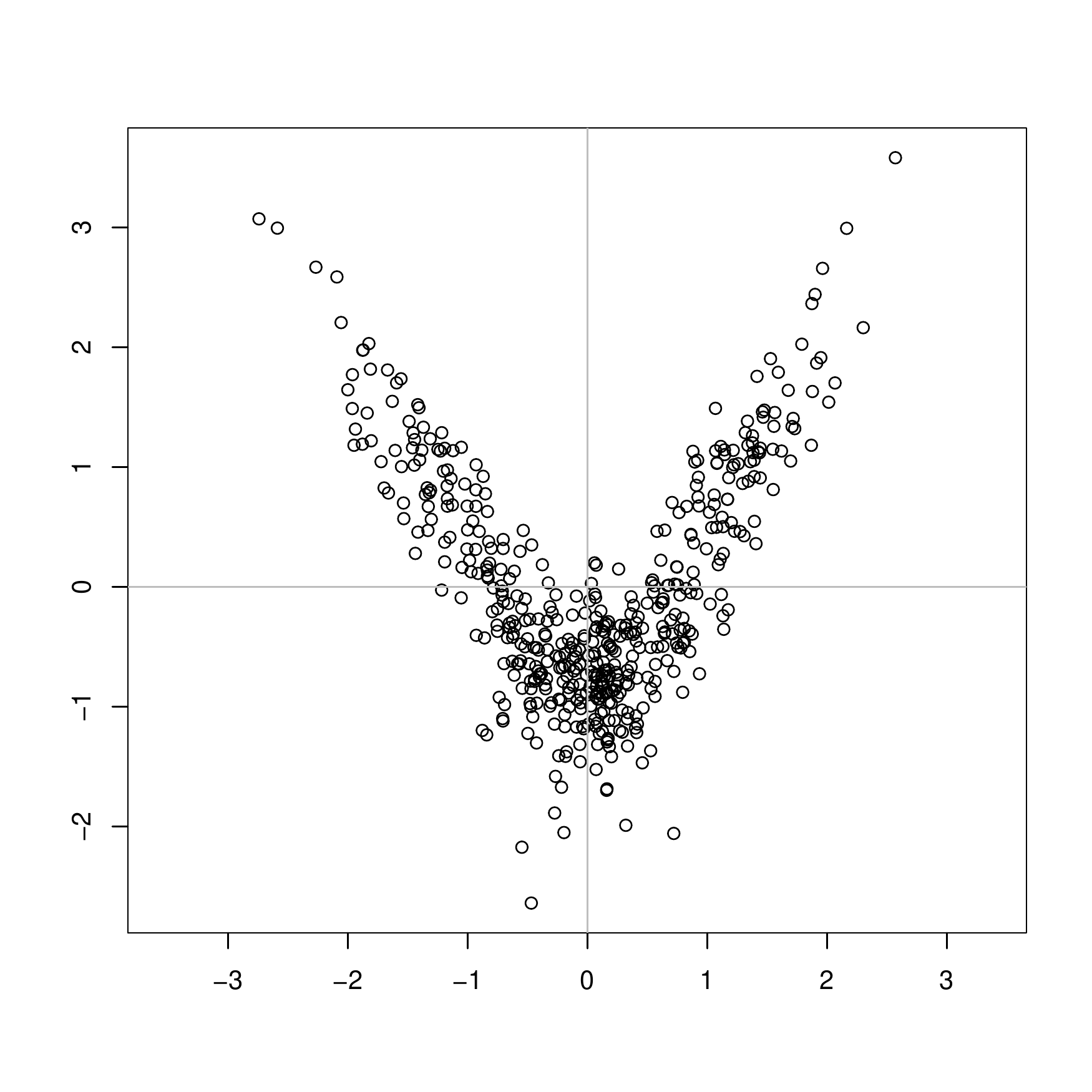}
\end{center}
\caption{\label{fig:banana} \slshape\small Scatterplot of a sample of size $500$ from the ``boomerang-shaped'' mixture~\eqref{eq:banana}.}
\end{figure}

\end{document}